\documentclass[12pt]{iopart} 
\usepackage{comment}
\usepackage{graphicx}
\usepackage{amssymb}
\usepackage{cite}
\newcommand{\zlabel}[1]{\label{#1} }
\newcommand{\fc}{\frac} 

\newcommand{\mrr}{\mathbf{r}}
\newcommand{\mq}{\mathbf{q}}

\newcommand{\pr}{\prime}

\newcommand{\mss}{\mathbf{s}}
\newcommand{\mx}{\mathbf{x}}
\newcommand{\my}{\mathbf{y}}
\newcommand{\mrrp}{\mathbf{r}^\prime}

\newcommand{\al}{\alpha}
\newcommand{\be}{\beta}
\newcommand{\pa}[2]{\frac{\partial #1}{\partial #2}}

\newcommand{\bv}{\mathbf{v}}
\newcommand{\bu}{\mathbf{u}}
\newcommand{\bw}{\mathbf{w}}
\newcommand{\bp}{\mathbf{p}}
\newcommand{\dive}{\text{\bf div}\;\!}
\newcommand{\grad}{\text{\bf grad}\;\!}

\newcommand{\lab}[1]{ (\ref{#1}):\; }

\newcommand{\veps}{\varepsilon}
\newcommand{\sPsi}{\mbox{\small $\Psi$}}
\newcommand{\defi}{\,\,\dot{=}\,\,}
\newcommand{\defim}{\,\,\dot{=}\,}

\newcommand{\text}{\mbox}
\newcommand{\boldsymbol}{\mathbf}
\newcommand{\alt}[1]{}

\newcommand{\zero}{\mathbf{0}}
\newcommand{\imply}{\;\Rightarrow\;}
\newcommand{\Pu}{P}\newcommand{\Pv}{p}
\newcommand{\ES}{E}\newcommand{\Etheta}{F}
\newcommand{\bS}{\textbf{\textit{S}}}

\newcommand{\hs}[1]{\hspace{#1}}

\newcommand{\spa}{\hspace{0.1ex}}
\newcommand{\intp}{\hs{0.8ex}\hat{\hs{-1.5ex}\int}_{\hs{-1ex}1+}}
\newcommand{\intf}[1]{\hs{1.5ex}\hat{\hs{-1.5ex}\int}_{\hs{-1ex}#1}\hs{1ex}\!}

\newcommand{\mmu}{\text{{\boldmath $\mu$}}}
\newcommand{\oomega}{\text{{\boldmath $\omega$}}}

\newcommand{\sch}{Schr\"odinger }
\newcommand{\Ve}{V_{\mbox{\small e}}}
\newcommand{\coulf}[2]{\mbox{\Large$/$} \overrightarrow{|#1 \mbox{\tiny $\!\;\setminus\!\;$} #2|}\mbox{} ^{2}}
\newcommand{\coulp}[2]{\mbox{\large $/$}|#1 \mbox{\tiny $\!\;\setminus\!\;$} #2|\mbox{}}
\newcommand{\qcharge}{\text{\small $Q$}}

\newcommand{\room}{\hspace{0.15ex}}
\newcommand{\rsb}[1]{\raisebox{0ex}{\fbox{#1}}}
\newtheorem{lemma}{Lemma}
\newtheorem{theorem}{Theorem}
\newtheorem{corollary}{Corollary}
\newtheorem{definition}{Definition}
\newtheorem{proof}{Proof}
\newtheorem{hypothesis}{Hypothesis}

\begin{document} 

\title[A Formulation of Quantum Fluid Mechanics and Trajectories]{\hs{-3ex}\mbox{A Formulation of Quantum Fluid Mechanics and Trajectories}}


\author{James P. Finley: }

\address{Department of Physical Sciences,
Eastern New~Mexico University,
Portales, NM 88130}
\ead{james.finley@enmu.edu}


\date{\today}


\begin{abstract}

  
A formalism of classical mechanics is given for time-dependent many-body states of quantum
mechanics, describing both fluid flow and point mass trajectories.  The familiar equations of
energy, motion, and those of Lagrangian mechanics are obtained.  An energy and continuity
equation is demonstrated to be equivalent to the real and imaginary parts of the time dependent
\sch equation, respectively, where the \sch equation is in density matrix form.  For certain
stationary states, using Lagrangian mechanics and a Hamiltonian function for quantum mechanics,
equations for point-mass trajectories are obtained.  For 1-body states and fluid flows, the
energy equation and equations of motion are the Bernoulli and Euler equations of fluid
mechanics, respectively.  Generalizations of the energy and Euler equations are derived to
obtain equations that are in the same form as they are in classical mechanics.  The fluid flow
type is compressible, inviscid, irrotational, with the nonclassical element of \emph{local}
variable mass. Over all space mass is conserved. The variable mass is a necessary condition for
the fluid flow to agree with the zero orbital angular momentum for s states of hydrogen. Cross
flows are examined, where velocity directions are changed without changing the kinetic energy.
For one-electron atoms, the velocity modification gives closed orbits for trajectories, and
mass conservation, vortexes, and density stratification for fluid flows.  For many body states,
Under certain conditions, and by hypotheses, Euler equations of orbital-flows are
obtained. One-body \sch equations that are a generalization of the Hartree-Fock equations are
also obtained. These equations contain a quantum Coulomb's law, involving the 2-body pair
function of reduced density matrix theory that replace the charge densities.
\end{abstract}



\section{Introduction}

A formalism of classical mechanics is presented for time-dependent many-body states of quantum
mechanics, describing both fluid flows and point mass trajectories. An energy and continuity
equation is demonstrated to be equivalent to the real and imaginary parts of the time dependent
\sch equation, respectively, where the \sch equation is in density matrix form.  For a class at
stationary states, the energy equation implies a Hamiltonian that is used in Lagrangian
mechanics to obtain equations of motion.

For fluid flow of 1-body states, the equation of motion is the Euler equation, and the energy
equation is the Bernoulli equation. The flow type is compressible, inviscid, irrotational, with
the nonclassical element of \emph{local} variable mass. Over all space mass is conserved.  For
many body states, under certain conditions, and by hypotheses, Euler equations of orbital-flows
are obtained. Under additional conditions, 1-body \sch equations that are a generalization of
the Hartree-Fock equations are also obtained. These equations contain a quantum Coulomb's law,
involving the 2-body pair function of reduced density matrix theory, replacing charge
densities. In order to obtain local conservation of mass for fluids flows and trajectories
with continuous velocity vectors, cross flows are examined, where velocity directions are
changed without changing the kinetic energy. These cross flows provide vortexes for
one-electron atoms and closed orbits for trajectories.  Conditions are given for continuity
satisfaction and incompressible fluid flow, giving density stratification, where the mass
density is only constant of each streamline.  As in the relationship between Madelung and
Bohminan mechanics for one-body states, the streamlines of the fluid flow are also the paths of
point mass trajectories. Because of the this close relationship, we move freely back and forth
between the two descriptions.

Needed wavefunction identities are obtained Sec.~\ref{2890}.  Sec.~\ref{4193} defines variables
for quantum states that correspond to linear self-adjoint operators, and explicit forms of
these variables are given with the wavefunction in polar form.  Sec.~\ref{8920} uncovers much
mathematical structure, with many relationships, involving the terms that originate in the
imaginary part of the \sch equation in density matrix form.

The energy equation of Sec.~\ref{4491} is a combination and generalization of the energy
equations of the Bernoulli equation of fluid mechanics, used to describes 1-body states with a
real valued wavefunctions \cite{Finley-Bern}, and the one from Bohmian Mechanics
\cite{Bohm:52a,Bohm:52b,Takabayasi,B2,B4,B5,B6,B7,B8,B9,Jung,Renziehausenb}, used to describe
states with complex value wavefunctions.  The energy equation can be represented with a certain
pressure, for fluids, or a certain body potential, for trajectories.  There are two linearly
independent velocity components: one velocity component of Bohmian mechanics, denoted $\bv$,
and one from the Bernoullian formalism mentioned above, denoted $\bu$.  Because the energy
equation has a nonclassical feature of kinetic energy when $\bu\cdot\bv$ does not vanish, and
we wish to remove such elements, a generalized energy equation is derived without this feature.

In Sec.~\ref{2283}, for 1-body states, an Euler equation of fluid dynamics and a continuity
equation is proved to be equivalent to the time dependent \sch equation. Because terms of the
momentum time-derivative of classical fluids are missing in the Euler equation, a range of
generalizations are obtained. The most generalized one has all classical terms, and the most
number of additional forces, including bulk viscosity. The additional forces involve
coupling between the two velocity field components.  Interpretations of the various
equations, and the identifications of variables are given in Secs.~\ref{2215} and \ref{2217}.
Sec.~\ref{4938} covers cross flows.  Sec.~\ref{9918} provides Lagrangian mechanics for quantum
states.  Sec.~\ref{4918} obtains orbital flows, and mixed flows, involving a single 1-body
wavefunction describing all bodies. These equations contain a quantum Coulomb's law, involving
the 2-body pair function of reduced density matrix theory that replace the charge densities.

There are a number of overlapping formalism that go by various names for the description
of quantum states: the De Broglie--Bohm theory
\ \cite{Bohm:52a,Bohm:52b,Takabayasi,B2,B4,B5,B6,B7,B8,B9,Jung,Renziehausenb}, also called
Bohmian mechanics and quantum hydrodynamics, and Madelung fluid mechanics
\cite{Madelung:26,Madelung:27}. All of these methods have the same Madelung velocity
$\bv$, and, except for the earliest ones, explicitly contain the quantum
potential $Q$ \cite{Bohm:52a,Bohm:52b,B6}. For many body states, the theories based on a
Madelung fluid differ from the Bohmian particle trajectories \cite{Renziehausenb}.
Heifetz and coworkers \cite{Heifetz,Heifetz2} explores the thermodynamics of Madelung fluids.
There are many modifications and generalizations of the Madelung equations
\cite{Sorokin,Broadbridge,Schonberg,Caliari,Jamali,Waegell}.
The generalization by Broadbridge \cite{Broadbridge} and Jamali \cite{Jamali} use a complex velocity.
Tsekov \cite{Tsekov} also uses a complex velocity to derives a complex Navier--Stokes equation.
Vadasz \cite{Vadasz} derived an extension of the Schr\"odinger equation from the Navier--Stokes
equation. Quantum hydrodynamic theory has been employed to treat systems with single particle wave functions
\cite{6,7,8,9,10,11,12,14,15,16,17,18,19,20,21,22}.
The method also also been generalized to treat many particle systems
\cite{8b,Renziehausen}. Application of this formalism include the investigation of spin effects
\cite{14b,15b}, Bose--Einstein condensates \cite{16b}, graphene \cite{17b} and plasmas
\cite{18b,22b,23b}.

The quantum potential of Bohmian mechanics suffers from kinetic energy contamination.  Because
of the missing kinetic energy, Bohmian mechanics has a difficulty in explaining the
nonvanishing expectation value of kinetic energy for states with real valued wavefunction,
given that these states are described as being static. Bohmian mechanics theory provides no
information about these states, nor explain what the Bohmian potential does for static
states. It is possible that spin angular momentum might be accounted for by an orbiting
electron, and there is a formula for the electron motion derived by Holland
\cite{Holland:99,Colin1,Colin2}. However, the kinetic energy responsible for the creation of
the orbital momentum must be present in nonrelativistic quantum mechanics, and the only place
it could be is in the Bohm potential.  Bohmian mechanics also cannot explain the $x$ and $y$
portions of orbital angular momentum. Furthermore, the Bohm potential, besides being difficult
to interpret, gives a stress tensor for force for equations of motion \cite{B6,Renziehausen},
even for stationary states, where the forces should be conservative and be defined by a
potential.  The addition of the velocity vector component $\bu$ of the Bernoullian formulation,
and the elimination of the Bohm potential, removes these problems \cite{Finley-Bern}.

Adding the velocity vector $\bu$ and a pressure $\Pu$ has other advantages \cite{Finley-Bern}.
The force from the pressure removes the necessity of tunneling, so that nothing special happens
when a particle moves into a region that is, otherwise, classically forbidden.  The velocity
vector $\bu$ explains the zero angular momentum for the $s$ electronic states of the hydrogen
atom, since the velocity field is in the radial direction.  It is also easily demonstrated that the
pressure force $-\nabla\Pu$ helps stabilize one-electron atoms, giving, for the hydrogen 1s
state, a strong repulsive force close to the nucleus and a weak attractive force beyond a
distance of about 1.37 a.u. For energy, a portion of the pressure cancels the energy from the
proton nucleus, giving a finite energy at the nucleus. Equations for the speed of sound have
also been obtained, where the extrema of the momentum per volume occur at points that are Mach
1 speed \cite{Finley-Arxiv}.

\section{Runner Notation System \zlabel{3428}} 

In this section, a notational system is introduced, applicable to the sections that follow,
where equations for the general $n$-body case have the same appearance as in the special case
where $(n=1)$.  This notation permits the omission of the summation sign along with the
corresponding subscripts that run over an index set. Consider the two equations:
$(-\nabla^2\Psi/2 + U\Psi = \ES)$ and $(K \defim -\!\nabla^2\Psi/2)$. It is understood, by
the rules given below, that the first equation is a single equation, such that $\nabla^2$ denotes
the series $(\nabla^2_1 + \cdots + \nabla^2_n)$. The second one is a sequence of $n$
definitions, the $i$th being $(K_i = -\nabla^2\Psi_i\div 2)$.

\begin{definition} 
The core family $(\mathbf{\Psi}\defi \{\text{Re}\Psi,\text{Im}\Psi,\rho,S\}$) is a set of
real-valued functions from the Cartesian product
$\bS\times\left(\mathbb{R}\times\mathbb{R}^{3n}\right)$. The spin space is defined by
$\bS\defi\{(\omega_1,\omega_2,\cdots, \omega_n)\vert\omega_i\in\{-1,+1\}\}$, and
$\mathbb{R}\times\mathbb{R}^{3n}$ is a Euclidean space. The members of $\mathbf{\Psi}$ are two
times continuously differentiable almost everywhere on the space
$\mathbb{R}\times\mathbb{R}^{3n}$ for all $\mathbf{s}\in\bS$.  We write
$\hat{\mathbb{R}}\times\mathbb{R}^{3n}$ to emphasize that the variable of $\mathbb{R}$
is time. The elements of $\mathbf{\Psi}$ are related via the polar form of the
wavefunction: $(\Psi = \sqrt{\rho} e^{i S/\hbar})$.
\end{definition}
If two functions are equal almost everywhere (a.e.), they are considered physically
indistinguishable. The symbol ``$=$'' means equal a.e. The same idea is used for equivalence of
equations: A function satisfies one of the equations a.e, if and only if it satisfies the other
one a.e. Hence, there will be on reason to write ``a.e,'' since the condition is worked into
the formalism.  For the most part, we do not use, nor do we need, the special structure from the
$L^3(\mathbb{R}^{3n})$ Hilbert space, that is not already in the underlining vector space.
 
\begin{definition}\mbox{\rm (Runners)}
Let $(\mathbf{\Sigma},+)$ be a vector space of $n$-tuples fitted with the sequence to series map
$(X)\mapsto [X]$, where $(X) = (X_1,X_2,\cdots X_n) \in\mathbf{\Sigma}$ and $([X] \defi X_1 +
X_2 \cdots + X_n)$. The runner notation system uses special shorthand notation: If
$(X),(Y)\in~\mathbf{\Sigma}$, then the $n$-equations
%
\[X_i  = Y_i, \qquad i = 1,\cdots,n,\]
are denoted by $(X=Y)$. However, if $U\notin\mathbf{\Sigma}$ appears as a term in an equation
with one or more elements from $\mathbf{\Sigma}$, then, for example, $(U = X + Y)$ denotes the
series equality: $(\room U = [X] + [Y]\room)$, that is
\[U = X_1 + X_2 + \cdots X_n  + Y_1 + Y_2 + \cdots Y_n. \]
The elements $(X)$ and $(Y)$ are said to run in the equation series $(X =Y +U)$ and step in the
equation sequence $(X=Y)$.  The members of $(\mathbf{\Sigma},+)$ are called runners.

Let $\rho,S\in\mathbf{\Psi}$ and $U\notin\mathbf{\Sigma}$.  For runners $(\nabla\rho)$,
$(\nabla S)\in\Sigma$ of an $n$-body state, here are two examples of the notation:
\[\hspace{-3ex}
\nabla^2Y \defi \nabla\rho\cdot\nabla S\quad \text{means} \quad
\nabla_i^2Y = \nabla_i\rho\cdot\nabla_i S, \hspace{2ex} i\in\{1,2,\cdots n\},
\]
\[\hspace{-8ex}
\nabla^2Y = \nabla\rho\cdot\nabla S + U \quad \text{means} \quad
\sum_{i=1}^n\nabla_i^2Y = \sum_{i=1}^n\nabla_i\rho\cdot\nabla_i S + U,
\]
where it is understood that the collection $\mathbf{\Sigma}$ is extended to include $Y$, and
the symbol ``$\defi$'' is used for definitions.  If $f\in\mathbf{\Psi}$ and $(\hat{O} =
\hat{O}_1+\cdots \hat{O}_n)$ is a differential operator, where $\hat{O}_i$ acts on the $i$th
component of $\mathbb{R}^{3n}$, then we require that
$(\room(\hat{O}f)\in\mathbf{\Sigma}\room)$.

Other notation used includes the following items. 1) A semicolon in an equation means ``such that'' or
``where.''  For example, ``$(A > B^{-1};$ $B\ne 0)$'' means ``$A$ is greater than $B$ such that
$(B\ne 0)$.''  2) ``Eq.~($N$)$\times c$'' means Eq.~($N$) that is multiplied by $c$. 3) The
symbol $\partial$ is the partial time derivative operator, that is
$(\hs{0.25ex}[\partial\Psi](t) \defi \partial\Psi/\partial t\hs{0.15ex})$. 4) The following
notation is used for spin summations and spatial integrations:
\begin{equation*}
\hs{-12.5ex}\sum_m \text{ and  } \int_m \text{ are the sum and integration over the $m$th component of elements} 
\end{equation*}
from $\bS$ and $\mathbb{R}^{3n}$, respectively, For example, for $(n=3)$ and $(m=2)$, we have
\hspace{1ex}
\[ \hs{-12ex}
\left[\int_2\sum_2\Psi\right](\mrr_1,\omega_1;\mrr_3,\omega_3) =
\int \sum_{\omega_2\in\{-1,+1\}}\left[\Psi(\mrr_1,\omega_1;\mrr_2,\omega_2;\mrr_3,\omega_3) \right] \, d\mrr_2. \]
The following notation is used:
\begin{equation*}\hs{-10ex}
  \int_{m+} = \int_{m+1}\cdots\int_n,
  \quad \sum_{m+} = \sum_{m+1}\cdots\sum_{n},
\quad \intf{m} = \sum_m\int_m, \quad \intf{m+} = \sum_{m+}\int_{m+}, 
\end{equation*}
\begin{equation}\label{spinsum2}
 \lefteqn{\hs{-7ex}\text{and}\hs{7ex}}\quad \intf{} = \intf{1}\cdots \intf{n}.
\end{equation}
\end{definition}

It is common in mathematical proofs to prove that a statement, involving a given function, is true for
only a value $f(x)$ of the function $f$.  It is then pointed out that, since the point $x$ is
arbitrary, the proof holds for the function $f$. We use this approach for spin, where the spin
variables are treated as parameters.
For density matrices, this differs from the usual convention, where spin is an explicit
variable or the spin variables are summed over, giving spinless density matrices.

\section{Correspondence Variables \zlabel{4930}} 



In quantum mechanics, physical properties are replaced by observables. Linear self-adjoint
operators represent observables that,
together with rules, provide predictions for the statistical outcome of a sufficiently
large collection of measurements of an observable. There are special probabilistic predictions
given for the order of measurements and the simultaneous measurements of related pairs of
observables, position and momentum being an example. There is a special axiom declaring the
physical meaning of the probability distribution ($\rho\defi \Psi^*\Psi$), called the Born
rule.

In this paper, variables are defined that correspond to observables of quantum mechanics. These
correspondence variables are needed to give a description of quantum states based on classical
mechanics. These variables are generalized real-valued fields, with domain
$\hat{\mathbb{R}}\times\mathbb{R}^{3n}$, that are defined using the wavefunction $\Psi$ and
operators of quantum mechanics.  This approach yields two variables in some cases, one from the
real part, and one from the imaginary part, of a defining complex valued-function.
There are two classes of correspondence variables: \emph{particle and fluid variables}. The two
classes being related by the map: \text{\rm (}$Y\mapsto Y\rho$\text{\rm ) }, where $Y$ is a
particle variable, $Y\rho$ is a fluid variable, and \text{\rm (}$\rho \defi
\Psi^*\Psi$\text{\rm )} is the probability distribution. These definitions are given to avoid
saying ``per mass'' and ``per volume.''

It is a requirement of a wavefunction for all times that the subset of $\mathbb{R}^{3n}$ where
the probability distribution $\rho$ vanish has measure zero, and that a wavefunction and its
first derivatives vanish at infinity. These conditions can be satisfied by requiring the
external potentials to be members of the $(L^{3/2} + L^\infty)$ Hilbert space \cite{Lieb}.
Fluid variables have the same domain as the wavefunction they are determined by.  For particle
variables, the domain is the support of $\Psi^*\Psi$, defined as the region of
$\mathbb{R}^{3n}$ space where $\Psi^*\Psi$ does not vanish, for a given point of time. 

Next we define the classes of correspondence variables.  Let $\hat{L}$ be an operator that
represents an observable that is given by a sum over 1-body operators. For example, $(-L =
i\hbar \nabla_1 + \cdots + i\hbar \nabla_n)$.\vspace{1ex}

\noindent
    {\bf 1a)} \mbox{\rm(\mbox{\it \bf Weak variable from operator})}
Let $(\room(\Psi^*\hat{L}\Psi)\in\mathbf{\Sigma})$.  A weak fluid variable $(\room
(Y)\in\mathbf{\Sigma}\room)$, corresponding to operator $\hat{L}$, is a real-valued generalized
field, with domain $\mathbb{R}^{3n}$, that satisfies the following equation sequence:
\begin{equation} \label{matel} 
  \left\langle\,\hat{L}\,\right\rangle \defi \intf{}\, \Psi^*\hat{L}\Psi = \intf{}\;Y, 
\end{equation}
where $\left\langle\,\hat{L}\,\right\rangle$ is the expectation value of an observable that
corresponds to $\hat{L}$, and $\hs{1ex}\hat{\hs{-1ex}\int}$ is the $n$-body spin-sum,
spatial-integration operator~(\ref{spinsum2}). 
\vspace{1ex}

\noindent 
{\bf 1b)} \mbox{\rm(\mbox{\it \bf Strong variable by operator})} The two strong fluid
variables, corresponding to linear self adjoint operator $\hat{L}$, are the real and imaginary
parts of the function \mbox{\rm $(\Psi^*\hat{L}\Psi)$}.  The strong particle variables are the
real or imaginary parts of $(\room(\Psi^*L\Psi/\rho)\hspace{0.2ex)} \in\mathbf{\Sigma}\room)$.
It follows from the definitions that the strong fluid variable \mbox{\rm
  ($\text{Re}\Psi^*\hat{L}\Psi$)} is weak, since the expectation value of a Hermitian operator
must be real.  The strong--particle variables by composition are defined by the real or
imaginary parts of the composite $(\room(g\circ(\Psi^*L\Psi/\rho)\hspace{0.2ex)}
\in\mathbf{\Sigma}\room)$, involving a given map $g$ that acts on each component of the
sequence $(\Psi^*\hat{L}\Psi/\rho)$.  We also take the external potential $U$ and
$(i\hbar\Psi^*\partial\Psi/\rho)$ to be strong particle variables of potential and total energy.
These variables are not runners. 

\vspace{1ex}

\noindent
{\bf 1c)} \mbox{\rm(\mbox{\it \bf First order reduced variables by operator})} Let $f$ and $g$
be the strong particle variables of $\hat{L}$:
\begin{equation}\zlabel{7720}
f\rho + ig\rho \defi \Psi^*\hat{L} \Psi;\qquad f,g\in\mathbf{\Sigma}.
\end{equation}
The first order, reduced, strong
particle-variables, $\hat{f}$ and $\hat{g}$, corresponding to operator $\hat{L}$, are
\begin{equation} \zlabel{7722} 
  \hat{f}\hat{\rho} + i\hat{g}\hat{\rho} \defi
  \intf{1+}\hs{-1ex}\Psi^*\hat{L}_1 \Psi,
\end{equation}
where $\hs{1ex}\hat{\hs{-1ex}\int}_{1+}$ is the spin-sum, spatial-integration operation~(\ref{spinsum2})
with $(m = 1)$; the probability density $\hat{\rho}$, defined by
\begin{equation} \zlabel{2914}
\hat{\rho} \defi \intf{1+}\rho; \qquad \rho = \Psi^*\Psi, \quad \intf{1}\hat{\rho} = 1, 
\end{equation}
is normalized to unity, instead of the usual $n$ normalization for a wavefunction normalized to
unity. Defs.~(\ref{7720}) an (\ref{7722}) give 
\begin{equation} \label{7724}
  \hat{f}\hat{\rho} = \intf{1+} f_1\rho \quad\text{and}\quad
  \hat{g}\hat{\rho} = \intf{1+} g_1\rho. 
\end{equation}

\vspace{1ex}

\noindent 
{\bf 2)}\mbox{\rm(\mbox{\it \bf Variable from equation})} For the special case of 1-body
\text{\rm ($n=1$)}, the field is a variable of an equation of classical mechanics, and the
equation is equivalent to, or implied by, a 1-body Schr\"odinger equation.  

\vspace{1ex}

\noindent 
{\bf 3)} \mbox{\rm(\mbox{\it \bf A free, or distributive, variable})} is a fluid variable whose
integration over all space $\mathbb{R}^{3n}$ vanish at all times.  If $(Y = W\rho)$ and $(Y)$
is free, $(W)$ is a weak particle variable. Note that if $(Y)$ is a weak fluid variable, such
that
\begin{equation} \zlabel{7922}
Y = Z + X,
\end{equation}
where $X$ is free, then $(Z)$ is also a weak variable of operator.  The strong fluid variable
\mbox{\rm ($\text{Im}\Psi^*\hat{L}\Psi$)} is free, since the expectation value of a
Hermitian operator must be real.

\section{Wavefunction Identities \zlabel{2890}}

The lemmas in this section are needed to obtain formulae for certain correspondence variables.
Note that displayed equations in Lemmas~\ref{Lap1} and \ref{Lap2} below, are sequences of
$n$ equations.
Let
\[\hspace{-3ex}(\Psi\nabla^2 \Psi), \Psi \nabla\Psi, \vert\nabla\Psi\vert^2,
\nabla\cdot\left(\Psi^*\nabla\Psi\right)\in\mathbf{\Sigma},\quad\text{and}\quad
    \hat{\text{\rm P}}_i \defim -i\hbar\nabla_i. \]
\begin{lemma}[Momentum Operator Form of the Laplacian] \label{Lap1} \mbox{}
%
\begin{equation} \label{laplacian2b}
    \Psi^*\left(-\hbar^2\nabla^2 \Psi\right)
    =  \vert\hat{\text{\rm P}}\Psi\vert^2 - i\hbar\nabla\cdot\left(\Psi^*\hat{\text{\rm P}}\Psi\right),
    \quad \hat{\text{\rm P}} \defim -i\hbar\nabla.
\end{equation}
%
\begin{equation} \label{momen}
\hs{5ex}\vert\hat{\text{\rm P}}\Psi\vert^2\rho = \vert\Psi^*\hat{\text{\rm P}}\Psi\vert^2.
\end{equation}
\vspace{-0.5ex}
\begin{equation} \label{laplacian2} 
  \Psi^*\left(-\fc{\hbar^2}{2m} \nabla^2 \Psi\right)\rho
  = \fc{1}{2m}\vert\Psi^*\hat{\text{\rm P}}\Psi\vert^2
   - i\fc{\hbar}{2m}\nabla\cdot\left(\Psi^*\hat{\text{\rm P}}\Psi\right)\rho.
\end{equation}
\end{lemma}


\textbf{\textit{Proof}}
Substituting the following two identities into the third one gives Eq.~(\ref{laplacian2b}):
\[\hspace{-2ex} \vert\hat{\text{P}}\Psi\vert^2 =
 (-i\hbar\nabla\Psi)\cdot(i\hbar\nabla\Psi^*)
= \hbar^2(\nabla\Psi^*\cdot\nabla\Psi),\]
\[i\hbar\nabla\cdot\left(\Psi^*\hat{\text{P}}\Psi\right)
   = \hbar^2\nabla\cdot\left(\Psi^*\nabla\Psi\right)\]
\[\hspace{-2ex} \Psi^*\left(-\hbar^2\nabla^2 \Psi\right)
= \hbar^2\nabla\Psi^*\cdot\nabla\Psi - \hbar^2\nabla\cdot\left(\Psi^*\nabla\Psi\right).\]
%
\[ \hs{-14ex} \text{For Eq.~(\text{\ref{momen}})}, \hs{1.1ex}
\vert\hat{\text{\rm P}}\Psi\vert^2\rho = \rho(-i\hbar\nabla\Psi)\cdot(i\hbar\nabla\Psi^*)
= (-i\hbar\Psi^*\nabla\Psi)\cdot(i\hbar\Psi\nabla\Psi^*) = \vert\Psi^*\hat{\text{\rm P}}\Psi\vert^2. 
\]
The multiplication of identity~(\ref{laplacian2b}) by $\rho/2m$ and
using identity (\ref{momen}) gives Eq.~(\ref{laplacian2}).
\begin{lemma}[Schr\"odinger Eq.~and Vanish of Divergence Momentae Integral]\zlabel{vanlem} \mbox{}  
The $n$-body Schr\"odinger equation, in the form
  \begin{equation} \hs{-5ex}\label{schrodinger}
    i\hbar\Psi^*\partial\Psi = \Psi^*\hat{H}\Psi; \qquad
    \hat{H} \defi = -\fc{\hbar^2}{2m}\nabla^2 + U,\quad\!\text{and}\quad \nabla^2\Psi\in\mathbf{\Sigma},
\end{equation} 
  can be written \vspace{0.5ex}
\begin{equation} \label{schrodinger2}
i\Psi^*\hbar\partial\Psi =
   \fc{1}{2m}\vert\hat{\text{\rm P}}\Psi\vert^2 - i\fc{\hbar}{2m}\nabla\cdot\left(\Psi^*\hat{\text{\rm P}}\Psi\right) + |\Psi|^2\room U.
\end{equation} \vspace{1.5ex}
The real and imaginary parts are
\begin{equation} \hs{-5ex}\label{schre}
  -\hbar\text{Im}\left(\Psi^*\partial\Psi\right) = \fc{1}{2m}\vert\hat{\text{\rm P}}\Psi\vert^2
  + \fc{\hbar}{2m}\text{Im}\left(\nabla\cdot\left(\Psi^*\hat{\text{\rm P}}\Psi\right)\right) + \vert\Psi\vert^2 U,
  \quad\text{and}
\end{equation} 
\begin{equation} \label{schim}
  \hspace{2ex}-\hbar\text{Re}\left(\Psi^*\partial\Psi\right)
   = \fc{\hbar}{2m}\text{Re}\left(\nabla\cdot\left(\Psi^*\hat{\text{\rm P}}\Psi\right)\right),
\end{equation}
respectively, where $U$ is the external potential.
For each component $(\nabla_i\cdot\left(\Psi^*\hat{\text{\rm
    P}_i}\Psi\right)$ of the sequence
$(\nabla\cdot\left(\Psi^*\hat{\text{\rm P}}\Psi\right) \in \mathbf{\Sigma})$,
\[
\text{if $\Psi$ vanish at infinity then}\hspace{1ex}
\int_{\mathbb{R}^3}\nabla_i\cdot\left(\Psi^*\hat{\text{\rm P}_i}\Psi\right)\hspace{0.9ex} \text{ vanish}.
\]
\end{lemma}  
\textbf{\textit{Proof}} Eq.~(\ref{schrodinger2}) follows from the substitution of
Eq.~(\ref{laplacian2b})$\div 2m$ into Eq.~(\ref{schrodinger}).  Eqs.~(\ref{schre}) and
(\ref{schim}) follow by separating Eq.~(\ref{schrodinger2}) into real and imaginary parts
and using the identities ($\text{Re}(iX) = -\text{Im}X)$ and ($\text{Im}(iX) = \text{Re}X )$.
For the second part,
\[
\hspace{-4ex}
\text{let }\quad \bar{\text{\rm P}}_\al = -i\hbar \fc{\partial}{\partial \al},
\quad  \al\in \{x,y,z\} = \mrr \in(\mrr_1,\mrr_2,\cdots,\mrr_n).
\]
\[\hspace{0ex}
\int_{\mathbb{R}}\fc{d}{d\al}\left(\Psi^*\bar{\text{\rm P}}_\al\Psi\right)
 = \left.\Psi^*\bar{\text{\rm P}}_\al\Psi\right\vert_{-\infty}^\infty = 0.
\]
\begin{flushright}\fbox{\phantom{\rule{0.5ex}{0.5ex}}}\end{flushright}
\begin{definition} 
  \mbox{}

  We replace the wavefunction $\Psi$ as the fundamental variable, with the variables
  $\rho$ and $S$, such that  
\begin{equation}\zlabel{def00}
  \Psi \defi Re^{iS/\hbar} \defi e^{(iS-\theta)/\hbar}, \quad\rho\defi \Psi^*\Psi = R^2,
  \quad \int\rho = 1,
\end{equation}
where $\Psi$ is defined to within a phase factor. Hence, we have the composite $\Psi = \Psi(\rho,S)$.
Also, $\theta$ is defined under the condition $(R> \zero)$, giving $(\rho = e^{-2\theta/\hbar})$.
(For convenience, this definition is given in a different form with Defs.~(\ref{theta}).) 
\end{definition}
\begin{lemma}[Gradient, Divergence and Laplacian Identities in Polar Form] \label{Lap2} 
\end{lemma}
\begin{equation} \zlabel{p5204}
\Psi^*\hat{\text{P}}\Psi = \rho\nabla S - i\fc{\hbar}{2}\nabla\rho. 
\end{equation}
\begin{equation} \zlabel{p5208}
  \fc{1}{2m}\vert \Psi^*\hat{\text{P}}\Psi\vert^2
  = \fc{\hbar^2}{8m}\vert\nabla\rho\vert^2 + \fc{1}{2m} \rho^2\vert\nabla S\vert^2.
\end{equation}
\begin{equation} \zlabel{p5212}
  - i\fc{\hbar}{2m}\nabla\cdot\left(\Psi^*\hat{\text{\rm P}}\Psi\right)
  =  -\fc{\hbar^2}{4m}\nabla^2\rho - i\fc{\hbar}{2m}\nabla\cdot(\rho\nabla S).
\end{equation}
\begin{equation} \hspace{-10ex}\label{pexpect}
  \Psi^*\left(-\fc{\hbar^2}{2m} \nabla^2 \Psi\right)\rho
  = \fc{\hbar^2}{8m}\vert\nabla\rho\vert^2
    + \fc{1}{2m}\rho^2\vert\nabla S\vert^2
  -\fc{\hbar^2}{4m} \rho\nabla^2\rho
  -i\fc{\hbar}{2m}\rho\nabla\cdot(\rho\nabla S).
\end{equation}
\textbf{\textit{Proof}}
Let ($S^\pr \defi S/\hbar$). Consider the development
\[Re^{-iS^\pr}(-\nabla) Re^{iS^\pr} = -iR^2\nabla S^\pr - R\nabla R  = -i\rho\nabla S^\pr - \fc12 \nabla \rho.\]
Multiplying the equation by $i\hbar$, setting ($S^\pr = S/\hbar$), and using the
definition of the momentum operator $(\hat{\text{P}}\Psi) = (-i\hbar\nabla\Psi)$, we obtain
Eq.~(\ref{p5204}). Eqs.~(\ref{p5208}) and (\ref{p5212}) follow immediately. Substituting
Eqs.~(\ref{p5208}) and (\ref{p5212}) into Eq.~(\ref{laplacian2})$\div(2m)$ we obtain Eq.~(\ref{pexpect}).
\begin{flushright}\fbox{\phantom{\rule{0.5ex}{0.5ex}}}\end{flushright}
For later convenience, we note that Eq.~(\ref{p5204}) holds with the replacement
$(-\nabla)\to\partial$:
\begin{equation} \zlabel{def02}
  i\hbar\Psi^*\partial\Psi = -\rho\partial S + i\fc{\hbar}{2}\partial\rho.
\end{equation}
Only the derivatives of $S$ are needed for the correspondence variables. It follows from
(\ref{p5204}) and (\ref{def02}) that $\nabla S$ and $\partial S$ can be computed from a
wavefunction, as the real parts of $\Psi^*\hat{\text{P}}\Psi/\rho$ and
$-i\hbar\Psi^*\partial\Psi/\rho$, respectively, even if the wavefunction is not written in
polar form, as in a linear combination.

\section{Polar Forms of Correspondence Variables \zlabel{4193}}

In this section, polar forms of the correspondence are derived.

\subsection{Momentae Strong and Kinetic Energy Strong by a Composition \zlabel{momentae}}

In this subsection, polar forms of the correspondence strong variables for momentum are given,
where the ordered momentum pair is called a momentae. The kinetic energy
variable by composition is defined and obtained.  A discussion of one of the velocity fields
is given.
\begin{definition}\mbox{\rm (Momentae Runners $\{(\rho_m\bv) ,(\rho_m\bu)\}\subset\mathbf{\Sigma}$)}
  \begin{equation}\label{def10} 
\Psi^*\text{\rm $\hat{\text{P}}$}\Psi \defi \rho_m\bv + i\rho_m\bu, \quad \rho_m \defi m\rho,
\end{equation}
\begin{equation} \label{theta} 
  \zeta_0 \defi \fc{\hbar}{2},
  \qquad \zeta = \fc{1}{m}\zeta_0,
   \qquad \eta = \zeta_0\rho,
  \quad  \theta \defim -\zeta_0\ln\rho.
\end{equation}
The parameter $m$ is the mass of each and every identical particle of a quantum-mechanical
system. The functions $\bv$ and $\bu$ are the first and second velocities of the order pair
$(\bv,\bu)$.
(To define $\ln\rho$, $\rho$ is unitless in $\ln\rho$.)
\end{definition}

Comparing (\ref{p5204}) and (\ref{def10}), and using $(\ref{theta})$, we find that
\begin{equation} \hspace{-8ex}\zlabel{5204b} 
  \rho_m\bv = \rho\nabla S \defi \zeta_0\nabla\phi + \zeta_0\nabla\times\mathbf{J}
  \defi \zeta_0\mathbf{Q},
\qquad
  \rho_m\bu = \rho\nabla\theta = -\zeta_0\nabla\rho,
\end{equation}
where the first equations string defines the scalar $(\phi)$ and vector $(\mathbf{J})$ fields by
Helmholtz decomposition.
By definition~(\ref{def10}) and definition~{\bf 1b} of Sec~\ref{4930},
$(\hspace{0.2ex}(\rho_m\bv) = (\rho\nabla S)\hspace{0.2ex})$ and $(\hspace{0.2ex}(\rho_m\bu) =
(\rho\nabla\theta)\hspace{0.2ex})$ are the strong--fluid momentae variables.
Similarly, $(m\bv)$ and $(m\bu)$ are the strong--particle momentae variables. These variables
are irrotational fields with $S$ and $\theta$ particle-momentae potentials.
For a fluid interpretation of 1-body states, $\rho_m$ is the mass density, and 
the real and imaginary parts of the following are the mass fluxes:
\begin{equation*}
  \nabla\cdot\left(\Psi^*\hat{\text{P}}\Psi\right) = \nabla\cdot(\rho_m\bv) + i\nabla\cdot(\rho_m\bu).
\end{equation*}
Using the composite $(\hspace{0.2ex}g(\rho^{-1}\Psi^*\hat{\text{P}}\Psi)\hspace{0.2ex})$ with
$(g(x) = \vert x\vert^2/2m)$, where $x$ is  complex valued, and noting that
$(\rho^{-1}\Psi^*\hat{\text{P}}\Psi = m\bv + im\bu)$, we obtain a strong--particle
kinetic-energy variable by composition:
\begin{equation} \hspace{-5ex}\zlabel{5208}
  \fc{1}{2m}\left\vert\fc{\Psi^*\hat{\text{P}}\Psi}{\Psi^*\Psi}\right\vert^2 =
  \fc{1}{2m}\vert\nabla\theta\vert^2 + \fc{1}{2m} \vert\nabla S\vert^2 = \fc12mu^2 + \fc12m v^2.
\end{equation}

For later convenience in Corollary~(\ref{bohmian}), we note that starting from the second equation
from~(\ref{5204b}), we have
\begin{equation} \zlabel{def77}
m\bu = -\zeta_0\fc{\nabla\rho}{\rho}\;\Rightarrow\; \fc12mu^2 =
\fc12\zeta_0^2m\vert\nabla\rho\vert^2\rho^{-2}.
\end{equation}
As in Def.~(\ref{2914}) for probability density $\hat{\rho}$, let 
\begin{equation} \zlabel{2915}
  \hat{\gamma} = \intf{1+}\! \gamma, \qquad \gamma = \phi, \mathbf{J}.
\end{equation}
As in Eq.~(\ref{5204b}), let $\hat{\bv}$ and $\hat{\bu}$ be defined by
\begin{equation} \zlabel{4025}
  \hat{\rho}_m\hat{\bu} \defi  -\zeta_0\nabla_1\hat{\rho}, \quad
  \hat{\rho}_m\hat{\bv} \defi  \zeta_0\nabla_1\hat{\phi} + \zeta_0\nabla\times\mathbf{\hat{J}}.
\end{equation}
Next we show that these definitions for $\hat{\bu}$ and $\hat{\bv}$ are equivalent to the
strong, first order reduced variables given by Def.~(\ref{7724}):
\begin{equation} \zlabel{4035}
\hat{\rho}\hat{\bu} = \intf{1+} \rho\bu_1,\qquad
  \hat{\rho}\hat{\bv} = \intf{1+} \rho\bv_1. 
\end{equation}
Def.~(\ref{5204b}.2) and (\ref{2915}) give
\begin{equation} \hs{-4ex} \zlabel{4372}
  \hat{\rho}_m\hat{\bu} = \intf{1+}\rho_m\bu_1 =
  -\zeta_0\intf{1+}\nabla_1\rho \; d\mrr_1^\pr =  -\zeta_0\nabla_1\hat{\rho}.
\end{equation}
Similarly, Def.~(\ref{5204b}.1) and (\ref{2915}) give
\begin{equation}  \hs{-12ex} \zlabel{4374}
  \hat{\rho}_m\hat{\bv} = \sum_{1+}\int\rho_m\bv_1\, d\mrr_1^\pr =
\zeta_0\sum_{1+}\int  \left(\nabla_1\phi + \zeta_0\nabla_1\times\mathbf{J}\right) \; d\mrr_1^\pr
= \zeta_0\nabla_1\hat{\phi} +  \zeta_0\nabla_1\times\mathbf{\hat{J}}.
\end{equation}
Hence, the two definitions for vector fields $\hat{\bu}$ and $\hat{\bv}$ are equivalent. This
same result holds for any vector field-variable that can be expressed as a Helmholtz decomposition,
where $\hat{\rho}_m\hat{\bu}$ is special, since it is irrotational.

\begin{flushright}\fbox{\phantom{\rule{0.5ex}{0.5ex}}}\end{flushright}

\hs{-2ex}\mbox{\rm(\mbox{\it \bf Velocity singularities})} The velocities $(\bv)$ and $(\bu)$ have
singularities when $\rho$ vanish at nodes, and this causes some equations obtained below not to
be defined for these subspaces \cite{Finley-Arxiv,Finley-Bern,B6}. Since these subspaces have
measure zero, these singularities pose no problems if  one works in a Hilbert space. These
singularities can be removed by using the fluid-variable momentum to replace the
particle-variable momentum, since such a modification give finite limits as $\rho\to 0$, at
least for a few systems that have been studied \cite{Finley-Bern}.  However, such treatments
seem to be counter production, since singularities in variables of classical mechanics can be
useful. For example, the electric field on the surface of a conductor can be discontinuous,
giving singularities in the electric-field gradient. Since singularities involving $(m\bu)$ and
$(m\bv)$ might have some usefulness, and these singularities arise naturally out of the
formalism, in some equations and definitions we choose to use the particle momentum-variables
$(m\bv)$ and $(m\bu)$, instead of the corresponding fluid ones $(\rho_m\bv)$ and $(\rho_m\bu)$.

\hs{-2ex}\mbox{\rm(\mbox{\it \bf Other work with velocity $\bu$})} The second velocity $\bu$ of
Eq.~(\ref{def77}) for one-body systems, appears in other investigations.  The function
$2\rho_m\vert\bu\vert^2$ is a term of the Hamiltonian functional of the generalized
fluid-dynamics formalism by Broer \cite{Broer}, where the Hamiltonian functional is derived
from the time-dependent Schr\"odinger equation.  Salesi \cite{Salesi} obtains a Lagrangian
function that is equivalent to the Madelung equations that contains the term
$\rho_m\vert\bu\vert^2/2$.  He interprets $\vert\bu\vert^2/2$ as the internal energy of the
relative motion in the center of mass coordinate frame from the \mbox{Zitterbewegung} (ZWB)
model of spin. The velocity magnitude $\vert\bu_\pm\vert$ follows as a non-relativistic
approximation of a velocity expression of Hestenes \cite{Hestenes} of Schr\"odinger--Pauli
theory.  Tsekov \cite{Tsekov} obtains the same velocity choice $\bu$ of Eq.~(\ref{def77}), as
the imaginary component of a complex velocity, where the formalism involves diffusion.
Furthermore, the velocity choice $\bu$ is used in the Bernoullian equation, an equation
equivalent to the real-valued one-body Schr\"odinger equation, where the flow is compressible
with variable-mass \cite{Finley-Arxiv,Finley-Bern}.

\subsection{Kinetic energy strong and weak  \zlabel{1210}}

In this subsection, polar forms of strong and weak correspondence variables for kinetic energy are
given.  The following lemma contains definitions and a proof that are needed below.
\begin{lemma}[The Free Pressure Variables] \zlabel{presseq}
The definition
\begin{equation} \hspace{-4ex}\zlabel{5205} 
  -i\fc{\hbar}{2m}\nabla\cdot\left(\Psi^*\hat{\text{\rm P}}\Psi\right) \defi \Pu + i\Pv,
  \qquad \hat{\text{\rm P}}\Psi = -i\hbar\nabla\Psi,
\end{equation}
is equivalent to the following two alternative definitions:
\begin{equation} \hspace{-7ex}\label{4720}
  \mathbf{1)}\quad \Pu \defim \zeta_0\nabla\cdot(\rho\bu), \quad\;\; \Pu \defim -\zeta\zeta_0\nabla^2\rho,
  \qquad \Pu\in\mathbf{\Sigma}.
\end{equation}
\begin{equation} \hspace{-7ex} \label{4725}
  \mathbf{2)} \quad \Pv \defim -\zeta_0\nabla\cdot(\rho\bv),\quad  \Pv \defim -\zeta\zeta_0\nabla^2\phi,
    \qquad \Pv\in\mathbf{\Sigma}.
\end{equation}
Also, for each component of the sequences $(\Pu)$ and $(\Pv)$, 
\begin{equation} \hspace{-7ex}\label{4730}
  \int \Pu_i, \quad\text{and}\quad \int \Pv_i \;\text{ vanish.}
\end{equation}
\end{lemma}
\textbf{\textit{Proof}} Substituting Eq.~(\ref{def10}) into (\ref{5205}), followed by the
substituting of Eq.~(\ref{5204b}) and using vector calculus identity $(\nabla\cdot\nabla\times\mathbf{J} =
\zero)$, gives Eqs.~(\ref{4720}) and (\ref{4725}).  Statement~(\ref{4730}) follows from
lemma~\ref{vanlem} and definition~(\ref{5205}).
\begin{flushright}\fbox{\phantom{\rule{0.5ex}{0.5ex}}}\end{flushright}
%
%
%
\begin{lemma}[Correspondence Variables from the Kinetic Energy.] \zlabel{5538}\mbox{}
Let $(\Pu)$ and $(\Pv)$ be defined by (\ref{5205}), and let $(u) \defi (\vert\bu\vert)$ and
$(v) \defi (\vert\bv\vert)$, where $(\bu)$ and $(\bv)$ are defined by (\ref{def10}), in other words
\begin{equation}\zlabel{def7}
  v\rho_m \defi \left|\text{\rm Re}\hspace{0.1ex}(\Psi^*\text{\rm $\hat{\text{P}}$}\Psi)\right|, \quad
  u\rho_m \defi \left|\text{\rm Im}\hspace{0.1ex}(\Psi^*\text{\rm $\hat{\text{P}}$}\Psi)\right|.
\end{equation}
\begin{equation} \zlabel{lapl}
\Psi^*\left(-\fc{\hbar^2}{2m} \nabla^2 \Psi\right)  = \fc12 \rho_mu^2 + \fc12\rho_mv^2 + \Pu  + i\Pv.
\end{equation}
\end{lemma} 
\textbf{\textit{Proof}}
From definition~(\ref{def10}), and the ones for $(u)$ and $(v)$ above, we have
%
\begin{equation} \zlabel{8302}
  \fc{1}{2m}\vert\Psi^*\hat{\text{P}}\Psi\vert^2\rho^{-1} = \fc12\rho_mu^2 + \fc12\rho_m v^2.
\end{equation}
Substituting this result, and also the definitions for $(\Pu)$ and $(\Pv)$,
Eq.~(\ref{5205}), into identity (\ref{laplacian2})$\times\rho^{-1}$,
\begin{equation*} \hspace{-5ex} 
 \Psi^*\left(-\fc{\hbar^2}{2m} \nabla^2 \Psi\right) = \fc{1}{2m}\vert\Psi^*\hat{\text{\rm
      P}}\Psi\vert^2\rho^{-1} - i\fc{\hbar}{2m}\nabla\cdot\left(\Psi^*\hat{\text{\rm P}}\Psi\right),
\end{equation*}
we obtain Eq.~(\ref{lapl}).
\begin{flushright}\fbox{\phantom{\rule{0.5ex}{0.5ex}}}\end{flushright}
The real and imaginary parts of Eq.~(\ref{lapl}) gives formulae for the two strong--fluid
kinetic-energy variables, and Eq.~\hspace{-0.45ex}(\ref{pexpect})$\times\rho^{-1}$ gives the
polar form representation of these variables.  Combining identities (\ref{lapl}) and
(\ref{laplacian2}), we discover, after using Eq.~(\ref{4730}) and the last result from
Lemma~\ref{vanlem}, that the expectation value of the kinetic energy can be written
\begin{equation} \label{kinweak}
  \left\langle -\fc{\hbar^2}{2m} \nabla^2\right\rangle = \intf{}\;\fc{1}{2m}\vert\Psi^*\hat{\text{P}}\Psi\vert^2 \rho^{-1}
  = \intf{}\;\; \fc12\rho_mu^2 + \fc12\rho_mv^2.
\end{equation}
From this result and Eq.~(\ref{5208}), we
find that $\left(\vert\Psi^*\hat{\text{P}}\Psi\vert^2\rho^{-1}/(2m)\right)$ is a fluid
kinetic-energy variable that is strong by composition and weak by equality~(\ref{kinweak}).
\vspace{3ex}

Since the probability distribution is symmetric with respect to the interchange of two electron
coordinates, e.g., $\rho(\mrr_1,\mrr_2,\cdots) = \rho(\mrr_2,\mrr_1,\cdots)$, each of the $n$
terms of the integrand in Eq.~(\ref{kinweak})---the terms involving a runner---are equal, so we
can write
\begin{equation} \zlabel{kinfunct}
  \left\langle -\fc{\hbar^2}{2m} \nabla^2\right\rangle = n\intf{1}\;\, \left(\fc12\rho_mu_1^2 + \fc12\rho_mv_1^2\right)
  = \int_1 K(\mrr),
\end{equation}
where the second integration is over $\mathbb{R}^{3}$, and
\[
K(\mrr) \defi n\sum_1 \intf{1+} \left(\fc12\rho_mu_1^2 + \fc12\rho_mv_1^2\right).
\]
Equation~(\ref{kinfunct}) is an alternative way to express the kinetic energy functional used
in reduced density matrix theory \cite{Parr:89,Cioslowski}, replacing the integrand
$(-\hbar^2/2m)\nabla_\mrr\rho_1(\mrr,\mrr^\pr)\vert_{\mrr^\pr = \mrr}$ of the kinetic energy
functional with the kinetic energy density $K$, where $\rho_1$ is the reduced, one-body,
density matrix.
Comparing the two, it seems that the extra variable of $\rho_1(\mrr,\mrr^\pr)$ has no physical
meaning.  For ground states, the kinetic energy functional~(\ref{kinfunct}) is equal to the
kinetic energy functional of the Hohenberg--Kohn theorem \cite{Parr:89,Dreizler}, a functional
that is also used in Kohn--Sham density functional theory as part of the kinetic energy
correction.

\subsection{Total energy strong and weak}

In this subsection, polar forms of correspondence variables for the total energy are given.
\begin{definition} \mbox{} 
\begin{equation} \hspace{2ex} \zlabel{5405} 
    i\hbar\Psi^*\partial\Psi \defi \ES\rho + i\Etheta\rho,\quad \ES,\Etheta\notin\mathbf{\Sigma}.
\end{equation}
The potential $U$ of the Hamiltonian operator $\hat{H}$ is a multiplicative operator such that
for each eigenvector $\Psi$ of $\hat{H}$, the subspace $\mathbf{Z} \defi \{ \Psi(\mrr,t) =
0\vert\mrr\in\mathbb{R}^{3n}\}$ has measure zero for all $t\in\mathbb{R}$.  (A condition on $U$
that is satisfied if $U\subset (L^{3/2} + L^\infty)$, where $(L^{3/2} + L^\infty)$ is a Hilbert
space \cite{Lieb}.)
\end{definition}
Comparing Eq.~(\ref{def02}) and (\ref{5405}) and using Eq.~(\ref{theta}) for $\theta$,
we discover that
\begin{equation} \zlabel{en2}
  \ES = -\partial S, \quad  \Etheta\rho = \zeta_0\partial\rho,\quad\text{and}\quad 
    \Etheta = -\partial\theta. 
\end{equation}
%
%
Taking the gradient and then using the particle momentae Eq.~(\ref{5204b}), we have
\begin{equation} \zlabel{4887}
  -\nabla\ES = m\partial\bv,\qquad -\nabla\Etheta = m\partial\bu,
\end{equation}
\emph{the local balance of particle momentae}.  The fields $\partial\bv$ and $\partial\bu$ of
1-body states are identified as \emph{local} accelerations of classical fluid dynamics
\cite{Munson}. $\partial\bv(\mrr,t)$ and $\partial\bu(\mrr,t)$ are the rate of change of the
velocity fields $\bv$ and $\bu$ at point $\mrr\in\mathbb{R}^3$ and time $t\in\mathbb{R}$. For
steady flow, these partial derivatives vanish, and the acceleration of a fluid particle is
obtained by the time derivative of the composite, $(\room(\bu + \bv = [\bu + \bv](\mrr)\room)$,
such that $\mrr = \mrr(t)$.
We identify $-(\nabla\ES)$ and $-(\nabla\Etheta)$ of Eqs.~(\ref{4887}) as correspondence
variables by equation of local components of force, as defined in Sec.~(\ref{4930}).
Sec.~\ref{4779} demonstrates that the energae $(\ES,\Etheta)$ have the important property of
being conserved over all space for all times.

Adding the potential-energy term $U\rho$ from the Schr\"odinger equation to both sides of
Eq.~(\ref{lapl}), and then using the definition of the Hamiltonian operator $\hat{H}$,
Eq.~(\ref{schrodinger}.2), we find the pair of strong--fluid energy variables from the Hamiltonian
operator $\hat{H}$:
\begin{equation} \zlabel{enstrong}
\Psi^*\hat{H}\Psi  = \fc12 \rho_mu^2 + \Pu  + \fc12\rho_mv^2 + U\rho + i\Pv.
\end{equation}
It follows from equalities~(\ref{4730}) and (\ref{kinweak}), and the definition of the
Hamiltonian operator $\hat{H}$, that
\begin{equation} \hspace{-3ex}\label{enweak}
  \left\langle\hat{H}\right\rangle = \intf{}\;\; \fc{1}{2m}\vert\Psi^*\hat{\text{P}}\Psi\vert^2\rho^{-1}  + U\rho
  = \intf{}\;\;\, \fc12\rho_mu^2 + \fc12\rho_mv^2 + U\rho.
\end{equation}
Hence, the integrand on the rhs is a weak fluid-variable of the energy.

\subsection{Energy equations \zlabel{1070}}

In this subsection we obtain an energy equation that is a generalization of the Bernoulli equation
of fluid mechanics. In Sec.~\ref{4491}, this equation, together with a continuity equation, is
proved to be equivalent to the space time Schr\"odinger equation $(i\hbar\partial\Psi = \hat{H}\Psi)$.
A discussion about the energy equation is given, where correspondence variables by equation are
identified.

So far in this paper, all equalities are either definitions or they are implied by
definitions. Therefore, their satisfaction does not required the wavefunction determining
the variables to satisfy a Schr\"odinger equation.
For this paragraph, we impose the condition that the wavefunction $\Psi$ is a solution of the
Schr\"odinger equation $(i\hbar\partial\Psi = \hat{H}\Psi)$ with potential $U$. This
requirement and definition~(\ref{5405}) implies that ($\ES\rho + i\Etheta\rho =
\Psi^*\hat{H}\Psi$).  Substituting Eqs.~(\ref{enstrong}) into this implication, we discover two
energy equations:
\begin{equation} \zlabel{eneq} 
  \ES\rho = \fc12\rho_mu^2 + \fc12\rho_mv^2 + \Pu + U\rho, \qquad \Etheta\rho = \Pv.
\end{equation}
It immediately follows Eqs~(\ref{eneq}), (\ref{enweak}), and (\ref{4730}) that
the expectation value of the energy from the Hamiltonian operator is
\begin{equation} \label{4529}
\langle\hat{H}\rangle = \langle\ES\rangle; \qquad \langle\ES\rangle \defi \intf{}\; \ES\rho. 
\end{equation}
A reasonable generalization involving $\Etheta$ and some operator $\hat{G}$ is
$\langle\hat{G}\rangle$, such that $(\Psi^*\hat{G}\Psi \defi \Etheta\rho)$, and from
Eqs.~(\ref{eneq}) and (\ref{4730}), $\langle\hat{G}\rangle$ vanish. From Sec.~\ref{4930},
definition {\bf 3}, $\Etheta$ is a free particle-variable of the energy.

\section{Time derivative relations \zlabel{8920}}

In this section, a number of relationships are obtain in the case of continuity equation
$(\partial\rho = -\nabla\cdot(\rho\bv)\hspace{0.2ex})$ satisfaction involving the first
velocity (generalized) field $\bv$.

\subsection{Equivalent forms of the continuity equation}

From the following theorem, six equations are shown to be equivalent to the continuity
equation, including the imaginary part of the Sch\"odinger Eq.~(\ref{schrodinger}). A set of
scalar fields containing the element $\partial\rho$ is identified, such that if one member
vanish then all members vanish.

\begin{theorem}[The Continuity Six and One for All Vanish] \zlabel{conti} 
\mbox{}
  
\noindent
{\bf I.} The following equations are equivalent:
\begin{equation} \label{cont}
  \hs{-5ex} {\bf 1)} \hspace{2ex} \Pv = \zeta_0\partial\rho.\qquad 
{\bf 2)} \hspace{2ex} \partial\rho + \nabla\cdot(\rho\bv) = \mathbf{0}.
\qquad  {\bf 3)}\hspace{2ex} (\ref{eneq}): \Etheta\rho = \Pv.
\end{equation}
\begin{equation*}
 {\bf 4)}\hspace{6ex} \text{Im}\left(i\hbar\Psi^*\partial\Psi\right) = \text{Im}\left(\Psi^*\hat{H}\Psi\right).
\end{equation*}
\begin{equation} \label{probc}
{\bf 5)}\hspace{2ex} \partial\rho  = -\nabla\cdot\mathbf{j}\hs{0.2ex}; \quad \mathbf{j} =  \fc{\hbar}{2mi}\left(\Psi^*\nabla\Psi - \Psi\nabla\Psi^*\right).
\end{equation}
\begin{equation*}
{\bf 6)}\hspace{2ex} Eq.~(\ref{schim})\!:\;
\hbar\text{Re}\left(\Psi^*\partial\Psi\right)
   = -\fc{\hbar}{2m}\text{Re}\left(\nabla\cdot\left(\Psi^*\hat{\text{\rm P}}\Psi\right)\right).
\end{equation*}

\noindent
{\bf II.} Let the continuity set $\mathbb{C}$ be 
\[\hs{-2ex}\mathbb{C} = \{\partial\rho,\Pv,\Etheta,\hspace{0.55ex}\nabla\cdot(\rho\bv), 
 \nabla\cdot\mathbf{j}\hspace{0.1ex}\hspace{0.2ex},
   \hspace{0.7ex}\text{Re}\hspace{0.1ex}\nabla\cdot(\Psi^*\hat{\text{\rm P}}\Psi),
\hspace{0.1ex}\text{Im}(i\hbar\Psi^*\partial\Psi)\},\] 
where it is understood that the runner members run, e.g., $\Pv$ above means $\Pv_1 + \cdots +
\Pv_n$. Let $\mathbb{S}\in\mathbb{R}^{3n}$ have non zero measure. Let Eq.~\#{\bf 4} be true on $\mathbb{S}$.  If
there exist an element $X\in \mathbb{C}$ that vanish on $\mathbb{S}$, that is $(X\vert_S = \zero)$, then
each and every element of $\mathbb{C}$ vanish on $\mathbb{S}$. \vspace{1ex}

 \noindent
    {\bf IIIa} Let $(\Psi = \Psi(\rho,S))$ satisfy the Schr\"odinger Eq.~(\ref{schrodinger}). Let
    $\ES$ be a constant.  $(-\partial S = \ES)$ implies all members of the continuity set
    $\mathbb{C}$ vanish. \vspace{1ex}

 \noindent
    {\bf IIIb} Let $(\Psi = \Psi(\rho,S))$ satisfy the continuity Eq.~(\ref{conti}.1).  $(\nabla
    S = \zero)$ implies all members of the continuity set $\mathbb{C}$ vanish.

 \noindent
     {\bf IIIc} Stationary state implies $S(\mrr,t) = \ES t + \beta(\mrr)$;
     $t\in\hat{\mathbb{R}}$ and $\beta\in\mathbb{R}^{3n}$ and $\ES = \text{const.}$.
     
\end{theorem} 
\textbf{\textit{Proof}} 

{\bf I.} by inspection of Def.~(\ref{4720}) for $\Pv$, we
find: \raisebox{-0.3ex}{\fbox{$\text{\bf 1} \Leftrightarrow \text{\bf 2}$}}.  By inspection of
Eq.~(\ref{en2}), $(\Etheta\rho = \zeta_0\partial\rho)$, we find:
\raisebox{-0.3ex}{\fbox{$\text{\bf 1} \Leftrightarrow \text{\bf 3}$}}.  Consider Eq.~{\bf
  6}. From Def.~(\ref{def02}), we find: $(\hbar\text{Re}\left(\Psi^*\partial\Psi\right) =
\hbar\partial\rho/2)$, and from Def.~(\ref{5205}), we find:
$-\hbar\text{Re}\left(\nabla\cdot\left(\Psi^*\hat{\text{\rm P}}\Psi\right)\right)\div 2m =
\Pv$.  Hence, \raisebox{-0.3ex}{\fbox{$\text{\bf 1} \Leftrightarrow \text{\bf 6}$}}. From
Lemma~\ref{vanlem}, \raisebox{-0.3ex}{\fbox{$\text{\bf 6} \Leftrightarrow \text{\bf 4}$}}.
Consider Eq.~{\bf 5}. From Def.~(\ref{5205}) for $\Pv$, we have
\[\Pv = -\fc{1}{2m}\text{\rm Im}\left(i\hbar\nabla\cdot\left(\Psi^*\hat{\text{\rm P}}\Psi\right)\right)
= -\fc{1}{2m}\text{\rm Im}\left(\hbar^2\nabla\cdot\left(\Psi^*\nabla\Psi\right)\right),\]
\begin{equation} \zlabel{5201}
\Pv = -\fc{\hbar^2}{4mi}\nabla\cdot\left(\Psi^*\nabla\Psi - \Psi\nabla\Psi^*\right).
\end{equation}
Using Eq.~{\bf 1}, in the form $(\partial\rho = \zeta_0^{-1}\Pv)$, we find:
\raisebox{-0.3ex}{\fbox{$\text{\bf 1} \Leftrightarrow \text{\bf 5}$}}.

So far we have {\bf 1} is equivalent to {\bf 2,3,5} and {\bf 6}. We also have: {\bf 1}
$\Leftrightarrow$ {\bf 6} and {\bf 6} $\Leftrightarrow$ {\bf 4} $\imply$ {\bf 1}
$\Leftrightarrow$ {\bf 4}. Hence, {\bf 1} is equivalent to each and every one. By the
transitive property of equivalence relations, each and every pair from the set of six
statements is equivalent.
\vspace{1ex}

{\bf II.} Let the implication
$(A\vert_S = \zero)\imply (B\vert_S = \zero)$, a relation between $A,B\in\mathbb{C}$, be
denoted $A\rightarrow B$. Let the equivalence relation $A\leftrightarrow B$ be $A\rightarrow B$
and $B\rightarrow A$.  We need to prove that for all $C,D\in\mathbb{C}$, $C\leftrightarrow
D$. Let \underline{$(Q = \text{Re}\hspace{0.1ex}\nabla\cdot(\Psi^*\hat{\text{\rm
      P}}\Psi)\hspace{0.1ex})$}.

First we prove that for all $C,D\in\mathbb{C}$, except $C = Q$ or $D = Q$, $C\leftrightarrow
D$.  From {\bf 1} and {\bf 2}, \rsb{$\Pv \rightarrow
  \partial\rho\rightarrow\nabla\cdot(\rho\bv)$}. From {\bf 3} and Def.~(\ref{4720}) for $\Pv$,
$\rsb{$\nabla\cdot(\rho\bv)\rightarrow\Etheta$}$ a.e.  From Def.~(\ref{5405}) for $\Etheta$,
\rsb{$\Etheta \rightarrow \text{Im}(i\hbar\Psi^*\partial\Psi)$}. From condition {\bf 4}, we
also have \rsb{$\text{Im}(i\hbar\Psi^*\partial\Psi)\rightarrow
  \text{Im}(\Psi^*\hat{H}\Psi)\notin\mathbb{C}$}.

The implication $\text{Im}(\Psi^*\hat{H}\Psi) \rightarrow \nabla\cdot\mathbf{j}$ requires three
steps. From identity~(\ref{laplacian2}), and because $U\rho$ is real valued, {\bf i)}
$\text{Im}(\Psi^*\hat{H}\Psi) \leftrightarrow Q$ a.e.  By noting that {\bf ii)} $Q \rightarrow
\text{\rm Im}(i\hbar\nabla\cdot(\Psi^*\hat{\text{\rm P}}\Psi))$, and examination
Eqs.~(\ref{5201}) with the beginning part `$\Pv = $'' omitted, {\bf iii)} $\text{\rm
  Im}(i\hbar\nabla\cdot(\Psi^*\hat{\text{\rm P}}\Psi))\rightarrow \nabla\cdot\mathbf{j} $.
Hence, \rsb{$\text{Im}(\Psi^*\hat{H}\Psi) \rightarrow \nabla\cdot\mathbf{j}$}.  From {\bf 5},
\rsb{$\nabla\cdot\mathbf{j} \rightarrow \partial\rho$}. From {\bf 1}, \rsb{$\partial\rho
  \rightarrow \Pv$}. Hence, for all $C,D\in\mathbb{C}$, except $C = Q$ or $D = Q$,
$C\leftrightarrow D$.  Since, as mentioned above, $\text{Im}(\Psi^*\hat{H}\Psi) \leftrightarrow
Q$, if follows from the transitive property of equivalence relations that each every element
$(X\vert_S = \zero)$, or equation, from on the set $\{(X\vert_S = \zero)\vert X
\in\mathbb{C}\}$ is equivalent. \vspace{1ex}

\noindent
\vspace{1ex}
{\bf IIIa.} $(-\partial S = \ES)$ implies that $\Psi = \Psi(\rho,S)$ satisfies the time independent \sch equation,
implying that $(\partial\rho = \zero)$. \vspace{1ex}

\vspace{1ex}
\noindent
{\bf IIIb.} The definition $(m\bv\defi\nabla S)$ continuity Eq.~(\ref{cont}.2) gives
\[\hs{-4ex}\nabla S = \zero\imply \bv = \zero \imply \nabla\cdot(\rho\bv) = \zero \imply \partial\rho
= \zero;\quad (\partial\rho\in\mathbb{C}). \]  
\vspace{1ex}
\noindent
{\bf IIIc.} With no loss of generality, let $(S(\mrr,t) = \ES^\pr t  +
\epsilon(\mrr,t) t + \beta(\mrr) \hs{0.1ex})$.  The only solution of $(i\hbar\Psi^*\partial\Psi = \ES\rho)$, with
$\ES$ being a constant, is $(\partial\epsilon = 0)$, where ($\ES = \ES^\pr + \epsilon$).
\begin{flushright}\fbox{\phantom{\rule{0.5ex}{0.5ex}}}\end{flushright}
The converse of part {\bf IIIa} is false: The vanishing of $\partial \rho$ does not imply a
stationary state. For such a state to be time dependent, from Eq.~(\ref{cont}.2),
$(\nabla\cdot\rho\bv) =0)$ must hold with $\rho$ fixed and $\bv$ changing with time. In this
special case, the real part of the \sch is a relation between $\nabla S$, $\partial S$ and $U$
with fixed $\rho$. Note that Eq.~(\ref{probc}) is an equation of probability conservation with
$\mathbf{j}$ being the probability current density \cite{Bransden}.


The equalities obtained for the remaining subsections of this section hold with continuity
equation satisfaction: $(\partial\rho + \nabla\cdot(\rho\bv) = \zero)$.
%

\subsection{Equations implied by the orthogonal condition $(\nabla\rho\cdot\nabla S = \zero)$. \zlabel{orthc}}

In this subsection, the special relationships are found when the orthogonal condition, defined
by $(\nabla\rho\cdot\nabla S = \zero)$; $\nabla\rho\cdot\nabla S\in\mathbf{\Sigma}$, is
satisfied. The interpretations given are for one-body states. Note that from Defs.~(\ref{5204b}),
$(\nabla\rho\cdot\nabla S = \zero)$ and $(\bu\cdot\bv = \zero)$ are equivalent.

The logic:
\begin{equation} \zlabel{5212}
  \bu\cdot\bv = \zero \Leftrightarrow \; m\nabla\cdot\bv = \nabla^2 S = -\fc{\partial\rho_m}{\rho};
  \qquad \nabla^2 S\in\mathbf{\Sigma},
\end{equation}
follows from the continuity equation, written ($\partial\rho + \nabla\rho\cdot\bv +
\rho\nabla\cdot\bv =\mathbf{0}$), ($m\bv = \nabla S$), and the observation that $\bu$ and
$\nabla\rho$ are perpendicular, if ($\bu\cdot\bv=\zero$). Applying the operator $\partial$
to Eq.~(\ref{5212}), and using Eqs.~(\ref{en2}), we have
%
\begin{equation*} 
  \hs{-0ex}
  \bu\cdot\bv = \zero \imply -\nabla^2 \partial S = \partial\left(\fc{\partial\rho_m}{\rho}\right)
\imply \nabla^2 \ES = \fc{1}{\zeta_0}\partial\Etheta.
\end{equation*}
The Poisson equation~(\ref{5212}) reduces to the Laplace equation for the time-independent case.
\begin{equation} \hs{-5ex}\label{5240}
  \bu\cdot\bv = \zero\hs{1ex}\text{and}\hs{1ex} \partial\rho = \zero \hs{2ex}
  \Leftrightarrow \hs{2ex}  \nabla\cdot\bv = \nabla^2 S = \zero;
    \qquad \nabla\cdot\bv \in\mathbf{\Sigma}.
\end{equation}
For 1-body, steady fluid flow with mass density $\rho_m$, we have $(\partial\rho_m = \zero)$, implying
that $(\nabla\cdot\bv = \zero)$, so that $\bv$ is both irrotational an solenoidal. A mass
conserving steady flow of classical mechanics that is both irrotational and solenoidal, and,
therefore, incompressible, is called potential flow, where the velocity potential satisfies
Laplace's equation.



\subsection{Local balance of the fluid momentum \text{\rm ($\rho_m\bu$)}}

Applying the gradient sequence $(\room(\nabla) \defi \nabla_1 + \cdots + \nabla _n\room)$ to the first
continuity Eq.~(\ref{cont}), gives the equation sequence ($\nabla[\hs{0.15ex}\Pv \hs{0.15ex}] =
\zeta_0\partial\nabla\rho$), where, in each of the $n$ equations, $([\hs{0.2ex}\Pv
  \hs{0.2ex}] = \Pv_1 +\cdots + \Pv_n)$. Using Eq.~(\ref{5204b})$\times\rho$,
$(\zeta_0\nabla\rho = -\rho_m\bu)$, we obtain the objective:
\begin{equation} \zlabel{7500}
-\nabla[\hs{0.2ex}\Pv \hs{0.2ex}] = \partial(\rho_m\bu),
\end{equation}
\emph{the local balance of fluid momentum $\rho_m(\bu)$}, where $(\nabla)$ and $(\bu)$ step.
For a one-body system, if the above equation is for a fluid with mass density $\rho_m$ and
velocity field $\bu$, then $\Pv$ is the portion of the total pressure responsible for the local
momentum time-derivative $[\partial\rho_m\bu](\mrr,t)$, at the fixed point $(\mrr,t)\in
\mathbb{R}^{3}\times\mathbb{R}$. According to Sec.~(\ref{4930}), this identifies $[\hs{0.2ex}\Pv\hs{0.2ex}]$ as
a correspondence pressure component from equation.
For later convenience, we note, from the derivation of
\text{Eq.~(\ref{7500})}, that
\begin{equation} \zlabel{8304}
  m\bu = -\zeta_0\fc{\nabla\rho}{\rho}\quad\text{and}\quad \Pv =
  \zeta_0\partial\rho \;\imply\; Eq.~(\ref{7500}).
\end{equation}

The remaining material from this section are not used later.
\subsection{Law of the pressures}

Applying the sequence operation $\zeta(\nabla\cdot)$ to
Eq.~(\ref{7500}) with switched order, we have
\[\zeta_0\partial\left(\nabla\cdot(\rho\bu)\right) = -\zeta\nabla^2[\hs{0.2ex}\Pv\hs{0.2ex}].\]
Substituting Eq.~(\ref{4720}), we discover
\begin{equation} 
  \partial\Pu = -\zeta\nabla^2\hs{0.05ex}[\hs{0.15ex}\Pv\hs{0.15ex}],
\end{equation}
\emph{the law of the pressures:} A sequence of Poisson equations.

\subsection{Laws of the momentae}
%
Multiplying the second continuity equation~(\ref{cont}) by $\zeta_0$, and using $(m\bv = \nabla S)$, we have
\[
-\zeta_0\partial\rho = \zeta_0\nabla\cdot(\rho\nabla S), \quad \text{and}
\]
\begin{equation*}
  \hspace{-14ex}
\text{then using the definition:\quad} \theta \defi -\zeta_0\ln\rho\;
\Rightarrow -\zeta_0\partial\rho = \rho\partial\theta, \;\;\;\text{from (\ref{theta}), gives}
\end{equation*} 
\begin{equation} \label{5722}
\rho\partial\theta = \zeta\nabla\cdot(\rho\nabla S),
\end{equation}
\emph{the law of the particle-momentae potentials}. 

Applying the sequence operator $\zeta_0\nabla$ to the to the second continuity Eq.~(\ref{cont}),
written $\Bigl(\partial\rho = -\nabla\cdot(\rho\bv)\Bigl)$, where $[\nabla\cdot\rho\bv]$ runs,
we obtain
\[\zeta_0\partial\nabla\rho = -\zeta_0\nabla\left[\nabla\cdot\rho\bv\right].\]
Substituting the second definition~(\ref{5204b}), written $(-\zeta_0\nabla\rho =
\rho_m\bu$, gives
\begin{equation} \zlabel{5328}
\partial(\rho_m\bu) = \zeta_0\nabla\left[\nabla\cdot\rho\bv\right],
\end{equation}
\emph{the law of the fluid momentae}.

\subsection{All space momentae, energae and probability density conservation \zlabel{4779}}

The proof below gives a relationship between the two energy functions $\ES$ and $\Etheta$,
involving two nondegenerate eigenstates of the Hamiltonian operator. Such states satisfy and do
not satisfy the time dependent and the time independent \sch equation, respectively. Both
energy quantities have the important property of being conserved over all space for all times.
A discussion is then give on how it can be generalized. The same idea holds for $\partial\rho$
and $\rho\partial S$, given Eq.~(\ref{en2}).
\begin{theorem}[Time dependence of $\partial\rho$ and $\rho\partial S$]
  Let
\[\hspace{-7ex} \left[e^{-i\varepsilon\circ}\right](t) - e^{-i\varepsilon t},  \quad \left[\cos\veps\circ\right](t) = \cos(\veps t),
  \quad \left[\sin\veps\circ\right](t) = \sin(\veps t)\]
\[ \hspace{-7ex}
\psi_i \defi e^{-i\veps_i\circ},\;\;
\vert C_1\vert^2 + \vert C_2\vert^2 = 1,\; C_i \in\mathbb{R},\; \rho_i\defi \phi_i\phi_i;\text{ and }\; i\in\{1,2\}.
\]
Let $\phi_1(\mrr)\psi_1(t)$ and $\phi_2(\mrr)\psi_2(t)$ be values of nondegenerate eigenfunctions of a Hamiltonian operator,
at time $t\in\hat{\mathbb{R}}$ and position $\mrr\in\mathbb{R}^{3n}$. Let $\sPsi$ be defined by the following:
\[\hs{-3ex}\sPsi(\mrr,t) \defi C_1\phi_1(\mrr)\psi_1(t) + C_2\phi_2(\mrr)\psi_2(t),\]
\begin{equation}\hs{-7ex}\label{2203} 
  \Etheta\rho = \hspace{0.5ex} C_1C_2\phi_1\phi_2(\veps_2-\veps_1)\sin[(\veps_1-\veps_2)\circ] = \fc{\hbar}{2}\partial\rho,
  \lefteqn{\hs{2ex}\text{and}}
\end{equation}
\begin{equation}\hs{-7ex}\label{2204}
    \ES\rho = C_1C_2\phi_1\phi_2(\veps_1+\veps_2)\cos[(\veps_2-\veps_1)\circ] + C_1^2\veps_1\rho_1 + C_2^2\veps_2\rho_2 = -\rho\partial S.
\end{equation}
Both of the fluid energy fields have nonlocal conservation in space at all points of time:
\begin{equation}\hs{-2ex}\label{2205} 
\intf{}\Etheta\rho = \intf{}\partial\rho = 0 \quad \text{and}\quad \intf{}\ES\rho = \intf{}\rho\partial S = C_1^2\veps_1 + C_2^2\veps_2.
\end{equation}
\end{theorem}
\textbf{\textit{Proof}} 
In the following equation sequence, we start with the linear combination wavefunction $\Psi$,
and calculate the left-hand side of Eq.~(\ref{5405}), for an $n$-body system:
\[\hspace{-4ex}\sPsi^* i\partial\sPsi =(C_1(\phi_1\psi_1)^* + C_2(\phi_2\psi_2)^*)(C_1\veps_1\phi_1\psi_1 + C_2\veps_2\phi_2\psi_2),\]
\[\hspace{-4ex}\sPsi^* i\partial\sPsi = C_1^2\veps_1\rho_1 + C_2^2\veps_2\rho_2
 + C_1C_2(\veps_1\phi_1\phi_2^*\psi_2^*\psi_1 + \veps_2\phi_2\phi_1^*\psi_1^*\psi_2),\]
\[\hspace{-4ex}\sPsi^*i\partial\sPsi = C_1^2\veps_1\rho_1 + C_2^2\veps_2\rho_2
  + C_1C_2\phi_1\phi_2\left(\veps_1e^{i(\veps_2-\veps_1)\circ} + \veps_2 e^{-i(\veps_2-\veps_1)\circ}\right).
\] 
Using (\ref{5405}) and Defs.~(\ref{en2}), we obtain energy Eq.~(\ref{2203}) and (\ref{2204}).
Since the functions $\phi_1$ and $\phi_2$ are orthonormal, the spacial energy averages, 
obtained by integrating (\ref{2203}) and (\ref{2204}) over all space $\mathbb{R}^{3n}$ and
summing over all spin variables, are given by Eqs.~(\ref{2205}).
\begin{flushright}\fbox{\phantom{\rule{0.5ex}{0.5ex}}}\end{flushright}
If we repeat the derivation in the proof with the following change $C_2 \to -C_2$, then the
sign of the $C_1C_2$ terms in Eq.~(\ref{2203}) and (\ref{2204}) change. If, instead, we make
the following change $C_2 \to iC_2$, nothing dramatic happens: The $\beta$ terms in
Eq.~(\ref{2204}), for $\ES$, and Eq.~(\ref{2203}), for $\Etheta$, switch places.  For the
general case, using a set of $m$ orthonormal spatial wavefunctions $\{\phi_1,\cdots \phi_m\}$
that are eigenfunctions of a single Hamiltonian operator, the same result is obtained: $\ES$
and $\Etheta$ are conserved over space. This result is easily seen to follow because, with no
loss of generality, the members of the set $\{\phi_1,\cdots \phi_m\}$ are mutually
orthogonal. The same result is obtained for a wavefunction defined by an infinite sequence of
orthonormal eigenfunctions, since each term of the corresponding sequences, for both $\ES$ and
$\Etheta$, are conserved, and such sets are complete.

\section{Energy Equation and Schr\"odinger Equation Equivalence \zlabel{4491}}

\subsection{Eqs.~(\ref{eneq}) and Schr\"odinger Eq.~(\ref{schrodinger}) equivalence}

Theorem~\ref{conti} proved that the imaginary part of the Schr\"odinger Eq.~(\ref{schrodinger})
is equivalent to five other equations, including the second energy Eq.~(\ref{eneq}):
($\Etheta\rho = \Pv$). The following proof demonstrates that the real part of the Schr\"odinger
Eq.~(\ref{schrodinger}) is equivalent to Eq.~(\ref{eneq}.1).  Since it is convenient to do so,
an alternate proof is given for the equivalence of the imaginary parts.  The theorem from
\ref{7825} proves that the two equations of Bohmian mechanics are equivalent to the
Schr\"odinger equation.
%
\begin{theorem}[Energy Eq.~and Schr\"odinger Eq.~Equivalence] \zlabel{theorem00}\mbox{} In the $L^2$
  Hilbert space, the real and imaginary parts of the Schr\"odinger Eq.~(\ref{schrodinger}) are
  equivalent to first and second equations of the 2-set Eqs.~(\ref{eneq}), respectively, where
  the wavefunction $\Psi$ satisfies~Def.~(\ref{def00}.1), the scalar fields $(\rho,S)$ are such
  that the velocities $\bu$ and $\bv$ satisfy Eqs.~(\ref{5204b}), the pressure $\Pu$ satisfies
  Eq.~(\ref{4720}.1), and $\ES$ and $\Etheta$ satisfy Eqs.~(\ref{en2}).

\end{theorem} 
\textbf{\textit{Proof}} Since the two equation $(A_1 = B_1)$ and $(A_2 = B_2)$ imply $(A_1 =
A_2)$ and $(B_1 = B_2)$ are equivalent, Eqs.~(\ref{def02}) and (\ref{enstrong}) imply that the
Schr\"odinger equation, in the form $(i\hbar\Psi^*\partial\Psi = \Psi^*\hat{H}\Psi)$, is
equivalent to
\[-\rho\partial S + i\zeta_0\partial\rho = \fc12 \rho_mu^2 + \Pu  + \fc12\rho_mv^2 + U\rho + i\Pv.\]
Substituting Eqs.~(\ref{en2}) gives the 2-set Eqs.~(\ref{eneq}), from the real and imaginary parts.
\begin{flushright}\fbox{\phantom{\rule{0.5ex}{0.5ex}}}\end{flushright}

\hs{-2ex}\mbox{\rm(\mbox{\it \bf Discussion about the energy Eq.~(\ref{eneq})})} For 1-body
states, Eq.~(\ref{eneq}) is the Bernoulli equation of fluid dynamics with $\rho_m$ being the
mass density for one-body states. However, this equations is being applied in a more general
way, since quantum flows are compressible and have variable mass.  Using (\ref{eneq}.1), flows
for one-body real valued wavefunctions have been investigated \cite{Finley-Arxiv,Finley-Bern},
where the generalized Bernoulli equation~(\ref{eneq}) describes compressible, irrotational,
steady flow with \emph{local} variable mass and $(\bv = \zero)$. Over all space, mass is
conserved, because the rate of mass creation from the sources are equal to the rate of mass
annihilation from the sinks. Also, each fluid element has a constant energy per mass
$\ES/m$. Having local variable mass is required for the flow of the 1s state of the hydrogen
atom to be steady with zero angular momentum, since such flows must be radially directed, and a
radially directed flow with conserved mass requires infinite mass.

As given above in Eq.~(\ref{4529}), $\langle\ES\rangle$ is the expectation value $\langle\hat{H}\rangle$ of
the Hamiltonian operator.  By considering the following equation for a one-body system with a
fluid interpretation, it is easily seen that the proper interpretation of the vector field
$\ES$ for these systems is that $\ES/m$ is the energy per mass:
\[
\langle\ES\rangle = \intf{}\; (\ES/m)\rho_m = \intf{}\; \ES\rho.  
\]
Similarly, $E/q$ can be interpreted as the energy per charge for a system with identical
particles, each with charge~$q$. Hence, $\ES$, despite having an energy unit, can be
interpreted to be the energy per mass and energy per charge that is multiplied by the total
mass and total charge, respectively, and these are intensive variables. However, for steady
states, where $\rho$ is normalized to unity, the field $\ES$ is uniform, giving
$(\ES=\langle\ES\rangle)$, where $\langle\ES\rangle$ is an extensive variable. Hence, $\ES$ can
be considered intensive or extensive, depending on the context.
For $n$-body systems, for proper scaling with respect to the number of particles, $\rho$ should
be normalized to $n$ instead of unity, and $E/m$ becomes $E/(nm)$.  If this is done, a similar
interpretation for the energy can be given as in the one-body case.  With the identification of
Eq.~(\ref{eneq}) as a generalization of the Bernoulli equation of fluid dynamics, the function
$(\Pu)$ is a correspondence variables from equation for the pressure, with an energy per volume
unit. Since, from Eq.~(\ref{4720}), it integrates to zero, from definition~3 of
Sec,~\ref{4930}, it a free, or distributive, variable, and a system has a total of zero $[\Pu]$
energy.  The presence of $[\Pu]$ as an energy does not require the input of energy, only the
separation of zero energy into positive and negative parts, where the energy is distributed
over space. For any change a state undergoes, the total free energy from $[\Pu]$ is conserved as
zero.

We also identify two velocity vectors $(\bu)$ and $(\bv)$ and their kinetic energy as
correspondence variables from equation (\ref{eneq}). However, this kinetic energy
identification is less than perfect, since the formula for the kinetic energy,
Eq.~(\ref{5208}), involves the square of the complex valued vector-field $(\bv + i\bu)$,
instead of the square of a real valued vector-field, as is done is classical mechanics.  Hence,
the kinetic energy definition has a non classical element for certain complex valued
wavefunctions. While some non classical elements can be tolerated, we wish to minimize such
elements. In the next subsection an energy equation is derived where the kinetic energy is
calculated from the singe velocity vector $(\bw = \bu + \bv)$, instead of a sum of two kinetic
energies from the 2--set $(\bv,\bu)$.

\subsection{Single velocity vector energy equation }

A way to remove the nonclassical kinetic energy problem, mentioned above, is given below by
Theorem~\ref{theorem22b}, where a single velocity vector $\bw$ is defined to replace the 2--set
$(\bv,\bu)$.  The alternate energy equation also introduces a new term that is proportional to
the volumetric dilatation rate $(\nabla\cdot\bv)$. The 2--set of energies $(\ES,\Etheta)$ are
also combined into one energy $\bar{E}$. A discussion about the energy equation is also given.
The following definition and lemma are used for the theorem below.
\begin{definition} 
  \mbox{}
\begin{equation} \zlabel{eta}
  \bar{\ES} \defi \ES + \Etheta = -\partial (S + \theta), \quad  \bw \defi \bu + \bv.
   \end{equation}
\end{definition}
%
\begin{lemma} \zlabel{5282} \mbox{}  
\begin{equation}\hs{-8ex} \zlabel{1222}
\text{Defintion~(\ref{4720}.1) for $\Pu$ and  (\ref{5204b}) for $\bu$ }  
\imply  \Pu = - \rho_m\bu\cdot\bu + \eta\nabla\cdot\bu
\end{equation}
  \begin{equation} \hs{-8ex} \zlabel{1221}
\text{Defintion~(\ref{4720}.2) for \,\,$\Pv$ and  (\ref{5204b}) for $\bu$ }
 \imply \Pv = \hspace{2ex}\rho_m\bu\cdot\bv - \eta\nabla\cdot\bv.
  \end{equation}
\end{lemma}
\textbf{\textit{Proof}} Using definition~(\ref{4720}) for $\Pu$ and $\Pv$, respectively,
definition above for $\eta$, and (\ref{5204b}) for $\bu$, we have
\[\hs{-10ex}\Pu = \hs{1ex}\zeta_0\nabla\cdot(\rho\bu) = \hs{2ex}\zeta_0\rho\nabla\cdot\bu
+ \zeta_0\nabla\rho\cdot\bu = \hs{1.5ex}\eta\nabla\cdot\bu - \rho_m\bu\cdot\bu, \lefteqn{\text{ and}} \]
\[\hs{-10ex}\Pv = -\zeta_0\nabla\cdot(\rho\bv) = -\zeta_0\rho\nabla\cdot\bv
-\zeta_0\nabla\rho\cdot\bv = -\eta\nabla\cdot\bv + \rho_m\bu\cdot\bv.\]
\begin{flushright}\fbox{\phantom{\rule{0.5ex}{0.5ex}}}\end{flushright}
%
%
\begin{theorem}[Single Velocity Total Energy Representation \label{theorem22b}]
The following equations are equivalent to the equation 2-set (\ref{eneq}), and, therefore,
also the Sch\"odinger equation from theorem~\ref{theorem00}:
\begin{equation}\hs{-8ex}\zlabel{eneq2b}
  \bar{E}\rho= \fc12\rho_m w^2 + \Pu  - \eta(\nabla\cdot\bv) + U\rho, \quad \Etheta\rho = \Pv.
\end{equation}
\end{theorem}
\textbf{\textit{Proof}} \mbox{}
Using the definition for $(\bw)$ from Eqs.~(\ref{eta}) and Def.~(\ref{1221}),
we discover that
\[\hs{-12ex}\fc12\rho_m w^2 - \eta(\nabla\cdot\bv) = \fc12\rho_m u^2 + \fc12\rho_m v^2 + \rho_m\bu\cdot\bv - \eta(\nabla\cdot\bv)
= \fc12\rho_m u^2 + \fc12\rho_m v^2 + \Pv.\]
Substituting this result into Eq.~(\ref{eneq2b}) and using ($\Etheta\rho = \Pv$), we find that
the resulting equation is the first equation of (\ref{eneq}) with $\Etheta\rho$ added to both
sides. Hence, this resulting equation is equivalent the first equation of the
2--set~(\ref{eneq2b}). (The first equations of (\ref{eneq}) and (\ref{eneq2b}) are not
equivalent.)

\begin{flushright}\fbox{\phantom{\rule{0.5ex}{0.5ex}}}\end{flushright}
\mbox{\rm(\mbox{\it \bf Discussion about the energy Eq.~(\ref{eneq2b})})} Besides combing the
2-set of velocities $(\bu,\bv)$ vectors into one, Eq.~(\ref{eneq2b}) combines the two energies
into one $\bar{\ES}$. From Eq.~$(\ref{eneq2b}.2)$ and Eq.~(\ref{4730}), $\Etheta$ is a free
particle energy. Hence, while $\Etheta$ can change the energy locally, it has no effect on the
total energy. Note that because the proof uses ($\Etheta\rho = \Pv$), Eq.~(\ref{eneq2b}.1), by
itself, is not equivalent to the real part of the Sch\"odinger equation.

\section{Euler fluid equations for one-body states \zlabel{2283}} 

\begin{definition}[Quantum flow of one-body states]
  A flow $(\mathcal{F} \defi (p,\phi,\mathbf{\mx})\hs{0,1ex})$ is a set containing a pressure
  $p$, mass density $\phi$ and velocity $\mathbf{\mx}$, with a given body force. A quantum flow
  $\mathcal{F}(\Pu,\rho_m,\bu,\bv)$, of a 1-body state, is a flow such that there exists a pair
  of scalar fields $(\rho,S)$, where the velocities $\bu$ and $\bv$ satisfy Eqs.~(\ref{5204b}),
  and the pressure $\Pu$ satisfies Eq.~(\ref{4720}.1). Also, $\ES$, $\Etheta$ and $\Pv$ use the
  definitions above. If there is no chance of confusion, some or all of the variables $\rho_m$,
  $\Pu$, $\bu$, and $\bv$ in the notation $\mathcal{F}(\Pu,\rho_m,\bu,\bv)$ can be suppressed,
  and a body force can be added.
  
  An improper steady flow is a flow where $(\partial\rho = \zero)$ and the corresponding
  wavefunction is not a stationary state. Since, we will not consider such flows, given a pair
  of scalar fields $(\rho,S)$ determining a wavefunction $\Psi$ and flow $\mathcal{F}$, the
  following statements are considered equivalent: $(\partial\rho = \zero)$, the flow $\mathcal{F}$ is
  steady, and the wavefunction $\Psi$ is a stationary state. For such a case, the members of
  the continuity set $\mathbb{C}$ of Theorem~\ref{conti} vanish.
  
    A smooth flow is a flow $\mathcal{F}(\Pu,\rho_m,\bu,\bv)$ such that the sequence $(\nabla
    S\cdot\nabla\rho)$, or, equivalently $(\bu\cdot\bv = \zero)$, vanish identically. As in
    Sec.~(\ref{orthc}), the equation $(\nabla S\cdot\nabla\rho =\zero)$ is called the
    orthogonal condition.
\end{definition}
For smooth flow stationary states, including one-electron atoms, giving $(\nabla\rho\cdot\bv =
\zero)$, it follows from continuity Eq.~(\ref{cont}.2) that $\bv$ is solenoidal, also $\Etheta$
vanish from Theorem~\ref{conti}. Hence, in this case, Eqs.~(\ref{eneq2b}) reduces to (\ref{eneq}).

\subsection{Euler fluid equations of motion for one-body states \zlabel{p1830}}

In this subsection, for one-body states, four Euler equations are derived that are equivalent
to energy equations~(\ref{eneq}.1). The flow type identified is compressible, irrotational flow with a
velocity field $\bw$ that is a sum of two velocity fields, $\bu$ and $\bv$.  The velocity
fields $\bu$ and $\bv$ are not, in general, perpendicular. The $\bu$ velocity does not satisfy
a continuity equation.  For stationary states, the flow is inviscid.  Of the four derived time
dependent equations, only one has the force $-\nabla\Pv$, and one has bulk a viscosity
$\nabla[\eta(\nabla\cdot\bv])$ involving velocity $\bv$.  Three out of the four Euler equations
are missing one or more terms involving both velocity fields, $\bv\nabla\cdot(\rho_m\bu)$ being
an example. All terms of the Euler equation are present for stationary states. Interpretations
of the results is given at the end of the section.


The next lemma, that can be proved by inspection, is needed for theorems that follow.
The lemma uses definitions~(\ref{5204b}), (\ref{4720}), and (\ref{4725}).
%
\begin{lemma}[Mass flux equations] \mbox{} 
\begin{equation}\hs{-10ex}\label{5501}
\nabla\cdot(\rho_m\bu) = \hs{1.65ex}\fc{1}{\zeta}\Pu \quad\text{and}\quad \bu = \!-\zeta\fc{\nabla\rho}{\rho}\hs{0.5ex}
\imply \nabla\cdot(\rho_m\bu)\bu = -\fc{\Pu}{\rho}\nabla\rho = \fc{1}{\zeta}\Pu\bu,
\end{equation}
\begin{equation}\hs{-10ex}\label{5502}
 \nabla\cdot(\rho_m\bv) = -\fc{1}{\zeta}\Pv \quad\text{and}\quad\hs{0.5ex} \bv = \hs{1.4ex}\zeta\fc{\mathbf{Q}}{\rho}\hs{1.7ex}
\imply \hs{0.7ex}\nabla\cdot(\rho_m\bv)\bv = \hs{0.5ex} -\fc{\Pv}{\rho}\hs{0.25ex}\mathbf{Q} = -\fc{1}{\zeta}\Pv\bv,
\end{equation}
\begin{equation}\hs{-10ex}\label{5503}
\nabla\cdot(\rho_m\bu) = \hs{1.5ex} \fc{1}{\zeta}\Pu \quad\text{and}\quad \bv = \hs{1.5ex}\zeta\fc{\mathbf{Q}}{\rho}\hs{1.7ex}
\imply \hs{0.7ex}\nabla\cdot(\rho_m\bu)\bv = \hs{1.5ex} \fc{\Pu}{\rho}\hs{0.1ex}\mathbf{Q} \hs{1ex} = \hs{1ex} \fc{1}{\zeta}\Pu\bv,
\end{equation}
\begin{equation} \hs{-10ex}\label{5504}
\nabla\cdot(\rho_m\bv) = -\fc{1}{\zeta}\Pv \quad\text{and}\quad \hs{0.5ex} \bu = -\zeta\fc{\nabla\rho}{\rho}
\imply \hs{0.7ex}\nabla\cdot(\rho_m\bv)\bu = \hs{1.5ex}\fc{\Pv}{\rho} \nabla\rho= -\fc{1}{\zeta}\Pv\bu.
\end{equation}
\end{lemma}

In the derivation of differential forms of the equations of motion of fluids, the special case
of conserved mass is assumed \cite{Munson,Shapiro,Kelly}. Since we have variable mass flows, we
need to make sure we have the correct momentum time-derivative. This is disposed of in the next
Lemma. The Lemma demonstrates that if mass is not conserved for a state with velocity
$\bw$, then the total momentum time-derivative has the additional term
$\bw(\hs{0.1ex}\partial\rho_m + \nabla\cdot(\rho\bw)\hs{0.15ex})$.  If $\bw$ is given by ($\bw
= \bu + \bv)$, and $\bw(\hs{0.1ex}\partial\rho_m + \nabla\cdot(\rho\bv)\hs{0.15ex})$ vanish,
then the additional term reduces to $\bw[\partial\rho_m + \nabla\cdot(\rho\bu)]$.

\begin{lemma}[Momentum time derivatives] \zlabel{tmom} \mbox{} 
  
\noindent
Let $\mathbb{V}$ be a time dependent subset of space $\mathbb{R}^3$ with surface $\mathcal{S}$,
and occupied by a fluid with variable mass. Let the velocity field of the fluid be
$(\mathbf{\bw}: \mathbb{R}^3\times\hat{\mathbb{R}} \rightarrow \mathbb{R}^3)$, where
$\mrr\in\mathbb{R}^3$ and $t\in\hat{\mathbb{R}}$ are spatial and temporal variables,
respectively.  Let the path of a fluid particle be $(\mrr = \mrr(t))$, giving the composition
$(\mathbf{\bw} = \mathbf{\bw}(\mrr(t)\hs{0.15ex})$, thereby, permitting total time
derivatives. Let $\mathcal{D}$ be the total time derivative, that is expressed as the
material derivative for steady flow when the chain rule is used. Let the tensor
$\nabla\mathbf{A}\hs{-0.4ex}\cdot\mathbf{B}$ be
\[\hs{-10ex} (\nabla\mathbf{A}\hs{-0.4ex}\cdot\mathbf{B})_i \defi \nabla A_i\hs{-0.4ex}\cdot\mathbf{B}
 = \mathbf{B}\cdot\nabla A_i,\quad i = 1,2,3,\quad
\qquad \mathbf{A}:\mathbb{R}^3 \rightarrow \mathbb{R},
\quad \mathbf{B}:\mathbb{R}^3 \rightarrow \mathbb{R},\]
giving the spacial part of the material derivative $\nabla\bw\hs{-0.4ex}\cdot\bw$ as a special case.
\vspace{1ex}
\begin{equation}\hs{-7ex} \zlabel{4125}
  \mathcal{D}\!\!\int_\mathbb{V}\mathbf{Y} = \int_\mathbb{V} \mathcal{D}\mathbf{Y} +  \mathbf{Y}(\nabla\cdot\bw);\quad
  \mathbf{Y}:\hat{\mathbb{R}}\longrightarrow \mathbb{R}^3,
\end{equation}
\begin{equation}\hs{-4ex} \zlabel{4792}
  \mathcal{D}\!\!\int_\mathbb{V}\rho = \int_\mathbb{V} \partial\rho + \nabla\cdot(\rho\bw),
\end{equation}
\begin{equation} \hs{-7ex}\zlabel{4125b}
  \mathcal{D}\!\!\int_\mathbb{V}\rho\bw = \int_\mathbb{V} \mathcal{D}(\rho\bw) +  \rho\bw(\nabla\cdot\bw),
\end{equation}
\begin{equation} \hs{-7ex}\zlabel{4125c}
  \mathcal{D}\!\!\int_\mathbb{V}\rho\bw = \int_\mathbb{V} \rho\mathcal{D}\bw,
  \quad\text{if}\hs{1.5ex} \mathcal{D}\!\!\int_\mathbb{V}\rho = \zero, \quad
  \text{and}
\end{equation}
\begin{equation} \hs{-10ex}\zlabel{4282b}
  \mathcal{D}\!\!\int_\mathbb{V} \rho\bw
  = \int_\mathbb{V} \; \partial(\rho\bw) + \rho\nabla\bw\hs{-0.4ex}\cdot\bw + \nabla\cdot(\rho\bw)\bw
  =  \int_\mathbb{V} \; \partial(\rho\bw) + \int_\mathcal{S}\rho\bw (\bw\cdot\hat{\mathbf{n}}). 
\end{equation}
If $\bw$ is irrotational, then
\begin{equation} \hs{-5ex}\zlabel{5940}
\mathcal{D}\!\!\int_\mathbb{V} \rho\bw =
\int_\mathbb{V} \; \partial(\rho\bw) + \fc12 \rho\nabla w^2  + \nabla\cdot(\rho\bw)\bw.
\end{equation}
Let $(\bw\defi\bu + \bv)$. If $\bw$ is irrotational and $(\partial\rho + \nabla\cdot(\rho\bv) = \zero)$, then
\begin{equation} \hs{-5ex}\label{5942}
\mathcal{D}\!\!\int_\mathbb{V} \rho\bw =
\int_\mathbb{V} \; \rho\partial\bw + \fc12 \rho\nabla w^2  + \nabla\cdot(\rho\bu)\bw.
\end{equation}
\end{lemma}
\textbf{\textit{Proof}} Let $\mathbb{V}: \hat{\mathbb{R}}\longrightarrow
\mathcal{P}(\mathbb{R}^3)$, where $\hat{\mathbb{R}}$ is temporal and
$\mathcal{P}(\mathbb{R}^3)$ is the power set of $\mathbb{R}^3$.  In other words, the value
$\mathbb{V}(t)$ is the Cartesian subspace $\mathbb{V}\in\mathbb{R}^3$ at time
$t\in\hat{\mathbb{R}}$.  Let $\bw$ be the velocity field of a compressible fluid that occupies
$\mathbb{V}(t)$ at time $t$. 
Let $\mathbf{Y}:\mathbb{R}^3\rightarrow \mathbb{R}$.
The time derivative of the $j$th Cauchy sum of the integral $\int_\mathbb{V}\mathbf{Y}$,
involving a partitioning into $m_j$ subsets: $(\mathbb{V} = \mathbb{V}_1\cup \mathbb{V}_2\cdots
\mathbb{V}_{m_j})$ with corresponding volumes $(V = V_1\cup V_2\cdots V_{m_j})$, is
\begin{equation} \label{4721}
\hs{-7ex}\mathcal{D}\!\sum_{i=1}^{m_j} \mathbf{Y}(\mrr_i) V_i =
\sum_{i=1}^{m_j}  \;\; [\mathcal{D}\mathbf{Y}(\mrr_i)]  V_i + \sum_{i=1}^{m_j}\mathbf{Y}(\mrr_i) \dot{V_i};
\qquad \dot{V}_i \defi \mathcal{D}V_i,
\end{equation}
where $\mathbf{Y}$ and $V_i$ are time dependent composites involving $(\mrr = \mrr(t))$. Also,
$\mrr_i\in \mathbb{V}_i$, and $m_j\longrightarrow \infty$ as $j\longrightarrow \infty$.  By the
definition of the divergence, in the limit of $j\longrightarrow \infty$, we have
\[\lim_{m_j\to\infty} \sum_{i=1}^{m_j}[\mathbf{Y}\dot{V}_i]\hs{0.3ex}(\mrr_i)
= \lim_{m_j\to\infty} \sum_{i=1}^{m_j}\left[\mathbf{Y}(\nabla\cdot\bw)V_i\right](\mrr_i),\]
where $\nabla\cdot\bw$ is the fluid volumetric dilation.  Because both functions have the same
limit, the rhs of this equation, without the limit, can replace the term on the rhs of
Eq.~(\ref{4721}), giving the same result for the limit. Hence
\begin{equation*}
\mathcal{D}\!\sum_{i=1}^{m_j} \mathbf{Y} V_i =
\sum_{i=1}^{m_j}  \;\; \left(\mathcal{D}(\mathbf{Y}) + \mathbf{Y}(\nabla\cdot\bw)\right)V_i.
\end{equation*}
By definition of Riemmann integration, we have Eq.~(\ref{4125}).

Substituting $(\rho\bw = \mathbf{Y})$ into Eq.~(\ref{4125}), we have 
\begin{equation*} \lab{4125b}
\mathcal{D}\!\!\int_\mathbb{V}\rho\bw = \int_\mathbb{V} \mathcal{D}(\rho\bw) +  \rho\bw(\nabla\cdot\bw).
\end{equation*}
Eq.~(\ref{4125c}) follows from the given mass conservation condition Eq.~(\ref{4792}), and
\[\mathcal{D}(\rho\bw) +  \rho\bw(\nabla\cdot\bw) =
\rho\mathcal{D}\bw + (\mathcal{D}\rho +  \rho(\nabla\cdot\bw)\hs{0.1ex})\bw.\]
Let $\bw$ be given by the composite $(\bw = \bw(t,\mrr(t))$; $(\hs{0.15ex}\mrr\in\mathbb{R}^3)$.
Using the chain rule, the total derivative of $\rho\bw$ is, by definition, the material
derivative of $\rho\bw$:
\[\hs{-12ex}\mathcal{D}(\rho\bw) = \rho\mathcal{D}\bw + \bw\mathcal{D}\rho
= \rho\partial\bw + \rho\nabla\bw\hs{-0.4ex}\cdot\bw + \bw\partial \rho + \bw\nabla\rho\cdot\bw.\]
\[\hs{-12ex}\lefteqn{\text{Hence,}} \hs{10ex} \mathcal{D}(\rho\bw) +  \rho\bw(\nabla\cdot\bw)
= \partial(\rho\bw) + \rho\nabla\bw\hs{-0.4ex}\cdot\bw + \nabla\cdot(\rho\bw)\bw.\]
Substituting into Eq.~(\ref{4125b}), we have (\ref{4282b}.1).
Equation.~(\ref{5940}) follows from the vector calculus equality
$(2\nabla\bw\hs{-0.4ex}\cdot\bw = \nabla w^2)$, for irrotational field $\bw$.
Equation.~(\ref{4792}), follows from Eq.~(\ref{4125}) with $\mathbf{Y}$ replaced by $\rho$.
Subtracting the continuity equation, multiplied by $\bw$: $(\bw\partial\rho_m +
\nabla\cdot(\rho_m\bv)\bw = \zero)$, from Eq.~(\ref{5940}), gives Eq.~(\ref{5942}).

Let $(\bp \defi \rho\bw)$.  For Eq.~(\ref{4282b}.2) we only need to to
prove that each of three components of the vector equation, given by
\[\hs{0ex} \int_\mathcal{S}\rho\bw (\bw\cdot\hat{\mathbf{n}}) =
\int_\mathbb{V} \; \rho\nabla\bw\hs{-0.4ex}\cdot\bw + \nabla\cdot(\rho\bw)\bw,\]
are equal:
\[\hs{-7ex}\left(\int_S\rho\bw (\bw\cdot\hat{\mathbf{n}})\right)_{\!\!i} = 
\int_{\mathcal{S}}p_i\bw\cdot\hat{\mathbf{n}}  = \int_{\mathbb{V}}\nabla\cdot(p_i\bw)
=\int_{\mathbb{V}}\nabla p_i\cdot\bw  + p_i\nabla\cdot\bw  \]
\[\hs{-7ex} =\int_{\mathbb{V}}\nabla(\rho w_i)\cdot\bw  + \rho w_i\nabla\cdot\bw
= \int_{\mathbb{V}}\rho  \nabla w_i\cdot\bw + w_i(\nabla\rho\cdot\bw) +  w_i(\rho\nabla\cdot\bw)\]
\[\hs{-7ex}= \int_{\mathbb{V}}\rho  \nabla w_i\cdot\bw  + \nabla\cdot(\rho_m\bw)w_i
= \left(\int_{\mathbb{V}}\rho\nabla\bw\hs{-0.4ex}\cdot\bw + \nabla\cdot(\rho\bw)\bw\right)_{\!\!i}.\]
\begin{flushright}\fbox{\phantom{\rule{0.5ex}{0.5ex}}}\end{flushright}
Next we obtain the first Euler equation of the set of four. This one is the simplest one. In
the special case of steady flow, the equation reduces to the complete Euler equation with
variable mass.
\begin{theorem}[Equation of motion for one-body states part 1] \zlabel{5293} \mbox{}
  
\noindent
For quantum flows $\mathcal{F}(\rho_m,\Pu,\bu,\bv)$ of one-body states, we have the following:
  
{\bf 1)} The first energy equation of (\ref{eneq}) is equivalent to the Euler equation:
\begin{equation} \hs{-5ex}  \zlabel{euler0}
\rho_m\partial\bv + \fc12\rho_m\nabla\left(u^2 + v^2\right)
+ \nabla\Pu + \nabla\cdot(\rho_m\bu)\bu + \rho\nabla U = \mathbf{0},
\end{equation}
that is, $\;\mathbf{(\ref{eneq}.1) \Leftrightarrow (\ref{euler0})}$. \hspace{1ex} {\bf 2)}
$\mathbf{(\ref{euler0}) \Leftrightarrow}$
$\text{Re}\Bigl(i\hbar\Psi^*\partial\Psi$ $= \Psi^*\hat{H}\Psi\Bigl)$.
\end{theorem}
\textbf{\textit{Proof}} Applying the gradient-sequence operation $(\nabla)$ to the energy
equation~(\ref{eneq})$\times\rho^{-1}$, we have the equation sequence:
\begin{equation} \zlabel{gradE}
\nabla\ES = \fc12m\nabla\left(u^2 + v^2\right) + \nabla\left(\Pu\rho^{-1}\right) + \nabla U.
\end{equation} 
Since this equation with a constant of integration is equivalent to Eq.~(\ref{eneq}.1) for
$\ES$, we only need to prove that $\mathbf{\underline{(\ref{gradE}) \Leftrightarrow
    (\ref{euler0})}}$ to prove part {\bf 1} of the theorem.
Multiplying Eq.~(\ref{gradE}) by $\rho$, followed by substituting of
Eq.~(\ref{4887}).1$\times\rho$: ($\rho\nabla\ES = \rho_m\partial\bv$), we have
\begin{equation} \hspace{-2ex} \zlabel{8888}
\rho_m\partial\bv + \fc12\rho_m\nabla\!\left(u^2 + v^2\right) + \rho\nabla\left(\Pu\rho^{-1}\right) + \rho\nabla U = \mathbf{0}.
\end{equation} 
Since this equation is equivalent to Eq.~(\ref{gradE}), we only need to prove that this
equation is equivalent to Eq.~(\ref{euler0}), that is,
$\mathbf{\underline{(\ref{8888}) \Leftrightarrow (\ref{euler0})}}$.
Substituting implication (\ref{5501}) into the identity
\begin{equation} \zlabel{2748}
  \rho\nabla\left(\fc{\Pu}{\rho}\right) = \nabla\Pu - \Pu\fc{\nabla\rho}{\rho},
\end{equation}
gives the following equality:
\begin{equation} \zlabel{2077}
  \rho\nabla\left(\Pu\rho^{-1}\right) = \nabla\Pu + \nabla\cdot(\rho_m\bu)\bu.
\end{equation}
Substituting this equation into the Eq.~(\ref{8888}), we obtain Eq,~(\ref{euler0}). Hence
$\mathbf{\underline{(\ref{8888}) \Leftrightarrow (\ref{euler0})}}$.

For part {\bf 2}, from Theorem~\ref{theorem00} we have $\Bigl(\text{Re}(i\hbar\Psi^*\partial\Psi$ $=
\Bigl.\text{Re}(\Psi^*\hat{H}\Psi)\Bigl)$ $\Leftrightarrow$ $\mathbf{(\ref{eneq}.1)}$ combining
this with part ${\bf 1\!:\,} \,\mathbf{(\ref{eneq}.1) \Leftrightarrow (\ref{euler0})}$, along
with the transitive and symmetric property of equivalence relations, gives statement {\bf
  2:\,}$\mathbf{(\ref{euler0}) \Leftrightarrow}$ $\Bigl(\text{Re}(i\hbar\Psi^*\partial\Psi\Bigl.$
  $= \Bigl.\text{Re}(\Psi^*\hat{H}\Psi)\Bigl)$.
\begin{flushright}\fbox{\phantom{\rule{0.5ex}{0.5ex}}}\end{flushright}
It follows from Lemma~\ref{tmom}, Eq.~(\ref{5940}), that the following terms are missing in
Euler Eq.~(\ref{euler0}): $\rho_m\partial\bu$, $\rho_m\nabla\bu\cdot\bv$ and
$\nabla\cdot(\rho_m\bu)\bv$. The next theorem obtains an equation of motion that improves the
situation by obtaining the first two. A second derived equation of motion includes the force
$-\nabla\Pv$, but at the cost of adding the term $\nabla\cdot(\rho_m\bv)\bu$.

\begin{theorem}[Equation of motion for one-body states part 2] \zlabel{5293b}
  Consider the following equations 
\begin{equation} \hs{-5ex} \zlabel{euler1}
    \rho_m\partial\bw + \fc12\rho_m\nabla w^2 
    + \nabla \Pu + \rho\nabla U + \nabla\cdot(\rho_m\bu)\bu
    = \eta\nabla\left(\nabla\cdot\bv\right).
\end{equation}
\begin{equation} \hs{-5ex} \zlabel{euler1d}
  \mathcal{D}\!\!\int_{\mathbb{V}} \rho_m\bw  - \int_{\mathbb{V}}\nabla\cdot(\rho_m\bu)\bv
  = \int_{\mathbb{V}} \; -\nabla \Pu - \rho\nabla U  + \eta\nabla\left(\nabla\cdot\bv\right).
\end{equation}
\mbox{}\vspace{-1ex}
\begin{equation}\hspace{-10ex} \zlabel{euler3} 
\partial(\rho_m\bu) + \rho_m\partial\bv + \fc12\rho_m\nabla\left(u^2 + v^2\right)
+ \nabla(\Pu + \Pv) + \nabla\cdot(\rho_m\bu)\bu + \rho\nabla U = \mathbf{0}.
\end{equation}
\[\hs{-10ex} -\int_{\mathbb{V}}\;\nabla\cdot(\rho_m\bv)\bu + \mathcal{D}\!\!\int_{\mathbb{V}} \rho_m\bw
- \int_{\mathbb{V}}\ \nabla\cdot(\rho_m\bu)\bv  + \rho_m\nabla(\bu\cdot\bv)\]
\begin{equation} \hs{17ex} \zlabel{euler3d}
  = \int_{\mathbb{V}} \;-\nabla \Pu -\nabla \Pv - \rho\nabla U  + \eta\nabla\left(\nabla\cdot\bv\right).
\end{equation}
For one-body quantum flows $\mathcal{F}(\rho_m,\Pu,\bu,\bv)$, under the condition of the
satisfaction of continuity Eq.~(\ref{cont}.2), the following three statements are equivalent:
{\bf 1)} Eq.~(\ref{euler0}), {\bf 2)} Eq.~(\ref{euler1}), and {\bf 3)} Eq.~(\ref{euler3}).
Also, Eqs.~(\ref{euler1}) and (\ref{euler1d}) are equivalent; Eqs.~(\ref{euler3}) and
(\ref{euler3d}) are equivalent.
\end{theorem}
\textbf{\textit{Proof}}
First we show that $\mathbf{\underline{(\ref{euler0})\Leftrightarrow
    (\ref{euler1})}}$, giving $\text{\bf 1} \Leftrightarrow \text{\bf 2}$.  Since, from
Theorem~\ref{conti}, Eq.~(\ref{eneq}).1:\hs{-1ex} $(\Etheta\rho = \Pv)$ is equivalent
to the continuity Eq.~(\ref{cont}) displayed above, implication~(\ref{1221})$\times\rho^{-1}$
holds with $\Pv\rho^{-1}$ replaced by $\Etheta$. Thus, continuity satisfaction implies
\[
\nabla\Etheta = m\nabla(\bu\cdot\bv) - \zeta_0\nabla(\nabla\cdot\bv).
\]
Using the Eq.~(\ref{4887}.2), $(-\nabla\Etheta = m\partial\bu)$, and definition $(\eta =
\zeta_0\rho)$, we can write this equation, after multiplication by $\rho$, as
\begin{equation} \zlabel{2815}
\rho_m\partial\bu + \rho_m\nabla(\bu\cdot\bv) = \eta\nabla(\nabla\cdot\bv).
\end{equation}
Adding this result to Eq.~(\ref{euler0}) gives (\ref{euler1}).
Hence
$\mathbf{\underline{(\ref{euler0}) \Leftrightarrow (\ref{euler1})}}$, giving
\raisebox{-0.1ex}{\fbox{$\text{\bf 1} \Leftrightarrow \text{\bf 2}.$}}
Next we show that $\mathbf{\underline{(\ref{euler0})\Leftrightarrow (\ref{euler3})}}$, giving
$\text{\bf 1} \Leftrightarrow \text{\bf 3}$.  from Theorem~\ref{conti}, the continuity 
Eq.~(\ref{cont}) displayed above is equivalent to $(\Pv = \zeta_0\partial\rho)$.
From logic~(\ref{8304}), we have Eq.~(\ref{7500}): $(\partial(\rho_m\bu)) + \nabla\Pv = \zero)$.
Adding this equation to Eq.~(\ref{euler0}), gives Eq.~(\ref{euler3}). Hence
\raisebox{-0.1ex}{\fbox{$\text{\bf 1} \Leftrightarrow \text{\bf 3}$}}.  By the transitive and
symmetric property of equivalence relations, we have
$\text{\bf 1} \Leftrightarrow \text{\bf 2}$ $\imply$
$\text{\bf 2} \Leftrightarrow \text{\bf 1}$ and
$\text{\bf 1} \Leftrightarrow \text{\bf 3}$ $\imply$
\raisebox{-0.1ex}{\fbox{$\text{\bf 2} \Leftrightarrow \text{\bf 3}$}}

It follows from Eq.~(\ref{5942}) that Eqs.~(\ref{euler1}) and (\ref{euler1d}) are
equivalent. For the same reason, Eqs.~(\ref{euler3}) and (\ref{euler3d}) are equivalent, after
making the substitution $(\bu\partial\rho_m = - \nabla(\rho_m\bv)\bu)$.
\begin{flushright}\fbox{\phantom{\rule{0.5ex}{0.5ex}}}\end{flushright}
For equation of motion~(\ref{euler1}), the only term missing is
$\nabla\cdot(\rho_m\bu)\bv$. The extra term $\eta\nabla(\nabla\cdot\bv)$ can be interpreted as
a sort of bulk viscosity, discussed below.
%

The next theorem, obtains all terms of the time derivative of the momentum. This is achieved
with additional forces.
\begin{theorem}[Equation of motion for one-body states part 3] \zlabel{5293c}
  Let the subspace $(\mathbb{V}\in\mathbb{R}^3)$ be occupied by a fluid particle that satisfies the
  Euler Eq.~(\ref{euler3}). Let $(\mathcal{P} \defim -\Pu)$.
\begin{equation} \hs{-11ex}\zlabel{0000}
  \mathcal{D}\!\!\int_\mathbb{V}\; \rho_m\bw  =
  \int_\mathbb{V}\; -\nabla \Pu - \rho\nabla U  + \nabla[\eta\left(\nabla\cdot\bv\right)]
+ \fc{1}{\zeta}\Bigl(\Pu\bv + \mathcal{P}\bu + (\bv\cdot\bu)\rho_m\bu\Bigl).
\end{equation}
\end{theorem}
\textbf{\textit{Proof}}
It follow from theorem~(\ref{5293c}) that we only need to prove that Eq.~(\ref{0000}) is equivalent to
Euler Eq.~(\ref{euler1}): 
\[\hs{-7ex}\rho_m\partial\bw + \fc12\rho_m\nabla w^2  + \nabla\cdot(\rho_m\bu)\bu
    = -\nabla \Pu - \rho\nabla U  + \eta\nabla\left(\nabla\cdot\bv\right).\]
For treating the term with factor $\eta$, consider the development: 
\[\hs{-12ex}\nabla\eta(\nabla\cdot\bv) = \zeta_0\nabla\rho\nabla\cdot\bv = -\rho_m\bu\nabla\cdot\bv
= m\bu(-\rho\nabla\cdot\bv) = m\bu\nabla\rho\cdot\bv - m\bu\nabla\cdot(\rho\bv),\]
giving
\begin{equation} \zlabel{2268}
 \nabla\eta(\nabla\cdot\bv) - m\bu\nabla\rho\cdot\bv + \nabla\cdot(\rho_m\bv)\bu = \zero.
\end{equation}  
The second term on the lhs satisfies: 
\[\hs{-10ex}-m\bu\nabla\rho\cdot\bv = (m\bu\rho)\fc{1}{\zeta}\left(\fc{-\zeta\nabla\rho}{\rho}\right)\cdot\bv
= (\rho_m\bu)\fc{1}{\zeta}(\bu\cdot\bv),\]
giving
\[ \nabla\eta(\nabla\cdot\bv) + \fc{1}{\zeta}(\bu\cdot\bv)\rho_m\bu + \nabla\cdot(\rho_m\bv)\bu = \zero.\]
Substituting equation~(\ref{5504}): $(\nabla\cdot(\rho_m\bv)\bu = -\zeta^{-1}\Pv\bu)$ and then
adding equation.~(\ref{5503}): ($\nabla\cdot(\rho_m\bu)\bv = \zeta^{-1}\Pu\bv$), we have
\begin{equation} \zlabel{0101}
\nabla\cdot(\rho_m\bu)\bv = \nabla\eta(\nabla\cdot\bv) +
\fc{1}{\zeta}\Bigl((\bu\cdot\bv)\rho_m\bu + \Pu\bv - \Pv\bu\Bigl).
\end{equation}
Substituting into Euler Eq.~(\ref{euler1}) and using the definition $(\mathcal{P} \defim -\Pu)$, we have
\[\hs{-7ex}\rho_m\partial\bw + \fc12\rho_m\nabla w^2  + \nabla\cdot(\rho_m\bu)\bw
=  -\nabla\Pu - \rho\nabla U + \nabla(\eta\nabla\cdot\bv) + \Gamma_\rho[\bu,\bv];\]
\[ \Gamma_\rho[\bu,\bv]= \fc{1}{\zeta}(\bu\cdot\bv)\rho_m\bu + \zeta^{-1}\Pu\bv + \zeta^{-1}\mathcal{P}\bu\]
Eq.~(\ref{0000}) follows from Eq.~(\ref{5942}).
\begin{flushright}\fbox{\phantom{\rule{0.5ex}{0.5ex}}}\end{flushright}

\subsection{Correspondence variables by Euler Equations \zlabel{2215}}

We identify, by definition~{\bf 2} of Sec~\ref{4930}, the mass density $\rho_m$, pressure
$\Pu$, and velocity $(\bu + \bv)$ as corresponding variables from Eq.~(\ref{0000}). The
pressure $\Pu$ is not a thermodynamic pressure, because, being a free energy, it is not
positive definite.
Euler Eq.~(\ref{0000}) is interpreted
as a generalization of the other three Euler equations, with Euler Eq.~(\ref{euler0}) being the
most specialized.  Consider equality~(\ref{0101}), if this equality is not satisfied, there can
be solutions of Eq.~(\ref{0000}) that do not satisfy the other Euler equations.  This
specialization, in the form of relationships between variables, explains why certain terms,
especially $\nabla\cdot(\rho_m\bu)\bv$, are missing in the specialized Euler equations: The
missing terms cancel.

For equation of motion~(\ref{euler3d}), the force $-\nabla\Pv$ is included, but at the cost of
introducing the term $-\nabla\cdot(\rho_m\bu)\bv$.  Because of this shortcoming, we lack
sufficient evidence to interpret $\Pu$ as a corresponding pressure. However, Eq.~(\ref{7500})
does identify the series $[\room\Pv\room]$ as a force.  The first definitions from
Eqs.~(\ref{4720}) and (\ref{4725}) suggests that the sign of $\Pv$ is wrong, indicating that
the definition $(\mathcal{P} = -\Pv)$ is more satisfying. If we interpret $\mathcal{P}$ as a
pressure, we can substitute $(\Pv = -\mathcal{P})$ into the gradient of Eq.~(\ref{1221}), giving
the possibility of finding an Euler equation with $-\nabla\mathcal{P}$ as a force.

The forces given by the last term on the rhs of Eq.~(\ref{0000}) are unknown.  These velocity
dependent forces are coupling forces involving the two velocity components.  They might be
associated with irreversible, unsteady flow; a sort of ``friction'' caused by nonsmooth flow.
However, states of wavefunctions that satisfy the space time \sch equation are closed, so if the flow
satisfies Euler Eq.~(\ref{euler0}), the total energy is conserved.
Since satisfaction of Eq.~(\ref{0000}) does not imply the satisfaction of the Schr\"odinger
equation, it might be applicable for mixed states, where the system of interest is part of a
larger closed system.

In order to interpret the term $\nabla(\eta\nabla\cdot\bv)$, consider the stress tensor
of the Navier--Stokes equation:
\[
\hs{-12ex}\sigma = -p\mathbf{I} +\lambda(\nabla\cdot\bu)\mathbf{I} + 2\mu(\grad\bu)_+;\quad
(\grad)_{ij} = \pa{u_i}{x_j}, \;\; (\grad\bu)_+ = \text{\bf sym}(\grad\bu),\]
\[\lefteqn{\hs{-5ex}\text{giving}}\hs{5ex} \sigma_{ij} = -p\delta_{ij}
+ \lambda\pa{u_k}{x_k}\delta_{ij} + \mu\left(\pa{u_i}{x_j} +\pa{u_j}{x_i}\right),\]
where $\mu:\mathbb{R}^3\rightarrow\mathbb{R}$ and $\lambda:\mathbb{R}^3\rightarrow\mathbb{R}$
are the dynamic- and second-viscosity coefficients, respectively.  Since the corresponding force is
\[\dive\sigma = -\nabla p + \nabla[\lambda\nabla\cdot\bu] + 2(\grad\bu)\nabla\mu +   \mu\nabla(\nabla\cdot\bu),\]
where $(\room\hs{0.1ex}(\dive\sigma)_i = \partial_1\sigma_{i1}/\partial x_1 + \partial_1\sigma_{i2}/\partial x_2
+ \partial_1\sigma_{i3}/\partial x_3$; $\mrr = (x_1,x_2,x_3)\room)$, and the implication
$(\nabla\times\bu = \zero \imply \mu\nabla^2\bu = \mu\nabla(\nabla\cdot\bu))$ is substituted,
$\nabla(\eta\nabla\cdot\bv)$ can only come from the diagonal elements of the stress tensor.
We now give two possible interpretations of
$\nabla(\eta\nabla\cdot\bv)$. One is obtained by choosing $(\mu = \zero)$, giving $(\eta =
\lambda)$. The second one is to take $\eta$ as the nonpressure part of the average normal stress:
\[\hs{-7ex} \fc13\text{tr} \sigma = \fc13\sigma_{ii} = -p + \lambda\pa{u_k}{x_k} + \fc23\mu\left(\pa{u_i}{x_i}\right)
  = -p + \left(\lambda+ \fc23\mu\right)\nabla\cdot \bu,\]
giving $(\eta = \lambda+ 2\mu\div 3)$, where $(\lambda+ 2\mu)$ is the bulk viscosity.  Since
this assignment is an approximation, where $2(\grad\bu)\nabla\mu$ is neglected, the first one
seems like a better interpretation.

Since $\eta$ is not a constant, the interpretation of $\eta\nabla(\nabla\cdot\bu)$ from
Eq.~(\ref{euler1}) is more difficult. Because of the conditions imposed on the stress
tensor---giving the forces of the Navier-Stokes equation---may not be applicable to quantum
flows, it seems reasonable to interpreted the function $\eta\nabla(\nabla\cdot\bu)$ as a sort
of bulk viscosity. So we do that: Both $\nabla(\eta\nabla\cdot\bu)$ and
$\eta\nabla(\nabla\cdot\bu)$ are interpreted as corresponding variable of bulk viscosity.

\subsection{Interpretations of energy and Euler equations \zlabel{2217}}

We have demonstrated that states of quantum mechanical satisfy energy equations and equations
of motion. We also have identified corresponding variables of velocity, energy, and force. This
is the easy part. The hard part is giving a complete description of states of quantum mechanics
that is consistent with these equations and variables.

While the equations derived have been identified as equations of fluid dynamics with variable
mass, one can imagine other possible descriptions that are consistent with these equations,
fluid pulsation and particle oscillation, being examples.
Also, given the independence of the two velocity vectors $\bu$ and $\bv$, and the the couplings
between them, the two velocities could be associated with different modes, or forms, of
energy. For example, Salesi \cite{Salesi}, using approximations of relativistic
quantum-mechanics, assigns the kinetic energy from the Madelung velocity $\bv$ to the center of
mass motion, and the velocity $\bu$ to the internal energy of the electron.

For a particle trajectory description, we take equation of motion~(\ref{8888})$\div\rho$, where, for
one-body states, the streamlines of the fluid are the possible paths for the particle. This
approach is Bohmmian mechanics with velocity $\bu$ and body force $\Pu\rho^{-1}$ replacing the
quantum potential.  The energy equation with the addition of pressure removes tunneling and
gives a constraint between the variables.  The addition of $\bu$ gives states with kinetic
energies that agrees with the corresponding expectation value. It also provides the necessary
additional orbital angular momentum, since, for example, velocity $\bv$ alone can only provide the
correct $z$ component of angular momentum of the electronic states of the hydrogen
atom. However, the trajectories with velocity $\bu$ are unstable, since they go off into
infinity with a velocity $\bu$ direction of $-\nabla\rho$.  Therefore, if the velocity
directions is not changed, as considered in the next section, the electron must frequently
jump from one path to another, or go back and forth, and such behavior, involving velocity
discontinuity, has been predicted to exist \cite{Oldofredi}. It is worth mentioning that such
jumps might be needed in standard quantum mechanics, if two systems in the same state are to be
identical over time, otherwise an electron would be stuck in a single orbit. Also, for the
hydrogen $s$ states with standard quantum mechanics, in order to satisfy the vanishing angular
momentum-requirement, the electron must stay in a straight radial line, going back and forth,
frequently changing to the opposite direction, and perhaps, go through the nucleus.

Since the $n$s wavefunctions of hydrogen are eigenfunctions of $\hat{L}^2$---the operator for
the square of the magnitude of the angular momentum---and the eigenvalues are zero, they must
have radially directed momentum, so that the angular momentum vector-field is the zero
function.  For a fluid description of the $n$s state of hydrogen, the vanishing angular
momentum condition is satisfied by velocity $\bu$ of Eq.~(\ref{5204b}).  Since the flows are
steady, the flows must have local variable mass, since a steady radial flow requires infinite
mass if mass is conserved. Over all space, the flows conserve mass, because the sources cancel
the sinks.  Anyone who works with quantum mechanics has to swallow ``quantum pills,'' tunneling
and the uncertainty principle being two examples. As quantum pills go, the local variable mass
pill is, perhaps, not too large too swallow.  (Fluid flows with velocity $\bu$ is explored elsewhere,
including the investigation of the hydrogen atom 1s and 2s states
\cite{Finley-Arxiv,Finley-Bern}.) 

While local variable mass and a jumping electron are possibilities, these are nonclassical
elements, and we wish investigate if these elements can be eliminated. The next section
investigates a way to modify the direction of quantum flows, and two of the applications in
doing so is for mass conservation for fluid flow and stable orbitals for trajectories.




\section{One-Body Cross Flows \zlabel{4938}}

In this section, for one-body states, we investigate ways to change velocity directions of
steady, quantum flows, without changing the kinetic or total energy, and conditions are given
for these flows to conserve mass and also to be incompressible along streamlines.  A special
type of flow, called cross flow, is explored, where a velocity field $\mmu$ is introduced that
replaces $\bu$, and $\mmu$ has the same speed as $\bu$.
In order to get an idea what we have in mind, consider the familiar orbital plot of the $3d_z$
electronic state of the hydrogen atom, containing one level surface of the probability
density. Imagine the torus filled with electron fluid undergoing steady flow. From symmetry,
the mass density and speed are constant, mass is conserved, and the flow is incompressible,
and, therefore, also solenoidal. Such incompressible flows have density stratification, where
the mass density is only constant of each streamline. Since velocity $\bv$ is
irrotational, if $\mmu$ is solenoidal, the two velocities can be considered components of a
velocity vector that can be expressed by a Helmolholtz decomposition $(\oomega = \nabla S\div
m + \nabla\times\mathbf{K})$, such that $(\mmu = \nabla\times\mathbf{K})$.

In Sec.~\ref{9918}, equations of motion are obtained for many-body states. Under certain
conditions, these equations are reduced to 1-body in Sec.~\ref{4918}, leading to orbital
equations. However, a difficulty arises in the integrations. This problem removed in
Sec.~\ref{7928} by a change in velocity direction to give incompressible flow.

Diverting velocity direction away from $-\nabla\rho$ might provide a
way to describe electron spin, given that electron spin could be from electron motion.
The following velocity field formula was obtained for spin one-half particles from
approximations of the Pauli equations by Holland \cite{Holland:99}, in spherical
coordinates:
\begin{equation} \zlabel{4261c}
  \mx_\pm\times\mathbf{\hat{z}} = \fc{\nabla\rho}{m\rho}\times\mathbf{s}_\mp
  = \bu\sin\theta\,\hat{\boldsymbol{\phi}},
  \qquad \mathbf{s}_\mp = \mp\fc{\hbar}{2}\mathbf{\hat{z}}.
\end{equation}
This formula was applied to states of the hydrogen atom by Colijn and Vrscay
\cite{Colin1,Colin2}. For fluid flow, the formula provides incompressible flow with mass
conservation and density stratification for the electronic states of the hydrogen atom. Because
of the $\sin\theta$ factor, for hydrogen states, formula~(\ref{4261c}) accounts for all the
kinetic energy for either a fluid particles or an orbiting electron when they are on the
$x$-$y$ plane containing the nucleus. In the assignment of spin angular momentum to
translation, a difficulty arises in how to describe both spin and orbital momentum with
translation in such a way that the two forms of momentum can be distinguished, and such that
the correct total kinetic energy is obtained.  For example, if all the kinetic energy is used
to describe electron spin without center of mass translation,
there will be nothing left to describe the $y$ and $z$ components of
orbital angular momentum for the hydrogen atom, since, as mentioned above, this portion of
the angular momentum cannot be described by velocity $\bv$ alone.
%

To get some idea about the special conditions that cross flows might satisfy, consider Euler
Eq.~(\ref{euler0}) for a steady, mass conserving-flow of classical mechanics, where $u^2$ is
replaced by $\mu^2$, $\bv\cdot\nabla\rho_m$ vanish, and, using the definition
$(\bu(\rho,\nabla\rho) =-\zeta\nabla\ln\rho)$, the mass flux-term $\nabla\cdot(\rho_m\bu)\bu$
is considered a functional of the density.  Also, the velocity vector is not necessarily
irrotational, giving
\[\rho_m\mathcal{D}\mmu + \nabla\Pu + [\nabla\cdot(\rho_m\bu)\bu](\rho,\nabla\rho) + \rho\nabla U = \mathbf{0},\]
where $\mathcal{D}\mmu$ is a material derivative. Furthermore, the only requirement we place on
$\mmu$ is $(\mu^2 = u^2)$.  Using the standard approach to integrate this equation of motion
\cite{Munson}, the dot product of the unit vector in the direction of $\mmu$ is taken.  In
order to obtain the resulting energy Eq.~(\ref{eneq}), where $\bu$ is replaced by $\mmu$, the
density functional $\nabla\cdot(\rho_m\bu)\bu$ must vanish. Hence, the direction of $\mmu$ must
satisfy $(\mmu\cdot\bu = \zero)$, giving constant density flow: $(\mmu\cdot\nabla\rho_m =
\zero)$. With this condition and mass conservation, we have $(\zero = \nabla\cdot(\rho_m\mmu) =
\rho_m\nabla\cdot\mmu)$, giving $(\nabla\cdot\mmu = \zero)$.  Hence, the flow is incompressible
\emph{on each streamline}. With $u^2$ replaced by $\mu^2$, and $(u^2 = \mu^2)$, energy
Eq.~(\ref{eneq}) is the Bernoulli equation, for the general case of density stratification, and
the special case, where the energy is uniform for the entire flow field.
This analysis suggests that the correct quantum flow, at least in this special case, is
incompressible, density stratified, mass conserving flow.
Next cross flow is defined. Under special conditions these flows are mass conserving and
solenoidal.
\begin{definition}[Cross Flow for one-body states]
  Let $(S,\rho)$ determine the quantum flow $\mathcal{F}(\bu,\bv)$ and the 1-body wavefunction
  $\Psi$. Let $\mmu$ be a vector field that satisfies $(\mmu\cdot\bu = \zero)$ and $|\mmu| =
  |\bu|$. Let $(\oomega \defi \mmu + \bv)$. The flow $\mathcal{F}_\times(\mmu,\bv)$ is called a
  cross flow of $\mathcal{F}$, where both flows have the same quantum parameters $\Pu$ and
  $\rho_m$. If $\mmu$ is solenoidal: $(\nabla\cdot\mmu = \zero)$, $\mathcal{F}_\times$ is said
  to be cross velocity solenoidal.  If $\mathcal{F}$ is smooth $(\nabla\rho\cdot \nabla S =
  \zero)$, $\mathcal{F}_\times$ is also said to be smooth.  The cross flow Euler equation for
  steady, smooth flow is Eq.~(\ref{euler0}) without the mass flux term
  $(\nabla\cdot(\rho_m\bu)\bu$, but they can also have additional nowork forces.  Since energy
  is invariant to velocity direction and the presence of nowork forces, a cross flow
  $\mathcal{F}_\times(\mmu,\bv)$ that satisfies Euler Eq.~(\ref{euler0}), also satisfy energy
  Eq.~(\ref{eneq}), with the same energy uniform $\ES$ as its corresponding flow
  $\mathcal{F}(\mmu,\bv)$.
%
\end{definition}
Given an ensemble of a quantum state with a distribution of initial configurations, distributed
according to the Born rule, the Born rule is preserved over time if the velocity vector is
$\bv$: The distribution of members of the ensemble satisfy the Born rule at all times
\cite{Towler}. For particle trajectories of one-body states, a cross flow velocity vector $\oomega$
also satisfies this condition, because the direction of the velocity vector component $\mmu$ is
always on a constant level surface of the probability density $\rho$, since
$\mmu\cdot\nabla\rho$ vanish.

For cross flows with vortexes, or orbiting electrons, there is an additional force.  For real
valued wavefunctions, in order to achieve the necessary radius of curvature, such a force must
act perpendicular to the velocity vector $\mmu$. The force also must do no work, or, at least,
be defined by a free energy-potential. Otherwise, such a force would contribute to the
energy. With the addition of a nowork force, the Euler equations above are only satisfied along
the velocity directions, but they still integrate to the same energy equation. For the hydrogen
s states, in atomic units, normal to the velocity direction, the force $\hat{F}$ must satisfy
\[\hs{-10ex}-\nabla U\cdot\hat{\mrr} -\rho^{-1}\nabla\Pu\cdot\hat{\mrr} + \hat{F}\cdot\hat{\mrr} = - \fc{\mu^2}{r^2}; \qquad
\mu^2 = u^2,\quad -\nabla U = -\fc{1}{r^2}\hat{\mathbf{\mrr}},\]
%
where $-\mu^2/r^2$ is the centrifugal acceleration.  For hydrogen 1s fluid flow, the speed has
a constant value of 1 a.u, the same as in the Bohr model, and the pressure satisfies
%
$(-\nabla\Pu\cdot\hat{\mrr}/\rho = r^{-2} + 2r^{-1} - 2)$ \cite{Finley-Bern}. An orbiting
electron-treatment of the 1s electron, where the force $\hat{F}$ is left out, gives a modified
Bohr model, where the radius of the sole stable orbit is 3/2 $a_0$, compared to 1 $a_0$ for the
Bohr atom.  For excited states, we expect that this model will give more then one orbit.  With
the force $\hat{F}$ included, in order to have a stable orbit at the Bohr radius, the force
$\hat{F}$ must be $-\hat{\mathbf{\mrr}}$ a.u. For fluid flow, the force must be such that the
Euler equations are satisfied perpendicular to the streamlines over the entire flow field.





%
The next lemma proves that if the velocity $\bu$ of a flow $\mathcal{F}(\rho_m,\Pu,\bu,\bv)$ is
replaced by velocity $\mmu$, such that the direction of $\mmu$ is $\bu\times\bv$, then $\mmu$
is solenoidal. The theorem is applicable to any two irrotational vector fields.
\begin{lemma}[The divergence of the cross product of irrotational fields] \zlabel{3739} \mbox{}

  \noindent
  Let $(\mmu \defi \mathbf{x}\times\bu)$, where $\mx$ and $\bu$ are irrotational vector
  fields. The vector field $\mmu$ is solenoidal. 
\end{lemma}
\textbf{\textit{Proof}} Let the scalar field $\theta$ and $\theta_\times$ be such that $(\bu =
\nabla\theta)$ and $(\mathbf{x} = \nabla\theta_\times)$, respectively. Using the identity
\[\nabla\cdot\left(\mathbf{A}\times\mathbf{B} \right) = \mathbf{B}\cdot(\nabla\times \mathbf{A}) - \mathbf{A}\cdot(\nabla\times \mathbf{B}), \]
and the vanishing of the curl of the gradient, we have,
\[\nabla\cdot\mmu =  \nabla\cdot(\mathbf{x}\times\bu) = \nabla\cdot(\nabla\theta_\times\times\nabla\theta) = \zero. \]
\begin{flushright}\fbox{\phantom{\rule{0.5ex}{0.5ex}}}\end{flushright}
%
The following theorem demonstrates that smooth cross flows have constant density:
$(\oomega\cdot\nabla\rho = \zero)$; $(\oomega \defi \mmu + \bv)$. Also, $(\nabla\cdot\mmu=
\zero)$ implies mass conservation, and the previous lemma demonstrates one way to define the
vector field $\mmu$ so that it is solenoidal. Also, conditions for incompressible flow are also
given.  Since it is easy to overlook, it is worth mentioning that steady, constant density-flow
does not imply mass conservation.
\begin{theorem} \zlabel{2629}
  Let $\mathcal{F}_\times$ be a cross flow of $\mathcal{F}$.

  \noindent
  {\bf 1)} If $\mathcal{F}$ is smooth, each streamline of $\mathcal{F}_\times$ has a constant
  density.

  \noindent
  {\bf 2)} The following statements are equivalent: $\mathcal{F}_\times$ is cross velocity
  solenoidal, $\rho\mmu$ is solenoidal, $\mathcal{F}_\times$ is mass conserving.
  
  \noindent
  If $\mathcal{F}_\times$ is cross velocity solenoidal, then vector field $\oomega$ can be
  represented by the Helmholtz decomposition:
  \begin{equation}\zlabel{2792}
    m\oomega = \nabla S + \nabla\times\mathbf{K}, \qquad \text{\boldmath $\mu$} =
    \nabla\times\mathbf{K}.
  \end{equation}

  \noindent
  {\bf 3)} If $\mathcal{F}_\times$ is cross velocity solenoidal, smooth and steady,
  $\oomega$ is incompressible.

  \noindent
  Here is a summary: \vspace{2ex}

  {\bf 1)} $\mathcal{F}_\times$ is smooth $\imply$   $(\nabla\rho_m\cdot\oomega = \zero)$.
  
  {\bf 2)} $(\nabla\cdot\mmu = 0)$ $\Leftrightarrow$ $(\nabla\cdot(\rho_m\mmu) = \zero)$  $\Leftrightarrow$ $(\partial\rho
  + \nabla\cdot(\rho_m\oomega) = \zero)$ $\imply$ Eq.~(\ref{2792}).  

  {\bf 3)} $(\nabla\cdot\mmu = 0)$ and smooth and steady $\mathcal{F}_\times$ $\imply$ $(\nabla\cdot\oomega = \zero)$.
 \end{theorem}
\textbf{\textit{Proof}}

\noindent
 {\bf 1)} For a cross flow $\mathcal{F}_\times$ from a corresponding smooth flow $\mathcal{F}$, we have
 $(\bu\cdot\mmu = \zero \imply \nabla\rho_m\cdot\mmu = \zero)$ and $(\nabla\rho\cdot\nabla S
 = \zero \imply \nabla\rho\cdot\bv= \zero)$, respectively, giving
 a constant density cross flow: $(\nabla\rho\cdot\oomega= \zero)$.\vspace{1ex}

\noindent
 {\bf 2)} Combining the solenoidal requirement $(\nabla\cdot\mmu = \zero)$ with the cross flow
 requirement $(\nabla\rho_m\cdot\mmu = \zero)$, we have $(\nabla\cdot(\rho_m\mmu) =
 \zero)$. Hence, $\rho\mmu$ is solenoidal. The logic:
\begin{equation}\hs{-5ex} \zlabel{3738} 
  \nabla\cdot(\rho_m\mmu) = \zero
  \text{ $+$ Continuity Eq.}~(\ref{cont}.2) \imply \partial\rho + \nabla\cdot(\rho_m\text{\boldmath $\omega$}) = \zero,
  \end{equation}
then indicates that mass is conserved.
%
Equation~(\ref{2792}) holds because a solenoidal vector field $\mmu$ can be written as the curl
of vector potential $\mathbf{K}$. By reversing the above two implications, $\mathcal{F}_\times$
being mass conserving implies it is cross velocity solenoidal.
Equation~(\ref{2792}) holds by the Helmholtz decomposition theorem. \vspace{1ex}

\noindent
    {\bf 3)} For smooth, steady flow the continuity Eq.~(\ref{cont}.2) reduces to
    $(\nabla\cdot\bv = \zero)$.  Combining this result with the cross velocity solenoidal property
    $(\nabla\cdot\mmu = \zero)$, we have $(\nabla\cdot\oomega = \zero)$.
\begin{flushright}\fbox{\phantom{\rule{0.5ex}{0.5ex}}}\end{flushright}
The following corollary uses information from lemma~\ref{3739}: $(\nabla\cdot\mmu = \zero)$,
Eq.~(\ref{5240}) for steady smooth flow: $(\nabla\cdot\bv = \zero)$, and Theorem~\ref{2629}.
\begin{corollary} \zlabel{2744}
  Steady $(\partial\rho = \zero)$ smooth $(\nabla\rho\cdot\nabla S = 0)$ cross flows
  $\mathcal{F}_\times(\mmu,\bv)$ with $(\mmu = \nabla S/|\nabla S|\times \bu)$ are solenoidal
  $(\nabla\cdot(\mmu +\bv) = \zero)$, conserves mass $(\nabla\cdot\rho_m(\mmu + \bv) = \zero)$,
  incompressible $(\nabla\cdot(\mmu + \bv) = \zero)$, and have a constant density $(\rho\vert_s
  = \text{const.})$ on each streamline $s$, if $\nabla S$ does not vanish on positive measure.
\end{corollary}
%
%

%

The next theorem provides conditions for smooth, cross flows to satisfy that imply that they
satisfy energy equation~(\ref{eneq}.1). From Theorem~\ref{2629}, these flows have a constant
density on each streamline, but the density depends on the streamline.  Let
$\mathcal{F}(\Pu,\rho_m,\bu,\bv)$ be a smooth flow, whose wavefunction satisfies the spatial
\sch equation.  For the next theorem, it is demonstrated that if there exists a cross flow
$\mathcal{F}_\times(\Pu,\rho_m,\mmu,\bv)$, corresponding to smooth flow
$\mathcal{F}(\Pu,\rho_m,\bu,\bv)$, that satisfies the equation of motion
\begin{equation} \zlabel{1895c}
  \rho_m \mathcal{D}\oomega + \nabla\Pu + \rho\nabla U = \rho\vec{F};\qquad  \oomega = \mmu + \bv,
\end{equation}
with \emph{any} nowork force $\vec{F}$, the cross flow $\mathcal{F}_\times$ satisfies the
same energy Eq.~(\ref{eneq}.1) as the smooth flow $\mathcal{F}$. However, the theorem given is
slightly more general, where the wavefunction $\Psi$ of the smooth flow $\mathcal{F}$ is not
necessarily an eigenfunction of the \sch equation. Section~\ref{7928} demonstrates that such
cross flows for many body states can yield reduced 1-body Euler equations.
\begin{theorem}[Smooth Cross Flow Implications] \zlabel{5084} 
  Let $S$ and $\rho$ be time-independent 1-body fields, giving the wavefunction $\Psi$ and
  smooth, steady flow $\mathcal{F}(\Pu,\rho_m,\bu,\bv)$. However, $\Psi$ is not necessarily an
  eigenvector of Hamiltonian $H$ with external potential $U$. Let $\mathcal{C}$ be the set of
  all smooth, cross flows $\mathcal{F}_\times(\Pu,\rho_m,\mmu,\bv)$ of $\mathcal{F}$ that
  satisfy Euler Eq.~(\ref{1895c}) with a nonwork force $\vec{F}$, where $\vec{F}$ can depend on
  the cross flow.
\vspace{2ex}

\noindent
The following hold if $\mathcal{C}$ in nonempty. \vspace{1ex}

\noindent
{\bf 1)} The flow $\mathcal{F}$ and each and every
$\mathcal{F}_\times\in\mathcal{C}$ satisfy the energy Eq.~(\ref{eneq}.1), where the energy scalar
field $\ES$ may depend on the streamline.  \vspace{1ex}

\noindent
{\bf 2)} Under the condition of continuity Eq.~(\ref{cont}.2) satisfaction for flow $\mathcal{F}$, the
following two statements are equivalent: \vspace{1ex}

{\bf a)} $\Psi$ is an eigenvector of the TISE. \vspace{1ex}
{\bf b)} The energy scalar field $\ES$ is uniform. 
%

\noindent
{\bf 3)} Eq.~(\ref{1895c}) is an Euler equation for the cross flows
$\mathcal{F}_\times(\Pu,\rho_m,\mmu,\bv)\in\mathcal{C}$, having the proper time derivative, if
and only if mass is conserved.


\end{theorem}
\textbf{\textit{Proof}} \mbox{}

\noindent
{\bf 1)} Let $\mathcal{F}_\times$ satisfy Eq.~(\ref{1895c}). Note that
\begin{equation}\zlabel{8384}
  \rho_m\mathcal{D}\oomega = \fc12\rho_m\nabla\omega^2 - \oomega\times(\nabla\times\oomega),
\end{equation}
where, by definition of the cross product $\times$, the second term on the rhs, a vector field,
is perpendicular to the streamline---the direction of $\oomega$.  Let $\mathbf{\hat{e}}$ be a
unit vector tangent to the streamline, that is, in the direction of $\pm\oomega$. Hence,
$(\oomega\times(\nabla\times\oomega)\cdot\mathbf{\hat{e}} = \zero)$.  Next we perform
the following usual procedure for the fields restricted to a streamline:
\begin{quote}
1) Substituting Eq.~(\ref{8384}) into Eq.~(\ref{1895c}).

2) Take the dot product of the resulting equation and $\mathbf{\hat{e}}$.

3) Use $(\vec{F}\cdot\mathbf{\hat{e}} = \zero)$---that holds because $\vec{F}$ does no work.

4) Use $(\oomega\times(\nabla\times\oomega)\cdot\mathbf{\hat{e}} = \zero)$---that holds by definition of $\mathbf{\hat{e}}$.
\end{quote}
The procedure gives 
\begin{equation} 
\fc12\rho_m\nabla\hspace{0.1ex}\mu^2\cdot\mathbf{\hat{e}} + \fc12\rho_m\nabla\hspace{0.1ex}v^2\cdot\mathbf{\hat{e}} 
       + \nabla\Pu\cdot\mathbf{\hat{e}} +  \rho_m\nabla U\cdot\mathbf{\hat{e}} = \zero.
\end{equation}
From Theorem~\ref{2629}, for smooth, cross flows, the density is constant on each streamline, giving
$(d(\rho_m\oomega)/dt = \rho_md\oomega/dt)$. Using this result, the differential form of the
above equation is 
\begin{equation} \zlabel{2188}
\fc12\, d(\rho_m\mu^2) + \fc12\, d(\rho_mv^2)  + d\Pu +  \rho\, dU = \zero.
\end{equation}
Since the density $\rho_m$ is constant on each streamline, the integration of Eq.~(\ref{2188})
gives energy Eq.~(\ref{eneq}.1) with velocity vector field $\mmu$ replacing $\bu$. Hence,
$\mathcal{F}_\times$ satisfies Eq.~(\ref{eneq}.1). Since the path integration of the above
equation is restricted to a streamline, the energy is only necessarily constant on each
streamline.  Since $(u^2 = \mu^2)$, and both flows share the same density $\rho$ and pressure
$\Pu$, $\mathcal{F}$ also satisfies energy Eq.~(\ref{eneq}.1).
\vspace{1ex}

\noindent
{\bf 2 part 1)} If $\ES$ is a constant, then $\mathcal{F}$ satisfies energy Eq.~(\ref{eneq}.1), as
is, with a constant energy~$\ES$. According to Theorem~(\ref{theorem00}), the wavefunction
$(\Psi = \Psi(\rho,S)\hs{0.1ex})$ satisfies the real part of the TISE. Since, by hypothesis,
$\mathcal{F}$ also satisfies continuity Eq.~(\ref{cont}.2), and from Theorem~\ref{conti},
continuity Eq.~(\ref{cont}.2) is equivalent to the imaginary part of the TISE
(\ref{schrodinger}), the wavefunction $\Psi$ also satisfies the imaginary part of the TISE
(\ref{schrodinger}). Hence, a constant energy $\ES$ implies that $\Psi$ is an eigenvector of the
TISE.
\vspace{1ex}

\noindent
{\bf 2 part 2)} If the wavefunction $(\Psi = \Psi(\rho,S))$ is an eigenvector of Hamiltonian
$H$ with external potential $U$, then, according the Theorem~\ref{theorem00}, $\mathcal{F}$
satisfies the energy Eq.~(\ref{eneq}.1), as is, with a constant energy~$\ES$.  (Since, both
$\mathcal{F}$ and $\mathcal{F}_\times$ satisfy the energy Eq.~(\ref{eneq}.1) with the same
energy $\ES$, the flow $\mathcal{F}_\times$ also satisfies the Euler Eq.~(\ref{euler0}), as is,
with a constant energy~$\ES$.)

\noindent
{\bf 3)} If mass is not conserved, Eq.~(\ref{1895c}) is not an Euler equation, only because it
is missing the mass flux term $(\nabla\cdot(\rho_m\oomega)\oomega$. If mass is conserved, this
term vanish.
\begin{flushright}\fbox{\phantom{\rule{0.5ex}{0.5ex}}}\end{flushright}
Next we make an argument that, because the direction of $\bu$ is $-\nabla\rho$, and $\bu$ is
irrotational, there is always a combination of a nowork force and prescribed cross flow
velocity direction that gives incompressible flow with mass conservation. Given that the
nowork force is at our disposal, we can ``steer'' a fluid particle with speed $|\bu|$ in any
continuous path, including with a streamline with tangent vector $\nabla S/|\nabla S|\times
\bu/|\bu|$. According to Corollary~(\ref{2744}), such a flow is incompressible and mass
conserving. In cases where $\nabla S$ is the zero vector, we have more options for the velocity
direction, so there could be more possibilities. We simply need a potential $\phi$, such that
$\nabla\phi$ is not in the same direction as $\nabla\rho$, and then take velocity $\mmu$ to be
$\nabla \phi/|\nabla \phi|\times \bu$. For time dependence, if the flow is laminar and changing
in a continuous manner, it seems like a reasonable assumption that you can find such a flow
with a time dependent nowork force.

Next we consider if a variable mass flow $\mathcal{F}_{\text{\scriptsize vm}}(\rho_m,\Pu,\bu)$
can have meaning in a flow of classical mechanics, where $\bu$ is not necessarily given by
Def.~(\ref{5204b}.2).  Let the flow $\mathcal{F}(\rho_m,\Pu,\mmu)$ be an actual flow of a
classical fluid. Let the variable mass-flow $\mathcal{F}_{\text{\scriptsize
    vm}}(\rho_m,\Pu,\bu)$ have a velocity direction perpendicular the $\mmu$ and satisfy $(u^2
= \mu^2)$.  Consider $\mathcal{F}_{\text{\scriptsize vm}}$ as a fictitious, or imaginary flow,
making $\mathcal{F}_{\text{\scriptsize vm}}$ a sort of cross flow of
$\mathcal{F}(\rho_m,\Pu,\mmu)$.  The imaginary flow $\mathcal{F}_{\text{\scriptsize vm}}$ can
be used to obtain information about flow $\mathcal{F}$ along lines with directions that are
normal to the streamlines of $\mathcal{F}$.  A fluid particle in the imaginary flow is like a
probe, giving information about the mass density, speed, pressure, and kinetic and potential
energies of flow $\mathcal{F}$ along its streamlines. The imaginary flow carries all the
information about the actual flow, except velocity direction.

Next we try to give meaning to formula~(\ref{4720}.1) for pressure $\Pu$.
Consider again the imaginary $\mathcal{F}(\rho_m,\Pu,\bu)$ flow of flow
$\mathcal{F}_\times(\rho_m,\Pu,\mmu)$, where both flows satisfy Bernoulli Eq.~(\ref{eneq}), and
$\mathcal{F}$ satisfies the Euler equation:
\begin{equation} \zlabel{4287}
\fc12\nabla u^2  + \rho^{-1}\left(\nabla\cdot(\rho_m\bu)\bu  + \nabla\Pu\right) + \nabla U = \vec{G},
\end{equation}
with nowork force $\vec{G}$.  The integration of this equation must give the Bernoulli
Eq.~(\ref{eneq}). Hence,
\[\hs{-7ex}\nabla\cdot(\rho_m\bu)\bu  + \nabla\Pu = \rho_m\nabla\left(\Pu\rho_m^{-1}\right)
= \nabla\Pu -\rho_m^{-1} \Pu\nabla\rho_m, \quad \text{giving}\]
\begin{equation} \zlabel{4284}
\Pu\nabla\rho_m = -\nabla\cdot(\rho_m\bu)\rho_m\bu.
\end{equation}
Taking the dot product with $\bu$, we have
\[\hs{-7ex}-\Pu = \fc{\nabla\cdot(\rho_m\bu)}{\nabla\rho_m\cdot\bu}\rho_m u^2
= \left(1 + \fc{\nabla\cdot\bu}{\nabla\rho_m\cdot\bu}\right)\rho_m u^2; \qquad \nabla\rho\ne 0,
\]
%
and these equations are satisfied in the special case with Defs.~(\ref{4720}.1) and
(\ref{5204b}) for $\Pu$ and $\bu$, respectively. (We also have $(\nabla\rho_m\cdot\bu =
\pm|\nabla\rho_m|\room|\bu|$.)  Hence, formula~(\ref{4720}.1) for the pressure $\Pu$ is a
special case of the more general relation~(\ref{4284}). Note that if $\nabla\cdot\bu$ vanish
but $\nabla\rho_m\cdot\bu$ does not vanish, then the pressure is $-2$ times the fluid kinetic
energy.  For relation~(\ref{4284}) to hold for a compressible, inviscid fluid of classical
mechanics, there must be a nowork force $\vec{G}$, such that Eq.~(\ref{4287}) is satisfied,
with the restrictions mentioned above, that is, $(\nabla\rho_m\cdot\mmu = \zero)$ and $(u^2 =
\mu^2)$. Also, the pressure $\Pu$ must be nonnegative, and the rhs of Eq.~(\ref{4284}) must
vanish if $\nabla \rho_m$ vanish. However, the pressure $\Pu$ can be replaced by pressure $(\Pu
+ \text{const}.)$, since Eq.~(\ref{4287}) is invariant to $\Pu$ with an additive constant.
Note that $\Pu$ is invariant to a change in sign of $\bu$.

\section{Lagrange's equations of motion for Many Bodies \zlabel{9918}}

Consider the classical Hamiltonian for $n$ identical point masses of mass $m$ with external $V$
and interacting $W$ potentials:
\[\bar{W} \defi \sum_{i=1}^nW_i, \qquad W_i(\mrr_1,\cdots\mrr_n) \defi \sum_{j\ne i}^nW_{12}(\mrr_j,\mrr_i)\]
\begin{equation} \hs{-3ex} \zlabel{3777}
E = \sum_{i=1}^n\,  \fc12m\mss_i\cdot\mss_i + V(\mrr_i)  + \;  \fc12 \sum_{i=1}^n\sum_{j\ne i} W(\mrr_i,\mrr_j), \quad E = \text{const},
\end{equation}
where velocities $\mss_i$ depends on time.  As is well known, it follows from Lagrangian
mechanics, that the $n$ equations of motion are
\begin{equation} \zlabel{3877}
 m\fc{d\mss_i}{dt} =  -\nabla_i V + \sum_{j\ne i}^n \nabla_i W(\mrr_i,\mrr_j), \qquad i = 1,\cdots n.
\end{equation}
Let $[(\mrr_1,\cdots,\mrr_n)](t)$ be the particle configuration at time $t$.  If this particle
configuration-function is invertible, giving $(t = t(\mrr_1,\cdots,\mrr_n))$, then we can
define a generalized velocity fields: $(\hs{0.3ex}\bv_i(\mrr_1,\cdots,\mrr_n) =
\mss_i\left[t(\mrr_1,\cdots,\mrr_n)\right]\hs{0.3ex})$ with domain $\mathbb{R}^{3n}$.  The
Hamiltonian above can then be written
\begin{equation} \hs{-5ex}  \zlabel{4852}
E = \sum_{i=1}^n \fc12m\bv_i^2 + V(\mrr_i) + \fc12 \sum_{j\ne i} W(\mrr_i,\mrr_j).
\end{equation}
Furthermore, if there exists a time range $\tau \in (t_0,t_1)$, where $(\mrr_i(t_a) =
\mrr_i(t_b))$ $\imply$ $(t_a = t_b)$, for all $i\in\{1,\cdots,n\}$, then, during the time
interval we have the maps $\mrr_i \mapsto t\in\tau$, and we can define the velocity fields
$(\hs{0.2ex}\my_i(\mrr_i) = \mss_i[t(\mrr_i)\hs{0.1ex}]\hs{0.3ex})$ with domain
$\mathbb{R}^3$. Using the chain rule for the acceleration, the
equation of motion for velocity field $\my_\al \in\{\my_1,\cdots,\my_n\}$ can be written
\[ \hs{-8ex}m\nabla\my_\al \hs{-0.4ex}\cdot\my_\al  =  -\nabla_i V(\mrr_\al) - \sum_{i\ne \al}^n \nabla_\al W(\mrr_\al,\mrr_i);
\qquad (\nabla\my_\al\hs{-0.4ex}\cdot\my_\al)_i = (\nabla y_{\al i}\hs{-0.4ex})\cdot\my_\al.\]

Given the similarities of Hamiltonians~(\ref{eneq})$\times\rho^{-1}$ and (\ref{4852}), it
should be possible to obtain $n$ equations of motion, for $n$ interacting point masses
involving Hamiltonian~(\ref{eneq}.1).  This section does that using Lagrangian dynamics for
stationary states under the orthogonal condition.  The Lagrangian equations obtained are in
particle, or per mass units, where the pressure $\Pu$ appears as a factor in the product
$\Pu\rho^{-1}$. For trajectories of point masses, $\Pu\rho^{-1}$ is a body force. Equation of
motion forms applicable to fluid flow are also given.  An equation of motion form involving
velocity fields $\bu_1$ and $\bv_1$, parameterized by the other $\mrr_2,\cdots\mrr_m$ position
variables, is also obtained. This equation is used in the section that follows to obtain
reduced 1-body equations. Except for difficulties that must be overcome, the reduced equation,
for a particle with position variable $\mrr_1$, is obtained by integrating an equation of
motion over the variables $\mrr_2, \cdots,\mrr_n$.

For later use we note that, since the forces are conservative in the equations of motion sequence~(\ref{3877}),
via integration, there exist constants $\epsilon_1,\cdots\epsilon_n$, such that
\[\epsilon_i = \fc12m\mss_i\cdot\mss_i  - V(\mrr_i) - \sum_{j\ne i}^n W(\mrr_i,\mrr_j).\]
Summing over the particle energies we see that the potential energy of 2-body interactions is double
counted when compared to energy Eq.~(\ref{3777}):
\begin{equation} \hs{-5ex}  \zlabel{4463}
\sum_{i=1}^n\epsilon_i = \sum_{i=1}^n \fc12m\mss_i\cdot\mss_i + V(\mrr_i) +\, \sum_{i=1}^n \sum_{j\ne i} W(\mrr_i,\mrr_j).
\end{equation}

\subsection{Many Body Trajectories and Stream Neighborhoods \zlabel{9920}} 

In this subsection, using the \emph{generalized} velocity fields
$\bu_i,\bv_i\!:\mathbb{R}^{3n}\rightarrow \mathbb{R}^3$, time and field dependent velocities
are defined for stationary quantum states under the orthogonal condition.  For trajectories of
particles, the velocities functions $\bu_i$ and $\bv_i$ are replaced by time dependent ones,
denoted below by $\dot{\mrr}_i:\Omega\rightarrow \mathbb{R}^3$ and
$\dot{\mss}_i:\Omega\rightarrow \mathbb{R}^3$, where $\Omega\subset\hat{\mathbb{R}}$, and
$\hat{\mathbb{R}}$ is the real number line with time as the variable. In classical mechanics,
if the displacement as a function of time, for a given time range, is invertible, velocities as
a function of position can be defined. These velocity fields, denoted by
$\mx_i,\my_i\!:\mathbb{R}^3\rightarrow\mathbb{R}^3$, are also defined.

Let $\bu_i:\mathbb{R}^{3n}\rightarrow \mathbb{R}$ and $\bv_i:\mathbb{R}^{3n}\rightarrow
\mathbb{R}$.  Let
\[\hs{-4ex}\vec{\mq} \defi(\mq_1,\cdots,\mq_n) \in\mathbb{R}^{3n}\quad \text{and}\quad
\vec{\mq}_0 \defi(\mq_{0,1},\cdots,\mq_{0,n}) \in\mathbb{R}^{3n},\]
where $\vec{\mq}_0$ and $(\Omega\defi (t_0,t)\in\hat{\mathbb{R}})$ are, respectively, the
initial configuration and time range, of an $n$-particle trajectory.
Also, $\vec{\mq}:\Omega\rightarrow \mathbb{R}^{3n}$, meaning that 
$\vec{\mq}(\tau)$ is the configuration at time $\tau\in\Omega$. The unique
path $P$ of motion, denoted $(P[\spa\vec{\mq}_0\spa]: \Omega\rightarrow\mathbb{R}^{3n})$, for
initial configuration $\vec{\mq}_0$, is the solutions of the first order differential
equations:
\begin{equation} \hs{-2ex}\label{2418f}
  \dot{\mrr}_i = \bu_i, \quad \dot{\mss}_i \defi \bv_i;\quad \dot{\mq}_i \defi \dot{\mrr}_i + \dot{\mss}_i,
  \quad i = 1,2,\cdots n,
\end{equation}
where $\dot{\mathbf{z}}_i:\Omega\rightarrow \mathbb{R}^3$; $(\dot{\mathbf{z}}_i = \dot{\mrr}_i,
\dot{\mss}_i,\dot{\mq}_i)$, are the $i$th time-dependent velocities.  The path of motion
$P[\spa\vec{\mq}_0\spa]$ is given as an $n$-tuple of curves, the $i$th one being $(\mq_i:
\Omega\rightarrow\mathbb{R}^3)$. If $(\hspace{0.2ex}\vec{\mq}(\tau) =
P[\hspace{0.1ex}\vec{\mq}_0\hspace{0.1ex}](\tau)\hspace{0.3ex)}$, then
\begin{equation} \hspace{-5ex}\zlabel{2430b}
  \dot{\mq}_i(\tau) = [\bu_i+\bv_i]\Bigl(\vec{\mq}(\tau)\Bigl) \defi \bw_i\Bigl(\vec{\mq}(\tau)\Bigl); \quad \tau\in \Omega.
\end{equation}
The path of motion $P[\hspace{0.1ex}\vec{\mq}_0\hspace{0.1ex}]$ is said to be component
invertible if $(\spa\mq_i(\tau_1) = \mq_i(\tau_2)\spa)$ $\imply$ $(\tau_1=\tau_2)\spa)$ for
each $i\in\{1,\cdots,n\}$. For component invertible paths, denoted
$P_I[\hspace{0.1ex}\vec{\mq}_0\hspace{0.1ex}]$, since each value $\mq_i(t)$ from $(t\in\Omega)$
is unique, we have the following bijective map: $(\mq_i\mapsto \vec{\mq}\spa)$.  For component
invertible path $P_I[\hspace{0.1ex}\vec{\mq}_0\hspace{0.1ex}]$, using this bijection, let the
velocity fields $\mx_i:\text{Range}(\mq_i) \rightarrow \mathbb{R}^3$ and
$\my_i:\text{Range}(\mq_i)\rightarrow \mathbb{R}^3$ be
\begin{equation} \hs{-8ex}\label{2430c}
  \mx_i(\mq_i) \defi \bu_i(\vec{\mq}), \quad \my_i(\mq_i) \defi \bv_i(\vec{\mq});
  \quad \vec{\mq} = \vec{\mq}(\mq_i),
  \quad(\vec{\mq}) \in \text{Range }P_I[\hspace{0.1ex}\vec{\mq}_0\hspace{0.1ex}].
\end{equation}
Comparing Eq.~(\ref{2418f}) and (\ref{2430c}), restricted to
  $P_I[\hspace{0.1ex}\vec{\mq}_0\hspace{0.1ex}]$, we find the following logic
\begin{equation} \hs{-8ex}\zlabel{3848}
\Bigl(\mq_1(\tau),\cdots,\mq_n(\tau)\hspace{0.1ex}\Bigl) = P_I[\hspace{0.1ex}\vec{\mq}_0\hspace{0.1ex}](\tau)
  \;\imply\; \dot{\mrr}_i(\tau) = \mx_i(\mq_i), \quad \dot{\mss}_i(\tau) \defi \my_i(\mq_i),
\end{equation}
and the the sum of the two implications gives
\begin{equation} \hs{-4ex}\label{3848b}
\Bigl(\mq_1(\tau),\cdots,\mq_n(\tau)\hspace{0.1ex}\Bigl) = P[\hspace{0.1ex}\vec{\mq}_0\hspace{0.1ex}](\tau) \;\imply\;
\dot{\mq_i}(\tau) = \mx_i(\mq_i) + \my_i(\mq_i).
\end{equation}

\subsection{Lagrangian equations of motion and the Hamiltonian}

In this subsection, using the relationships defined in the previous subsection, a Hamiltonian,
Lagrangian and equations of motion are obtained for $n$-body stationary quantum states that
satisfy the orthogonal condition and a  local mass-conservation law. We begin with the Hamiltonian.
\begin{theorem}[The Hamiltonian Function of Quantum Mechanics] \zlabel{theorem88e}
 Let $P_I[\hspace{0.1ex}\vec{\mq}_0\hspace{0.1ex}]$ be a component invertible path.
  The $n$--body spatial Schr\"odinger equation:
\begin{equation} \zlabel{1720b}
  -\fc{\hbar}{2m}\nabla^2\Psi + U\Psi = \ES\Psi; \quad \nabla^2\Psi\in\mathbf{\Sigma},
\end{equation}
and the orthogonal condition: $(\nabla\rho\cdot\nabla S = \zero)$,
is equivalent to the Hamiltonian
\begin{equation} \zlabel{5028c}
    H\spa(\hs{0.15ex}\dot{\vec{\mq}},\vec{\mq}\hs{0.15ex}) \defi T(\dot{\vec{\mq}}) + V(\vec{\mq});
   \quad H = \ES = \text{const.},
\end{equation}
and the continuity Eq.~(\ref{cont}.2), where
\begin{equation} \hspace{-2ex}\zlabel{1011c}
 T(\dot{\vec{\mq}}) \defi \fc12m\dot{\mq}\cdot\dot{\mq}, \qquad
  V(\vec{\mq}) \defi [\Pu\rho^{-1}](\mq) + U(\vec{\mq}),
\end{equation}
and the wavefunction $\Psi$ is restricted to $\text{Range}\left(P_I[\vec{\mq_0}]\right)$.  The
$i$th components of
\[(\dot{\mq}\cdot\dot{\mq}),\hspace{0.25ex} \hspace{0.25ex}[\Pu\rho^{-1}](\mq),\hspace{0.25ex} \in\mathbf{\Sigma},\]
are, respectively, the values $\dot{\mq}_i\cdot\dot{\mq}_i$ and 
\begin{equation}\label{3418c}
[\Pu\rho^{-1}]_i(\mq_i) \defi
\bigl[-m\mx_i\cdot\mx_i + \zeta_0\nabla\cdot\mx_i\bigl]\hspace{0.25ex}(\mq_i).
\end{equation}
\end{theorem} 
\textbf{\textit{Proof}} From Theorem~\ref{theorem00}, the time {\it dependent} Schr\"odinger
equation $\Leftrightarrow Eqs.~(\mathbf{\ref{eneq}})$, and from Theorem~\ref{conti}, the
imaginary part of the TISE is equivalent to continuity Eq.~(\ref{cont}.2). For a stationary
state, $\ES$ is constant; $(S(t) = -\ES t)$, and the time independent Schr\"odinger equation (TISE)
is equivalent to Eqs.~(\ref{eneq}) with constant $\ES$. The orthogonal condition $(\bu\cdot\bv
= \zero)$ and $(\bw = \bu + \bv)$, give $(\bw\cdot\bw = \bu\cdot\bu + \bv\cdot\bv)$. Hence, the
real part of the TISE and the orthogonal condition is equivalent the following special case of
Eq.~$(\mathbf{\ref{eneq}.1})\div \rho$:
\begin{equation} \hs{-2ex}\label{7028}
\ES = \fc12m\bw\cdot\bw + [\Pu\rho^{-1}](\bu) + U(\vec{\mq});\quad \ES = \text{const,}
\end{equation}
where, in accordance with Eqs.~(\ref{1222}), 
\[\hs{2ex}[\Pu\rho^{-1}](\bu) = -m\bu\cdot\bu + \zeta_0\nabla\cdot\bu.\]
Let the functions of Eq.~(\ref{7028}) be restricted to
$\text{Range}\left(P_I[\vec{\mq_0}]\right)\in\mathbb{R}^{3n}$ .  With this restriction, using
(\ref{2430c}), we make the following replacement in Eq.~(\ref{7028}): 
$(\bu)\rightarrow (\mx)$.
Also, using Eq.~(\ref{2430b}), we can replace $(\bw)$ by $(\dot{\mq})$.
Hence, {\bf(\ref{1720b})} $\Leftrightarrow$ {\bf(\ref{5028c})} with the given restriction to
$P_I[\spa\vec{\mq_0}\spa]$.
\begin{flushright}\fbox{\phantom{\rule{0.5ex}{0.5ex}}}\end{flushright}

Using the Hamiltonian from the previous theorem, we an now define a Lagrangian that corresponds
to this Hamiltonian. The theorem  that follows obtains the relationship between the two
functions, and the Lagrangian equations of motion for quantum states.
\begin{definition}[The Lagrangian of Quantum Mechanics]\zlabel{5025c}
The Lagrangian~$L$ of quantum mechanics for stationary states, defined on \text{Range}
$\left(P[\vec{\mq_0}]\right)\in\mathbb{R}^{3n}$, is
\begin{equation} \hs{-4ex}\label{8230}
  L(\vec{\mq},\dot{\vec{\mq}}) = T(\dot{\vec{\mq}}) + \text{\small Q}\mathbf{A}\cdot\dot{\mq} - V(\vec{\mq}),
  \qquad (\vec{\mq})\in \text{Range}\left(P[\vec{\mq_0}]\right), 
\end{equation}
where the kinetic $T$ and potential $V$ energies are defined by Eq.~(\ref{1011c}).  Also, the
scalar function $\text{\small Q}$ and vectors $(\mathbf{A})\in\mathbf{\Sigma}$ are
arbitrary. Furthermore, $\mathbf{A}_i\cdot\dot{\mq}_i$ is the $i$th component of the sequence
$(\mathbf{A}\cdot\dot{\mq})$, and the i$th$ component of $(\Pu\rho^{-1})$ is given by
Eq.~(\ref{3418c}).
\end{definition}


\begin{theorem}[The Lagrange's equation of motion and Hamiltonian part {\bf 1}] \zlabel{theorem88}
{\bf 1)} The Lagrangian equations of motion, given by the vector equation sequence
\begin{equation} \zlabel{1034b}
\pa{L}{\mq} = \fc{d}{dt}\pa{L}{\dot{\mq}}; \quad \partial L/\partial \mathbf{q},\hspace{0.5ex}\partial L/\partial \mathbf{\dot{q}}\in\mathbf{\Sigma},
\end{equation}
is equivalent to the Hamiltonian $H$ being constant in time:
\begin{equation} \zlabel{1134b}
  H \defi \dot{\mq}\cdot\pa{L}{\dot{\mq}} - L,
\end{equation}
for $\partial L/\partial\mq$ restricted to $\text{Range }P_I[\spa\vec{\mq}_0\spa]$, and where
the Hamiltonian is given by Eq.~(\ref{5028c}).

\vspace{1ex}
{\bf 2)} The Lagrange's equations of motion for quantum states are the following form of Newton's second law:
\begin{equation} \hs{-2ex}\zlabel{2520}
\fc{d}{dt}(m\dot{\mq}) + \text{\small Q}\fc{d\mathbf{A}}{dt} =
\left[-\nabla(\Pu\rho^{-1}) - \nabla U\right](\mq) + 
\text{\small Q}\nabla\left(\mathbf{A}\cdot\dot{\mq}\right),
\end{equation}
a sequence of $n$ vector equations.
\end{theorem}

\noindent
\textbf{\textit{Proof}}
For part {\bf 1}, using the chain rule, the time derivative of $(L = L(\mq\hspace{0.2ex},\dot{\mq})\hs{0.25ex})$ is
\begin{equation} \hs{-7ex} \label{4429}
\fc{dL}{dt} = \pa{L}{\mq}\cdot\dot{\mq}  + \pa{L}{\dot{\mq}}\cdot\ddot{\mq},
\qquad (\partial L/\partial\mq)\cdot\dot{\mq}, (\partial L/\partial\dot{\mq})\cdot\ddot{\mq} \in\mathbf{\Sigma},
\quad  \partial L/\partial t \notin\mathbf{\Sigma},
\end{equation}
where, according to Eq.~(\ref{8230}), $L$ does not explicitly depend on time $t$.  Suppose the
Lagrangian equations of motion~(\ref{1034b}) is true.  The substitution of the Lagrangian
equations into the above one gives
\[ \fc{dL}{dt} = \left(\fc{d}{dt}\pa{L}{\dot{\mq}}\right)\cdot\dot{\mq}
+ \pa{L}{\dot{\mq}}\cdot\ddot{\mq}
\imply  \fc{d}{dt}\left(\dot{\mq}\cdot\pa{L}{\dot{\mq}} - L\right) = \zero.\]
Comparison with definition~(\ref{1134b}), we find that $(H = \text{const.})$. Suppose $(H =
\text{const.})$. Comparing the last equality of the logic
\[\hs{-12ex}H = \dot{\mq}\cdot\pa{L}{\dot{\mq}} - L = \text{const.} \imply  \fc{d}{dt}\left(\dot{\mq}\cdot\pa{L}{\dot{\mq}} - L\right) = \zero
\imply \fc{dL}{dt} = \left(\fc{d}{dt}\pa{L}{\dot{\mq}}\right)\cdot\dot{\mq}
+ \pa{L}{\dot{\mq}}\cdot\ddot{\mq}, \]
with Eq.~(\ref{4429}), gives the the Lagrangian equations of motion~(\ref{1034b}).

\vspace{1ex}
For part {\bf 2}, equations of motion~(\ref{2520}) follows by substitution of Eq.~(\ref{8230}) into (\ref{1034b})
and using Eq.~(\ref{1011c}). 
\begin{flushright}\fbox{\phantom{\rule{0.5ex}{0.5ex}}}\end{flushright}

The next theorem obtains explicit forms of the Hamiltonian and equations of motion for
stationary quantum states that satisfy the orthogonal condition and a local mass conservation
restriction.  The following lemma is needed for the theorem that follows.
\begin{lemma} \zlabel{8527}
The $n$--body time independent Schr\"odinger equation with orthogonal condition $(\nabla\rho\cdot\nabla S = \zero)$ implies
Lagrange's equations of motion~(\ref{2520}) with the conservative forces restricted to
component invertible path $P_I[\spa\vec{\mq}_0\spa]$, within definitions of $\mathbf{A}$ and
$Q$.
\end{lemma}
\textbf{\textit{Proof}} From Theorem~(\ref{theorem88e}), the $n$--body TISE and
$(\nabla\rho\cdot\nabla S = \zero)$ imply both the kinetic $T$ and potential $V$ energy
functions of Hamiltonian (\ref{5028c}), restricted to $P_I[\spa\vec{\mq}_0\spa]$.  Via the Lagrange's
Eq.~(\ref{1034b}), we have $\mathbf{(\ref{5028c})} \imply
\mathbf{(\ref{2520})}$.
\begin{flushright}\fbox{\phantom{\rule{0.5ex}{0.5ex}}}\end{flushright}
\begin{theorem}[The Lagrange's equation of motion and Hamiltonian part {\bf 2}] \zlabel{theorem88b}
The $n$-body time independent Schr\"odinger equation imply
\begin{equation} \hs{-4ex}\zlabel{5028b}
  H = \fc12m\dot{\mrr}\cdot\dot{\mrr} + \fc12m\dot{\mss}\cdot\dot{\mss}
  + [\Pu\rho^{-1}]\hs{0.1ex}(\mq) + U(\vec{\mq}), 
\end{equation}
and the following two sequences of $n$ equations: 
\begin{equation} \hspace{-5ex}\label{3510}
  \fc12\rho_m\nabla x^2 + \fc12\rho_m\nabla y^2 + \nabla\cdot(\rho_m\mx)\mx  = -\nabla\Pu - \rho\nabla U,
\end{equation}
\begin{equation}  \hs{-8ex} \label{4028}
  \mathcal{D}\!\!\int_{\mathbb{V}}\rho_m\dot{\mrr} + \rho_m\dot{\mss}
  - \int_{\mathbb{V}}\nabla\cdot(\rho_m\mx)\my = \int_\mathbb{V} -\nabla\Pu - \rho\nabla U;
  \quad \mathbb{V}\subset\mathbb{R}^3,
\end{equation}
under the orthogonal condition $(\nabla\rho\cdot\nabla S = \zero)$, the condition that the
potential energy functions are restricted to $\text{Range
}\left(P_I[\vec{\mq_0}]\right)\in\mathbb{R}^{3n}$, and $(\mathbf{A}\cdot\dot{\mq})$ is not
present. In addition, for implication~(\ref{3510}) and (\ref{4028}), body mass conservation:
$(\nabla_i\cdot(\rho_m\bv_i) = \zero)$, for $(i = 1,2,\cdots n)$, and this condition replaces,
and implies, the weaker condition $(\sum_i\nabla_i\cdot(\rho_m\bv_i) = \zero)$.  (For one-body
states with steady flow, Eq.~(\ref{5028b}) correspond to energy Eq.~(\ref{eneq});
Eq.~(\ref{3510}) correspond to the Euler Eq.~(\ref{euler0}); Eq.~(\ref{4028}) corresponds to
Euler Eq.~(\ref{euler1}), using Eq.~(\ref{5240}): $(\nabla\cdot\bv = \zero)$.)
\end{theorem}

\textbf{\textit{Proof}} 
For implication~(\ref{5028b}), since, from Theorem~(\ref{theorem88e}), the $n$--body time
independent Schr\"odinger equation with orthogonal condition implies Eq.~(\ref{5028c}), we only
need to prove that $\mathbf{(\ref{5028c}) \imply (\ref{5028b})}$. From the orthogonal condition
$(\bu\cdot\bv = \zero)$, we have $(\dot{\mq}\cdot \dot{\mq} = \dot{\mrr}\cdot\dot{\mrr} +
\dot{\mss}\cdot\dot{\mss})$.  Substituting this condition into Eq.~(\ref{5028c}) via
Eq.~(\ref{1011c}), we obtain Eq.~(\ref{5028b}).

For implications~(\ref{3510}) and (\ref{4028}), since, from Lemma~(\ref{8527}), the $n$--body
time independent Schr\"odinger equation and the orthogonal condition imply
Eq.~(\ref{2520}), we only need to prove the two implications below.

\vspace{1ex} \underline{$\mathbf{(\ref{2520}) \imply (\ref{3510})}$}.  Consider a velocity
function of time $\mathbf{z}:\Omega\rightarrow \mathbb{R}^3$. If the path is invertible, then,
in Cartesian coordinates, we can write the velocity as a composite $(\mathbf{z}(t) =
\mathbf{z}( x(t),y(t),z(t) \spa)$, as is done in the usual procedure of fluid dynamics for treating
acceleration in the derivation of the Euler equation. This also gives the velocity field:
$(\mathbf{z} = \mathbf{z}(x,y,z)\spa)$. Hence, the same symbol $\mathbf{z}$ is used for two different
functions, both representing the same physical quantity.  If $\mathbf{z}$ is irrotational and
the flow is steady, the two functions satisfy
\[\fc{d\mathbf{z}}{dt} = \fc12\nabla \mathbf{z}\cdot\mathbf{z},\]
and this follows by taking the material derivative and utilizing an equality for irrotational
vectors. From the development in subsection~\ref{9920}, since different symbol for the
two functions are used, we have
\[\fc{d}{dt}(m\dot{\mq}) = \fc12m\nabla[(\mx + \my)\cdot(\mx + \my)]
= \fc12m\nabla \mx\cdot\mx  + \fc12m\nabla \my\cdot\my,\]
where the orthogonal condition is used, and also that $\bu_i$ and $\bv_i$ being irrotational
imply that $\mx_i$ and $\my_i$ are also irrotational.  It follows from logic~(\ref{3848b}) that
this equation is satisfied on path of motion $P_I[\hspace{0.1ex}\vec{\mq}_0\hspace{0.1ex}]$,
connecting the time dependent trajectory velocities and the field velocities.  Substituting
into Eq.~(\ref{2520}) and setting $(\hspace{0.1ex}(\mathbf{A}) = \zero)$, we have
\[ \fc12m\nabla x^2 + \fc12m\nabla y^2= -\nabla(\Pu\rho^{-1}) - \nabla U.\]
Multiplying by $\rho$, and using equality~(\ref{2077}), generalized to a sequence, we obtain
%
\[ \hs{-10ex}\text{Eq.~{\bf (\ref{3510})}}: \quad 
\fc12\rho_m\nabla x^2 + \nabla\cdot(\rho_m\mx)\mx + \fc12\rho_m\nabla y^2
= -\nabla\Pu - \rho\nabla U.\]
\vspace{1ex} \underline{$\mathbf{(\ref{3510}) \imply (\ref{4028})}$}.
Consider Eq.~(\ref{5940}) that holds for any irrotational velocity field. Using our explicit
notation for the velocity as a function of time $(\dot{\mrr} + \dot{\mss}$) and the velocity as a field $\bw$, for steady flow,
and under the orthogonal condition, this equation becomes
\begin{equation} \hs{-5ex}\zlabel{5282b}
  \mathcal{D}\!\!\int_{\mathbb{V}}\; \rho_m\dot{\mrr} + \rho_m\dot{\mss}
    = \int_{\mathbb{V}}\; \fc12\rho_m \nabla \mx\cdot\mx  + \fc12\rho_m \nabla \my\cdot\my + \nabla\cdot(\rho_m\bw)\bw.
\end{equation}
By hypothesis, $(\hs{0.1ex}(\nabla\cdot(\rho_m\bv) )\hs{0.1ex})$ vanish identically, giving
$(\hs{0.1ex}(\nabla\cdot(\rho_m\my)\hs{0.1ex}) = (\zero)$, with $\my$ given by
Def.~(\ref{2430c}). Hence,
\begin{equation*} \hs{-8ex} 
  \mathcal{D}\!\!\int_{\mathbb{V}}\; \rho_m\dot{\mrr} + \rho_m\dot{\mss}
  = \int_{\mathbb{V}}\; \fc12\rho_m \nabla \mx\cdot\mx  + \fc12\rho_m \nabla \my\cdot\my
  + \nabla\cdot(\rho_m\mx)\mx + \nabla\cdot(\rho_m\mx)\my.
\end{equation*}
Integrating Eq.~(\ref{3510}) over $(V\in
P_I[\hspace{0.1ex}\vec{\mq}_0\hspace{0.1ex}]\hs{0.1ex})$, followed by substituting into
Eq.~(\ref{3510}), we obtain Eq.~(\ref{4028}).


\subsection{Parameterized one body Euler and Schr\"odinger equations \zlabel{7545}}

Next we specialize the potential $U$ to a forms that is applicable to electronic
systems.
\begin{definition}[External Potential from 1- and 2-body potentials] \mbox{}

\noindent
The potentials $\bar{V}:\mathbb{R}^{3n}\rightarrow\mathbb{R}\;$ and
$\;\bar{W}:\mathbb{R}^{3n}\rightarrow\mathbb{R}$ are sums over one-body
$V_1:\mathbb{R}^{3}\rightarrow\mathbb{R}$ and two-body
$W_i:\mathbb{R}^3\times\mathbb{R}^3\rightarrow\mathbb{R}$ functions:
\begin{equation}\zlabel{4200a}
\bar{V}(\mrr_1,\mrr_2,\cdots\mrr_n) \defi V(\mrr_1) + V(\mrr_2) + \cdots V(\mrr_n), \quad\text{and}
\end{equation}
\begin{equation}\zlabel{4200b}
\bar{W} \defi \sum_{i=1}^nW_i, \qquad W_i(\mrr_1,\cdots\mrr_n) \defi \sum_{j\ne i}^nW_{12}(\mrr_j,\mrr_i).
\end{equation}
The body potential $U$ of Eq.~(\ref{eneq}) of the TISE is the following sum of one-body and two-body potentials:
\begin{equation}\hs{-7ex}\label{4027}
  U \defi \bar{V} + \fc12\bar{W} \imply \nabla_1 U \defi \nabla_1 V + \nabla_1W_1;
  \qquad W_{12}(\mrr,\mrrp) = W_{12}(\mrrp,\mrr),
\end{equation}
with the condition that $\bar{V}$and $\bar{W}_i$ vanish at infinity.
\end{definition}
Consider the first component $\bu_1$ of the $n$-tuple $(\bu)$. For such first component
functions, the next theorem treats the $(n-1)$ variables from
$(\room\mq_1^\pr = (\mrr_2,\cdots,\mrr_n)\in\mathbb{R}^{3(n-1)}\room)$ as parameters.  A 1-body Euler equation,
Newton's 2nd law, an energy equation, and a spatial Schr\"odinger\"odinger equation is obtained.  The
equation of motion is used in the next section to obtain 1-body reduced equations of motion and
a reduced spatial Schr\"odinger equation, under certain conditions.
\begin{theorem}[From n-body to parameterized 1-body Schr\"odger Equations] \zlabel{5935} \mbox{}
  
  Let $(\hs{0.3ex}(\hs{0.1ex}\mq_1,\mq_1^\pr)\in\mathbb{R}\times\mathbb{R}^{n-1} =
  \mathbb{R}^{3n} \hs{0.3ex})$. The $n$-body time-independent Schr\"odger equation
\begin{equation} \hs{-8ex}\zlabel{2720}
  -\fc{\hbar}{2m}\nabla^2\Psi + \left(\bar{V} + \fc12\bar{W}\right)\Psi = \ES\Psi; \quad \nabla^2\Psi\in \mathbf{\Sigma},
\quad  \Psi:\mathbb{R}^{3n}\rightarrow \mathbb{C},
\end{equation}
under the orthogonal condition, and the sequence of equations
\begin{equation} \zlabel{5207}
  \text{Re}\left(\nabla_i\cdot\left(\Psi^*\hat{\text{\rm P}}_i\Psi\right)\hs{0.2ex}\right) = \zero,
  \quad \hat{\text{\rm P}}_i \defim -i\hbar\nabla_i, \quad  i = 1,\cdots n,
\end{equation}
imply the 1-body Euler equation, Newton's 2nd law, the energy equation, and the TISE:
\begin{equation} \hs{-6ex}\label{2530} 
\fc12\rho_m\nabla_1(u_1^2 + v_1^2) + \nabla_1\cdot(\rho_m\bu_1)\bu_1 = -\nabla_1\Pu_1 -
\rho\nabla_1 V - \rho\nabla_1 W_1,
\end{equation}
\begin{equation} \hs{-5ex}\zlabel{1502}
\mathcal{D}(\rho_m\bu_1 + \rho_m\bv_1) + \rho_m\bu_1\nabla_1\cdot\bu_1
= -\nabla_1\Pu_1 - \rho\nabla_1V - \rho\nabla_1 W_1,
\end{equation}
\begin{equation} \hs{-3ex} \label{eneq-cross} 
  \ES_1\rho = \fc12\rho_mu_1^2 + \fc12\rho_mv_1^2 + \Pu_1 + (V + W_1)\rho,
\end{equation}
\begin{equation} \hs{-2ex}\label{1500}
  -\fc{\hbar}{2m}\nabla_1^2\Psi + \left(V + W_1\right)\Psi = \ES_1 \Psi;
  \qquad \Psi:\mathbb{R}^3\rightarrow \mathbb{C}, 
\end{equation}
where the points $(\mq_1^\pr)\in\mathbb{R}^{n-1}$ are taken as parameters for
$\Psi$ in Eq.~(\ref{1500}), and, therefore, the symbol $\Psi$ represents different functions in
Eqs.~(\ref{2720}) and (\ref{1500}).
\end{theorem}
\textbf{\textit{Proof}}
Let the subspace $\mathbb{R}_1^{3n}[\mq_1^\pr]\subset \mathbb{R}^{3n}$
parameterized by $(\hs{0.15ex}(\mq_1^\pr) \defi (\mq_2,\cdots\mq_n) \in\mathbb{R}^{3(n-1)}\hs{0.1ex})$ and isomorphic with
$\mathbb{R}^3$, be
\[\mathbb{R}_1^{3n}[\mq_1^\pr] \defi \{ (\mq_1,\mq_1^\pr)\in \mathbb{R}^{3n}\vert (\mq_1^\pr)\in \mathbb{R}_1^{3(n-1)} \!\text{ is fixed.}\}.\]
Consider Theorem~\ref{conti} part {\bf I} result {\bf 4} $\Leftrightarrow$ {\bf 6}. For a
1-body stationary state, this logic reduces to the special case
\begin{equation} \hs{-5ex} \zlabel{5629}
\text{Re}\left(\nabla_1\cdot\left(\Psi^*\hat{\text{\rm P}_1}\Psi\right)\right) = \zero
\Leftrightarrow
\text{Im}\left(\Psi^*\hat{H}\Psi\right) = \zero.
\end{equation}
Since $\text{Im}(\Psi^*\ES_1\Psi)$ vanish identically, the second equation of this logic is the
imaginary part of the 1-body TISE~(\ref{schrodinger}). We identify the first equation of
$n$-body sequence~(\ref{5207}) with domain $\mathbb{R}_1^{3n}[\mq_1^\pr]$ as the lhs equation
of logic~(\ref{5629}), parameterized with $(\mq_1^\pr)$. Hence, we have the following: {\bf The
  $n$-body equation sequence~(\ref{5207}), restricted to subspace
  $\mathbb{R}_1^{3n}[\mq_1^\pr]\subset \mathbb{R}^{3n}$, implies the satisfaction of the
  imaginary part of the 1-body TISE~(\ref{1500}), that is multiplied by $\Psi^*$ and parameterized with
  $(\mq_1^\pr)$.}

Substituting Def.~(\ref{def10}) into condition~(\ref{5207}), we recognize that
condition~(\ref{5207}) is equivalent to body mass conservation: $(\nabla_i\cdot(\rho_m\bv_i) =
\zero)$. From Theorem~\ref{theorem88b}, under the condition $(\nabla_i\cdot(\rho_m\bv_i) =
\zero)$, the $n$-body TISE \underline{$\mathbf{(\ref{2720})} \imply \mathbf{(\ref{3510})}$}
with the functions restricted to $\text{Range }P[\spa\vec{\mq}_0\spa]$. Consider the first
equation from equation sequence~(\ref{3510}) with initial position
$(\mq_1,\mq_1^\pr)\in\mathbb{R}^{3n}$ at time $t_0$ of range $(t_0,t)\in\mathbb{R}$ of path
$P[(\mq_1,\mq_1^\pr)]$.  At this initial point of space $(\mq,\mq_1^\pr)\in\mathbb{R}^{3n}$ and
initial time $t_0\in\mathbb{R}$, using Eq.~(\ref{2430c}) with $(\hs{0.1ex}(\vec{\mq}) =
(\vec{\mq}_0) = (\mq_1,\mq_1^\pr) \hs{0.1ex})$, this equation becomes Eq.~(\ref{2530}),
where the body force $-\nabla_1 U$ of Def.~(\ref{4027}) is substituted.  Since $\mq_1$ of the
point $(\mq_1,\mq_1^\pr)$ is arbitrary, Eq.~(\ref{2530}) can be extended from the point
$(\mq_1,\mq_1^\pr)$ to $\mathbb{R}_1^{3n}[\mq_1^\pr]$ a.e. Hence, by this extension,
\underline{$\mathbf{(\ref{2720})} \imply \mathbf{(\ref{2530})}$},

Almost everywhere, by inspection, {\bf \underline{Eq.~(\ref{2530}) is Euler
    Eq.~(\ref{euler0})}} {\it for one-body stationary states}, under the following special
conditions: {\bf 1)} $\mq^\pr$ are parameters, {\bf 2)} $\nabla U$ is defined by
definition~(\ref{4027}), {\bf 3)} the $n$-body state is stationary, and {\bf 4)}
constraint~(\ref{5207}), where the flow is a 1-body quantum flow. From statement {\bf 2} of
Theorem~\ref{5293}, we have the following: The 1-body Euler Eq.~(\ref{euler0}) $\imply$ the
real part of the 1-body TISE (\ref{1500}), multiplied by $\Psi^*$. Combining this with the
results above, we have $\mathbf{Eq.~(\ref{2720})} \text{ and } \mathbf{Eq.~(\ref{5207})})
\imply \mathbf{Eq.~(\ref{1500})}$. Implication Eq.~(\ref{eneq-cross}) follows, because from
Theorem~\ref{5293}, we have $\mathbf{Eq.~(\ref{2530})}) \imply
\mathbf{Eq.~(\ref{eneq-cross})}$.

For the implication of Euler Eq.~(\ref{1502}), we begin by combining the following equalities for
steady, irrotation, smooth flow: 
\[\hs{-10ex}\fc12\rho_m\nabla_1(u_1^2 + v_1^2) + \nabla_1\cdot(\rho_m\bu_1)\bu_1
  = \rho_m \mathcal{D}(\bu_1 + \bv_1 )  + \bu_1 \nabla_1\rho_m\cdot\bu_1 + \rho_m\bu_1\nabla_1\cdot\bu_1,\]
%
\[\hs{-12ex} \mathcal{D}(\rho_m\bu_1) = \rho_m \mathcal{D}\bu_1  + \bu_1 \nabla_1\rho_m\cdot\bu_1, \quad
 \mathcal{D}(\rho_m\bv_1) = \rho_m \mathcal{D}\bv_1 + \bv_1 \nabla_1\rho_m\cdot\bv_1 = \rho_m \mathcal{D}\bv_1,\]
to obtain
\[\hs{-5ex}\fc12\rho_m\nabla_1(u_1^2 + v_1^2) + \nabla_1\cdot(\rho_m\bu_1)\bu_1
  =  \mathcal{D}(\rho_m\bu_1 + \rho_m\bv_1) + \rho_m\bu_1\nabla_1\cdot\bu_1.\]
Substituting this result in Eq.~(\ref{2530}) we obtain Eq.~(\ref{1502}).
\begin{flushright}\fbox{\phantom{\rule{0.5ex}{0.5ex}}}\end{flushright}
Renziehausen and Barth \cite{Renziehausenb} obtain an equation of motion for $n$-body
trajectories of Bohmian mechanics. The equation has a single total momentum derivative involving
all $n$ bodies.  Also, an Euler equation for $n$-body states that generalizes
Euler~Eq~(\ref{euler3}) has also been obtained \cite{Finley-equations}. For, say particle \#1,
the equation contains velocities and pressures of the other $(n-1)$ bodies that act like
forces. These equations could be derived for classical systems, but they would be more
difficult to solve, then the standard collection of $n$ equations of motion for $n$
bodies. 

\section{Quantum Flow Mixture and Orbital Flow \zlabel{4918}}

In this section, the forces and fluid momentae of Eq.~(\ref{1502}) are reduced to 1-body. A
generalization of Coulomb's law is identified. Inferences and hypotheses are employed to move
forward and obtain 1-body quantum orbital flow, with motion, energy and Schr\"odinger equations
for $n$ orbital states.  Before that, a flow mixture approach is also obtained, where many-body
states are treated as 1-body mixtures.

\subsection{The quantum reduced forces}

In this subsection, the forces of Euler Eq.~(\ref{1502}) are reduced to 1-body.
The next theorem obtains a modified Coulomb's law with the product of charge densities are
replaced by the pair function $\rho_2$, a two-body function of reduced density matrix theory
\cite{Parr:89,Cioslowski}. We begin with notation.

\begin{definition}[Scalar $V_{\mbox{\small e}}$ and vector $\mathbf{J}_{\mbox{\small e}}$ potentials]
  \mbox{}
\begin{equation} \hs{-9ex} \label{8529}
  \vec{\mathcal{E}}(\mrr_1) \defi
 -\left[\nabla_1 V_{\mbox{\small e}} + \nabla_1\times \mathbf{J}_{\mbox{\small e}}\right](\mrr_1)
 \defi - 2\int_2 \qcharge(\mrr_1,\mrr_2)\nabla_1 W_{12}(\mrr_2,\mrr_1)\; d\mrr_2,
\end{equation}
\begin{equation} \hs{-2ex}\zlabel{7529}
    \qcharge(\mrr_1,\mrr_2) \defi \fc{\rho_2(\mrr_1,\mrr_2)}{\hat{\rho}(\mrr_1)},\quad
    \rho_2 \defi \fc{(n-1)}{2}\sum_2\intf{2+} \rho,
\end{equation}
\begin{equation} \label{5024}
  \qcharge(\mrr_1,\mrr_2) \defi \fc12\hat{\rho}(\mrr_2) + \fc12\rho_{xc}(\mrr_1,\mrr_2),
\end{equation}
\begin{equation} \hs{-2ex}\label{4282}
  1\coulp{\mrr_1}{\mrr_2} \defi \fc{1}{|\mrr_1 - \mrr_2|},
   \!\quad\text{and}\quad  1\coulf{\mrr_1}{\mrr_2} \defi \fc{\mrr_1\! -\! \mrr_2}{|\mrr_1 - \mrr_2|^3},
\end{equation}
where the map $(W_{12},\qcharge)\mapsto (V_{\mbox{\small e}},\mathbf{J}_{\mbox{\small e}})$,
defined by Eq.~(\ref{8529}) is such that $|\nabla_1\times\mathbf{J}_{\mbox{\small e}}|$ is
minimized.  For electronic systems,
$(W_{12}(\mrr_1,\mrr_2) \defi |\mrr_1 \mbox{\tiny $\!\;\setminus\!\;$} \mrr_2|^{-1})$.
Reduced $m$--body density matrices of the wavefunction $\Psi$ are called $m$--body dentrices.
The function $\qcharge$ is the quantum charge density, and $\rho_2$ is the pair- or 2-density,
normalized to $(n-1)$, instead of the usual number of pairs $n(n-1)\div 2$.  We have
$(\hs{0.1ex}\rho_2(\mrr,\mrrp) = \gamma_2(\mrr,\mrrp; \mrr,\mrrp)\hs{0.1ex})$, where $\gamma_2$
is the reduced 2-body dentrix, normalized to $(n-1)\div 2$. As indicated in Def.~(\ref{2915}),
the probability density function $(\hat{\rho}:\mathbb{R}^3\rightarrow \mathbb{R})$ is
parameterized by one spin variable. As indicated in Def.~(\ref{7529}.2), the second spin
variable $\omega_2$ of $\rho_2$, corresponding the second spatial component of
$(\mrr_1,\cdots \mrr_n)\in\mathbb{R}^3$, is summed over. However, the spatial variable
$\mrr_2$ is not integrated over.  Therefore, $\rho_2$ has one spin parameter that corresponds
to $\mrr_1$.
\end{definition}
The requirement that $|\nabla_1\times\mathbf{J}_{\mbox{\small e}}|$ is minimized is given
because a Helmholtz decomposition is not unique. This is easily verified by recognizing that a
vector field can be both irrotational and solenoidal; such vectors have potentials that satisfy
Laplace's equation. An example being the velocity vector $\bv$ for hydrogen electron states,
except for the s states.

In this section, we make a hypothesis when we want to know if a statement is true, and we
think it might be true. Such propositions are considered useful even if they are false, if
they  provide good approximations for a range of states.
\begin{hypothesis}[Force $\vec{\mathcal{E}}$ conservation.]
  \zlabel{1088} $(\nabla_1\times \mathbf{J}_{\mbox{\small e}} = \zero)$.
\end{hypothesis}
%



\begin{theorem}[Quantum Coulomb's law] \label{1970} Let $W_{12}(\mrr_1,\mrr_2) \defi |\mrr_1 \mbox{\tiny $\!\;\setminus\!\;$} \mrr_2|^{-1}$.
\begin{equation}\hs{-3ex}\label{5020}
  \vec{\mathcal{E}}(\mrr_1)  = \fc{1}{4\pi\epsilon_0}\int_2 2\qcharge(\mrr_1,\mrr_2)\coulf{\mrr_1}{\mrr_2};
  \qquad 4\pi\epsilon_0 = 1.
\end{equation}
Under hypothesis~\ref{1088}: $(\vec{\mathcal{E}}\cdot d\hat{\mathbf{s}} =
-\nabla V_{\mbox{\small e}}\cdot d\hat{\mathbf{s}})$, we have
\begin{equation}\label{5022}
V_{\mbox{\small e}}(\mrr_1) = \fc{1}{4\pi\epsilon_0}\int_{\infty}^{\mrr_1}
\left(-\int_2 2\qcharge(\mss,\mrr_2)\coulf{\mss}{\mrr_2} \right)\cdot d\mss,
\end{equation}
that can be written
\begin{equation}\hs{-10ex}\label{5022b}
V_{\mbox{\small e}}(\mrr_1) =
\fc{1}{2\pi\epsilon_0}   \int_2 \qcharge(\mrr_1,\mrr_2) \coulp{\mrr_1}{\mrr_2}
  - \fc{1}{2\pi\epsilon_0}\int_{\infty}^{\mrr_1}
  \left(\int_2 \nabla_\mss\qcharge(\mss,\mrr_2)\coulp{\mss}{\mrr_2}\right)\cdot d\mss. \hs{0.5ex}
\end{equation}
Using the definition for the exchange correlation hole $\rho_{xc}$, Def.~(\ref{5024}), we have
\begin{equation} \hs{-8ex} \zlabel{5026}  
  \vec{\mathcal{E}}(\mrr_1) 
  = \fc{1}{4\pi\epsilon_0}\int d\mrr_2\;\; \hat{\rho}(\mrr_2)\coulf{\mrr_1}{\mrr_2} 
  +  \fc{1}{4\pi\epsilon_0}\rho_{xc}(\mrr_2,\mrr_1)\coulf{\mrr_1}{\mrr_2},
\end{equation}
\end{theorem}
\textbf{\textit{Proof}} \mbox{} 
Using Def.~(\ref{8529}) and the identity
\begin{equation} \zlabel{2724}
\nabla_1 \left(1\coulp{\mrr_1}{\mrr_2}\right) = -1\coulf{\mrr_1}{\mrr_2},
\end{equation}
we obtain Eq.~(\ref{5020}), and from path integration and by Hypothesis~\ref{1088},
Eq.~(\ref{5022}). Equations~(\ref{5026}) follows immediately by
Def.~(\ref{5024}). Similarly, for Eq.~(\ref{5022b}), we have
\[ \hs{-8ex}
\nabla_\mss\left(\qcharge(\mss,\mrr_2)\coulp{\mss}{\mrr_2} \right)
= \nabla_\mss\qcharge(\mss,\mrr_2)\coulp{\mss}{\mrr_2}   - \qcharge(\mss,\mrr_2)\coulf{\mss}{\mrr_2}.\]
Solving for the last term on the rhs and substituting the  result into Eq.~(\ref{5022}), we obtain Eq.~(\ref{5022b}).
\begin{flushright}\fbox{\phantom{\rule{0.5ex}{0.5ex}}}\end{flushright}
If we use the 2-density of a closed shell determinantal-state, given by, for our normalizations, 
\[\rho_2(\mrr_1,\mrr_2) = \rho(\mrr_1)\rho(\mrr_2) - \rho_1(\mrr_1,\mrr_2)\rho_1(\mrr_2,\mrr_1),\]
for the first term on the rhs of Eq.~(\ref{5022b}), and neglect the second term of $\Ve$, we have
\[\hs{-10.5ex}\int d\mrr_1\; \Ve\hs{0.1ex} \rho(\mrr_1) 
  = \fc{1}{4\pi\epsilon_0}\int\int d\mrr_1d\mrr_2 \;\;\rho(\mrr_1)\rho(\mrr_1)\coulp{\mrr_1}{\mrr_2}
  - \rho_1(\mrr_1,\mrr_2)\rho_1(\mrr_2,\mrr_1)\coulp{\mrr_1}{\mrr_2}.\]
This is two times the Hartree Fock Coulomb and exchange energies, in agreement with their
contributions to the sum of the occupied orbital energies of Hartree Fock theory, and this
double counting is also in agreement with the classical energy Eq.~(\ref{4463}) for $n$ point
masses.

\begin{theorem}[The quantum reduced forces and momentae derivatives] \zlabel{1965}\mbox{}
  Let
\begin{equation} \zlabel{7824}
\hat{\Pu}_1 \defi  \intf{1+} \Pu =  -\zeta\zeta_0\nabla^2 \intf{1+} \rho = -\zeta\zeta_0\nabla^2\hat{\rho}.
\end{equation}
\begin{equation} \hs{-5ex}\zlabel{2042}
  \intp -\nabla_1\Pu_1 - \rho\nabla_1 V - \rho\nabla_1 W_1 
  = -\nabla_1\hat{\Pu} - \hat{\rho}\nabla_1 V + \hat{\rho}\hs{0.3ex}\vec{\mathcal{E}}.
\end{equation}
\begin{equation} \hs{-5ex}\label{2032}
  \intp -\rho\nabla_1W_1 = \hat{\rho}\hspace{0.1ex}\vec{\mathcal{E}}\hspace{0.15ex};
  \qquad \vec{\mathcal{E}} \defi -\nabla_1 V_{\mbox{\small e}} - \nabla_1\times \mathbf{J}_{\mbox{\small e}}.
\end{equation}
\begin{equation} \hs{-1ex}\label{2132}
  \nabla_1\cdot(\rho\bv_1) = \zero \imply \nabla\cdot(\hat{\rho}\hat{\bv}) = \zero.
\end{equation}
\begin{equation} \hs{-5ex}\zlabel{2147}
\intp \fc{d}{dt}(\rho_m\bu_1 + \rho_m\bv_1) = \fc{d}{dt}(\hat{\rho}_m\hat{\bu} + \hat{\rho}_m\hat{\bv}).
\end{equation}
Steady flow and $\nabla_1\cdot(\rho\bv_1)$ vanish imply
\begin{equation} \hs{-7ex}\label{1502r}
\fc{d}{dt}(\hat{\rho}_m\hat{\bu} + \hat{\rho}_m\hat{\bv})
+ \intf{1+}\rho_m\bu_1\nabla_1\cdot\bu_1 
= -\nabla_1\hat{\Pu} - \hat{\rho}\nabla_1 V + \hat{\rho}\hs{0.3ex}\vec{\mathcal{E}}.
\end{equation}
\end{theorem}
\textbf{\textit{Proof}} The reduction of $-\nabla\Pu_1$ follows from Def.~(\ref{7824}).
Recognizing from Def.~(\ref{4200a}) that $V$ only depends on $\mrr_1$, we have
\begin{equation} 
  \intp \rho\nabla_1 V = \hat{\rho}\nabla_1 V.
\end{equation}
For the remaining term,
%
using Def.~(\ref{4200b}), we have
\[\left[\intp \rho\nabla_1W_1\right](\mrr_1)
= \sum_{i\,> 1}^n\,\intp\rho(\mrr_1\cdots\mrr_n)\nabla_1W_{12}(\mrr_i,\mrr_1)\, d\mrrp_1.\]
As is done in density matrix theory, consider the term involving $W_{12}(\mrr_3,\mrr_1)$.  By
interchanging the $\mrr_2$ and $\mrr_3$ dummy variables in the integrand, and using the fact
that $\rho$ is symmetric with respect to interchange of coordinates, we recognize that the
terms with $W_{12}(\mrr_2,\mrr_1)$ and $W_{12}(\mrr_3,\mrr_1)$ are equal. Hence, by a
generalization, each of the $(n-1)$ terms of the sum are equal under integration, giving
\[\hs{-8ex}\left[\intp \rho\nabla_1W_1\right](\mrr_1)
= 2\int_{2}\rho_2(\mrr_1,\mrr_2) \nabla_1W_{12}(\mrr_2,\mrr_1), \]
where Def.~(\ref{7529}.2) for the 2-density $\rho_2$ is used.  Def.~(\ref{8529}) for
$\vec{\mathcal{E}}$ and Def.~(\ref{7529}) for $\qcharge$ then give
Eq.~(\ref{2032}). Eq.~(\ref{2132}) follows from Def.~(\ref{4035}.2).

For Eq.~(\ref{2147}), first we recognize that the time derivative acts on the first component
$n$-tuple. (It is understood that this first variable is a composition time and position,
and the derivative is computed using the chain rule.) Under the assumption of the interchange of
the order of differentiation and integration, Eq.~(\ref{2147}) follows from Eqs.~(\ref{4035}).
Eq.~(\ref{1502r}) follows immediately from Eqs.~(\ref{1502}) and (\ref{2042}).

\begin{flushright}\fbox{\phantom{\rule{0.5ex}{0.5ex}}}\end{flushright}
The reduced pressure $\hat{\Pu}$ also appears as a scalar pressure in the many-body
hydrodynamics equations of motion of Renziehausen and Barth \cite{Renziehausenb}, a formalism
involving a pressure tensor without fluid velocity $\hat{\bu}$.
%
%


We now define a type of flow with reduced variables.
\begin{definition}[Reduced quantum flows]
  The reduced flow $\mathcal{F}(\hat{\rho},\hat{\Pu},\hat{\bu},\hat{\bv})$, of the quantum
  $n$-body flow  $\mathcal{F}(\rho,\Pu,\bu,\bv)$, is the flow that has the reduced variables of
  $\hat{\rho}$, $\hat{\Pu}$, $\hat{\bu}$, $\hat{\bv}$ of $\mathcal{F}(\rho,\Pu,\bu,\bv)$,
  where $\hat{\bu}$ and $\hat{\bv}$ are given by Def.~(\ref{4035}), and the reduced pressure is
  $(\hat{\Pu} = -\zeta\zeta_0\nabla^2\hat{\rho})$.
\end{definition}

Unfortunately we do not know how to reduce the spoiler integral in Eq.~(\ref{1502r}) involving
the volumetric dilation, given the nonlinear velocity $\bu$ dependence, thwarting a complete
reduced equation of motion. Ideally, we would like to replace Eq.~(\ref{1502r}) with
\begin{equation} \zlabel{8322}
\mathcal{D}\!\!\int_\mathbb{V} \hat{\rho}_m\hat{\bu} + \hat{\rho}_m\hat{\bv} =
\int_\mathbb{V} -\nabla_1\hat{\Pu} - \hat{\rho}\nabla_1 V + \hat{\rho}\hs{0.3ex}\vec{\mathcal{E}}.
\end{equation}

To try to find a way around this problem, we make an argument that the spoiler integral can be
neglected.  Consider a generalization of equation of motion~(\ref{1502}), where we have a
moving control volume with velocity $\mx_c$. Letting $(\mx = \bu_1 - \mx_c)$ be the relative
velocity of an observer that moves with the control, and taking $(\bv =\zero)$ and steady flow,
we have
\begin{equation} \hs{-7ex} \label{1502b}
\int_\mathcal{S}\rho_m\bu_1 (\mx\cdot\hat{\mathbf{n}})
  = \int_{\mathbb{V}} -\nabla_1\Pu_1 - \rho\nabla_1V - \rho\nabla_1 W_1 + \rho_m G(\mx_c),
\end{equation}
where $\hat{G}$ is the noninertial force caused by the accelerating control volume. Let
$\mx:\mathbb{R}^3\rightarrow \mathbb{R}^3$ be unparameterized: velocity $\mx$ is required to be
the same velocity field for@all parameters from $\mathbb{R}^{3(n-1)}$.  Eqs.~(\ref{1502}) and
(\ref{1502b}) are two different coordinate systems observing the same event, and, therefore,
carry the same information.  In the special case where $(\mx_c = \zero)$ is taken, equation of
motion~(\ref{1502b}) becomes Eq.~(\ref{1502}). (The momentum time derivatives of the two are
related by identities~(\ref{4125b}) and (\ref{4282b}).)  Applying the operation
$\;\hat{\!\!\int}_{1+}$, we have
\[\int_\mathcal{S}\hat{\rho}_m\hat{\bu} (\mx \cdot\hat{\mathbf{n}}) 
= \int_\mathbb{V} -\nabla_1\hat{\Pu} - \hat{\rho}\nabla_1 V + \hat{\rho}\hs{0.3ex}\vec{\mathcal{E}} + \hat{\rho}_m\hat{G}(\mx_c).\]
Since $\mx$ is at our disposal, we choose $(\mx= \hat{\bu})$, giving
\begin{equation} \label{1502c}
\int_\mathcal{S}\hat{\rho}_m\hat{\bu} (\hat{\bu} \cdot\hat{\mathbf{n}}) = \int_\mathbb{V}
-\nabla_1\hat{\Pu} - \hat{\rho}\nabla_1 V + \hat{\rho}\hs{0.3ex}\vec{\mathcal{E}} +
\intf{1+}\rho_mG (\bu_1 -\hat{\bu}).
\end{equation}
Consider the generalized mean value theorem for an
integrand that is a  product of two functions $f,g:\mathbb{R}\rightarrow \mathbb{R}$ that are
continuous on $\mathbb{R}$: For some $\alpha\in \mathbb{R}$, we have
\[\lim_{a\to\infty} \int_{-a}^{+a} fg  = g^*\int_{-\infty}^{+\infty}f;\qquad g^* = g(\alpha).\]
Similarly, in our case, we have
\[\intf{1+}\rho_mG (\bu_1 -\hat{\bu}) = \hat{\rho}_m G (\bu_1^* -\hat{\bu}).\]
Since $\hat{\bu}$ is an average of $\bu_1$ over space $\mathbb{R}^{3(n-1)}$ with weight
$\rho_m$, $(\bu_1^*= \hat{\bu})$ is, at least, a reasonable approximation. Taking this
approximation, the relative velocity $\mx$ vanish, and the above equation
of motion becomes
\[\mathcal{D}\!\!\int_\mathbb{V} \hat{\rho}\hat{\bw}
= \int_\mathbb{V} -\nabla_1\hat{\Pu} - \hat{\rho}\nabla_1 V
  + \hat{\rho}\hs{0.3ex}\vec{\mathcal{E}}  + \intf{1+}\rho_mG (\zero),\]
where time derivative identity~(\ref{4282b}) is substituted. We interpret this equation as the
equation of motion for a person at rest relative to the fluid velocity, such that $G (\zero)$
is the inertial force per mass.  This equation and Eq.~(\ref{8322}) are, therefore, equations
of motions for the same fluid particle of the same flow, but with different observers:
Eq.~(\ref{8322}) is for an observer at rest, and the one above is for an observer that moves
with the fluid particle. Therefore, the two equations are equivalent.

For the general case, where $\bv_1$ does not vanish, we note that for a bases set such that one
of the basis vectors is in the direction of $\bv_1$, for smooth flow, velocities $\bu_1$ and
$\bv_1$ do not appear in the same equations of motion~(\ref{1502}), and the equation involving
$\bv_1$ can be reduced to 1-body.

Another approach is to use the mean value for the momentum flux directly:
\[\intf{1+}\int_\mathcal{S}\rho_m\bu_1 (\bu_1\cdot\hat{\mathbf{n}})
= \int_\mathcal{S}\hat{\rho}_m\hat{\bu} (\bu_1^* \cdot\hat{\mathbf{n}}),\]
followed by the same approximation $(\bu_1^*= \hat{\bu})$.

Renziehausen and Barth \cite{Renziehausenb}, in their many-body formalism of Madelung fluid
mechanics, are able to reduce the $n$-body momentum time derivative involving velocity $(\bv)$
to 1-body, and obtain reduced equations of motion. Their formalism does not include the second
velocity $(\bu)$ and contains a pressure tensor. They are also successful in the reduction of
the continuity equation to 1-body: ($\partial\hat{\rho} + \nabla\cdot(\hat{\rho}\hat{\bv}) =
\zero)$. This suggests that it might be possible to remove the body mass conservation
condition: $(\nabla\cdot(\rho_i\bv_i = \zero)$ from the formalism presented here, but note that
it is used in Theorem~\ref{5935}.


In subsection~\ref{7928}, a special type of incompressible, cross flows are considered,
where the volumetric dilation is absent, permitting, under certain conditions, the complete
reduction of the equation of motion to 1-body.  In this section, we move forward by inferences
and hypotheses.

\subsection{Quantum Flow Mixture} 

In this subsection, we describe quantum flow mixes, obtained by taking force conservation
Hypothesis~\ref{1088}, and treating the system as, effectively, having one body. We consider
electronic systems with the Bohn--Oppenheimer approximation.  In quantum mechanics, the
external potential and the number of bodies determines the spatial Schr\"odinger equation and,
thus, all eigenfunctions.  Taking Hypothesis~\ref{1088} and 1-body, we identity, from the
quantum reduced forces~(\ref{2042}), the 1-body external potential $(V + V_{\mbox{\small e}})$,
giving the spatial 1-body Schr\"odinger equation:
\begin{equation} \zlabel{4740}
-\fc{\hbar^2}{2m}\nabla^2\mathcal{I} + (V + V_{\mbox{\small e}})\mathcal{I} = \hat{\ES}\mathcal{I}.
\end{equation}
Let $\rho_2$ be the pair density of an $n$-body wavefunction with external potential $V$. Let
$V_{\mbox{\small e}}$ be determined by $\rho_2$ via Eq.~(\ref{5022}) and (\ref{7529}). Let
$\hat{\rho}$ be the 1-density determined by $\rho_2$ by integration.  We want to know if, given
such maps: $\rho_2\mapsto V_{\mbox{\small e}}$ and $\rho_2\mapsto\hat{\rho}$, is there an
eigenfunction $\mathcal{I}$ of Eq.~(\ref{4740}), such that $(\rho_2\mapsto\hat{\rho} =
\mathcal{I}\mathcal{I}^*)$.

In classical mechanics, once the forces are known, the momentum is determined by Newton's second
law.  Comparing Euler Eq.~(\ref{euler0}) with quantum forces~(\ref{2042}), we can infer the
equations of motion for the first order, reduced, fluid  momentum variables
$\hat{\rho}_m\hat{\bu}$ and $\hat{\rho}_m\hat{\bv}$. We do this now by hypothesis, where the
hypothesis is for the most general case, without necessarily the orthogonal condition.
This is followed by a definition for quantum flow mixture.
\begin{hypothesis}[Quantum flow mixture] \zlabel{2224}
Let $\mathcal{F}_R(\hat{\rho},\hat{\Pu},\hat{\bu},\hat{\bv})$ be the reduced flow of quantum
flow $\mathcal{F}_Q(\rho,\Pu,\bu,\bv)$ of an electronic system. Let the wavefunction of
$\mathcal{F}_Q$ satisfy the space time Schr\"odinger equation, and let $\rho_2$ be the pair
density of $\mathcal{F}_Q$. Let $(V_{\mbox{\small e}} = V_{\mbox{\small
    e}}(\qcharge)\hs{0.1ex})$ be given by Eq,~(\ref{5022}), where, it follows from
\begin{equation} 
  \qcharge(\mrr_1,\mrr_2) \defi
  \fc{\rho_2(\mrr_1,\mrr_2)}{\hat{\rho}(\mrr_1)}\quad\text{and}\quad \hat{\rho}(\mrr_1) =
  \fc{2}{n-1}\int_2\rho_2(\mrr_1,\mrr_2),
\end{equation}
that $(\qcharge = \qcharge(\rho_2)\hs{0.1ex})$.  The reduced flow $\hat{\mathcal{F}}_R$
satisfies the first order, reduced Euler-equation:
\begin{equation} \hs{-5ex}\label{1898d} 
\hat{\rho}_m\partial\hat{\bv} + \fc12\hat{\rho}_m\nabla\left(\hat{u}^2 + \hat{v}^2\right)
+ \nabla\cdot(\hat{\rho}_m\hat{\bu})\hat{\bu} + \nabla\hat{\Pu}
+ \hat{\rho}\nabla (V + V_{\mbox{\small e}}) = \hat{\rho}\vec{\mathcal{E}},
\end{equation}
such that $(\vec{\mathcal{E}} = \vec{\mathcal{E}}(\rho_2)\hspace{0.1ex})$ and $(\hat{\rho} =
\hat{\rho}(\rho_2)\hspace{0.1ex})$.  Also the velocity vector $\hat{\bv}$ is irrotational,
permitting the following definition for momentum potential $\hat{S}$ and 1-body wavefunction~$\mathcal{I}$:
\begin{equation}\zlabel{1945}
  \nabla\hat{S} \defi m\hat{\bv}\imply\mathcal{I}(\hat{\rho},\hat{S}) \defi \sqrt{\hat{\rho}}e^{i\hat{S}/\hbar},
\end{equation}
where $\mathcal{I}$, normalized to $n$, is called the orbital of state.

Taking Hypothesis~\ref{1088}, it follows from Theorems~(\ref{theorem00}) and (\ref{5293}) that the reduced
continuity Eq.~(\ref{2132}) and Euler equation (\ref{1898d}) imply
\begin{equation} \zlabel{redeneq} 
  \hat{\ES} \hat{\rho} = \fc12\hat{\rho}_m\hat{u}^2 + \fc12\hat{\rho}_m\hat{v}^2 
+ \hat{\Pu} + (V + V_{\mbox{\small e}})\hat{\rho},
\end{equation}
and the space time Schr\"odinger equation:
\begin{equation} \hs{3ex}\zlabel{4850}
  i\hbar\partial\mathcal{I} = -\fc{\hbar^2}{2m}\nabla^2\hspace{0,15ex}\mathcal{I}  + (V + V_{\mbox{\small e}})\hspace{0.25ex}\mathcal{I},
\end{equation}
becoming Eq.~(\ref{4740}) for stationary states.
\end{hypothesis}

\begin{flushright}\fbox{\phantom{\rule{0.5ex}{0.5ex}}}\end{flushright}
\begin{definition}[Quantum flow mixture] \label{2822}
The flow mixture $\mathcal{F}_M(\mathcal{I},\hat{\Pu},\hat{\bu},\hat{\bv})$ is the reduced
quantum flow that satisfies Eq.~(\ref{1898d}) and (\ref{1945}) with orbital of state
$(\mathcal{I} = \mathcal{I}(\hat{\rho},\hat{S})\hspace{0.1ex})$, such that the pair density
$\rho_2$, from the corresponding quantum flow, determines the 1-density $\hat{\rho}$ and the
quantum Coulomb's force $\hat{\mathcal{E}}$ of Euler Eq.~(\ref{1898d}).
\end{definition}
The kinetic energy from Eq.~(\ref{redeneq}) is calculated using the strong, first order
reduced velocity variables given by Def.~(\ref{7724}). This kinetic energy differs from the
sum of the strong, first order reduced kinetic energy variables:
\begin{equation*} 
\intp\; \rho_m u^2 + \rho_m v^2.
\end{equation*}
Given a choice between the two kinetic energies, the strong variables of the velocity is a
better fit for Euler Eq.~(\ref{1898d}), because equations of motion involves the velocity, not
kinetic energy.

The next definition gives a modification of the Dirac notation that explicitly distinguishes
between operators and their defining kernel, useful for denoting density matrices, that we call
dentrices.
\begin{definition}[Explicit Dirac Notation]\mbox{}
Let $\al\rangle \defi \al:\mathbb{R}^{3n}\rightarrow \mathbb{C}$ and $\langle\beta \defi
\be^*:\mathbb{R}^{3n}\rightarrow \mathbb{C}$.  The kernel of the integral operator
$|\al\rangle\langle\be|$, written in Dirac notation, and involving given functions $\al$ and
$\be$, is denoted $\al\rangle\langle\be$, in other words, we have
$(\hs{0.25ex}\hs{0.1ex}\al\rangle\langle\be\hs{0.1ex}(\mrr,\mrrp) =
\al(\mrr)\be^*(\mrrp)\hs{0.1ex})$.  The maps $\langle\al|$ and $\be\rangle\langle\al|$ are the
functional and operator with values $\langle\al|\gamma\rangle$ and
$\al\rangle\langle\al|\gamma\rangle$ at the point $\gamma\in L^2$, respectively.
\end{definition}

Since the function $(\hs{0.1ex}\hat{\rho}_1 \defi
\mathcal{I}\hs{0.25ex}\rangle\langle\hs{0.25ex}\mathcal{I}\hs{0.15ex})$ is given as a complex
product of a wavefunction $\mathcal{I}$ of a noninteracting state, it cannot be the 1-dentrix
$\rho_1$ of a wavefunction $\Psi$ of an interacting $n$-body state
\cite{Parr:89,Dreizler}. However, this difference is not relevant in the evaluation of the
truth value of Schr\"odinger Eq.~(\ref{4850}).
To investigate this kinetic energy issue further, consider the well known energy functional of
the 2-dentrix for stationary states \cite{Parr:89,Cioslowski}:
%
\begin{equation} \hs{-8ex}\label{4750}
  \ES \, =\,  n\intf{1}  -\fc{\hbar^2}{2m}\hat{\nabla}^2\rho_1\;
  + V\hat{\rho}\hs{1.1ex} + n\intf{1}\intf{2}\, \fc{1}{\mrr_{12}}\,\rho_2;
  \qquad \mrr_{12}(\mrr,\mrrp) = \vert\mrr-\mrrp\vert, 
\end{equation}
\begin{equation*} 
  \hs{-10ex}\left[\hat{\nabla}^2\rho_1\right](\mrr) = \left.\nabla_\mrr^2\rho_1(\mrr,\mrrp)\right\vert_{\mrrp = \mrr},
\quad\rho_1(\mrr,\mrrp) = \intp \Psi(\mrr,\mrr_2,\cdots,\,\mrr_n)\Psi^*(\mrrp,\mrr_2,\cdots,\,\mrr_n).
\end{equation*}
%
where the $n$ factors appear because of our normalization choice~(\ref{7529}) for $\rho_2$.
%
This equation is derived from the expectation value of the Hamiltonian $\hat{H}$. In order to
express the kinetic energy in the integral form above, it is necessary to introduce an
additional variable that endows the 1-dentrix $\rho_1$ with domain $\mathbb{R}^{3n}\times
\mathbb{R}^{3(n+1)}$. Such an extended domain, giving an unphysical function, does not appear
in the kinetic energy identity~(\ref{lapl}). Therefore, we cannot end up with
$\mathcal{I}\hs{0.25ex}\rangle\langle\hs{0.25ex}\mathcal{I}\hs{0.15ex}$ being the 1-dentrix
$\rho_1$ of the wavefunction $\Psi$ of the corresponding quantum flow.  (Note that $\mrr_{12}$,
defined above, is a function, not the value of a function at ($\mrr_1,\mrr_2$).)

\subsection{Quantum Orbital Flow}

In this subsection, we describe quantum orbital flow, obtained by taking force conservation
Hypothesis~\ref{1088}, where $n$-body systems are treated, explicitly, as having $n$
bodies.  We take the $n$ body external potential as a sum over a one body potentials,
with value $(W(\mrr_1) + \cdots W(\mrr_n)\hs{0.1ex})$, where $(W = V+V_{\mbox{\small e}})$ is the
potential of the reduced quantum forces~(\ref{2042}).

Suppose we replace the Coulomb potential $(W_{12}(\mrr,\mrrp) =1\div|\mrr-\mrrp|)$ by a
modification that is body dependent $(W_i(\mrr_i,\mrrp) = \al_i|\mrr_i-\mrrp|^{-1})$;
$(\al_i\in\mathbb{R})$ and $(i=1,\cdots,n)$. This symmetry breakage gives $n$ reduced equations
of motions, replacing the single Eq.~(\ref{1502r}), where the $i$th equation is obtained by
integration over $\mrr_1,\cdots,\mrr_{i-1},\mrr_{i+1},\cdots,\mrr_n$, and summation over the
corresponding spin variables. Under a hypothesis that is similar to Hypothesis~\ref{2224}, we
obtain $n$ Schr\"odinger equations, replacing Eq.~(\ref{4850}), where we normalize each
orbital $\mathcal{I}_i$ to unity. It is reasonable to expect that a solution set of $n$
linearly independent orbitals exist for the $n$ Schr\"odinger equations. If we take the limit
$(\al_i\rightarrow 1)$, we would then expect to have $n$ orthonormal orbitals satisfying the
$n$ Schr\"odinger equations. Hence, encoded in the orbital of state $\mathcal{I}$ must be $n$
linearly independent 1-body orbitals that can be extracted.

In Hartree--Fock theory for fermions, antisymmetry of the wavefunction with respect to
interchange of bodies is imposed on the wavefunction by using a Slater determinant, giving $n$
orthogonal orbitals for a solution set, instead of one occupied orbital for all bodies.
Similarly, we aim to replace Eq.~(\ref{4850}) with $n$ orbital equations. However, the
variational approach used by the Hartree--Fock method does not work for our
situation. Therefore, we must satisfy this need by hypothesis. Even though the hypothesis is as
general as possible, without modifications, it might only be suitable for states that can be
represented by a single Slater determinant: closed shell singlet and spin polarized open shells
on top of a closed shell.

\begin{hypothesis}[Quantum Orbital Flow] \zlabel{2234}
  Let the density $\hat{\rho}$ be normalized to the usual $n$, given by Def.~(\ref{2914}.3)
  with the replacement $1\rightarrow n$.
There exists a set of linearly independent orbitals $\{\psi_1,\psi_2,\cdots,\psi_n\}$; $(\psi_i
= \rho_ie^{iS_i/\hbar})$, that satisfy the following sequence of $n$ Euler equations:
\begin{equation} \hs{-12.5ex} \label{3823a}
\bar{\rho}_i\partial\bv_i + \fc12\bar{\rho}_i\nabla\left(u_i^2 + v_i^2\right) + \nabla\Pu_i +
(\nabla\cdot\bar{\rho}_i\bu_i)\bu_i + \rho_i\nabla (V + V_{\mbox{\small e}}) = \rho_i\vec{\mathcal{E}};
\quad \bar{\rho}_i = m\rho_i,
\end{equation}
\begin{equation} \label{2132r}
  \lefteqn{\hs{-13.5ex}and} \hs{3ex} 
  \partial\rho_i + \nabla\cdot(\rho_i\bv_i) = \mathbf{0}
   \quad (\imply \partial\hat{\rho} + \nabla\cdot(\hat{\rho}\hat{\bv})\room), 
\end{equation}
involving orbital densities $(\rho_i \defi \psi_i\psi_i^*)$, such that $(\vec{\mathcal{E}} =
\vec{\mathcal{E}}(\rho_2)\hspace{0.1ex})$, $(\hat{\rho} = \hat{\rho}(\rho_2)\hspace{0.1ex})$,
$(\hat{\rho} = \rho_1 + \cdots + \rho_n)$, and the set $\{S_1,\cdots,S_n\}$ satisfies the
weighted average equation for $\hat{S}$:
\begin{equation} \hs{-7ex}\label{5857}
\sum_i^n\hat{\rho}_i\nabla S_i \defi \nabla\hat{S};\quad \hat{\rho}_i
\defi \fc{\rho_i}{\hat{\rho}} \defi \fc{\psi_i\psi_i^*}{\hat{\rho}} 
\imply \hat{\rho}\hat{\bv} = \sum_i\rho_i\bv_i.
\end{equation}
Also under hypothesis~(\ref{1088}), giving $(\vec{\mathcal{E}} = -\nabla V_{\mbox{\small e}})$,
it follows from Theorems~(\ref{theorem00}) and (\ref{5293}) that the Euler equation
(\ref{3823a}) and continuity Eq.~(\ref{2132r}) imply corresponding sequences of Schr\"odinger
and energy equations.
\end{hypothesis}
\begin{flushright}\fbox{\phantom{\rule{0.5ex}{0.5ex}}}\end{flushright}
Next we take hypotheses~\ref{1088}, \ref{2224} and \ref{2234} and use Theorems~\ref{theorem00} and
\ref{5293} to define quantum orbital flow. We also use the runner notation for the special case
where nonrunners do not appear.

\begin{definition}[Quantum orbital flow \zlabel{orbflow}] 

Consider the following sequence of $n$ orbital field equations:
  \begin{equation} \hs{-4ex}\zlabel{5855}
  i\hbar\partial\psi  = -\fc{\hbar^2}{2m}\nabla^2\psi  + (V + V_{\mbox{\small e}})\psi; 
\quad \hat{\rho} = \psi\psi^*, \;\;\text{where}\;\; \psi\in\mathbf{\Sigma},
\end{equation}
reducing to the following equation sequence for stationary states:
\begin{equation} \hs{-3ex}\zlabel{5855b}
  -\fc{\hbar^2}{2m}\nabla^2\psi  + (V + V_{\mbox{\small e}})\psi = \varepsilon\psi,
  \qquad  \varepsilon\psi\in\mathbf{\Sigma}.
\end{equation}
Let $\mathcal{F}_Q$ be a quantum flow, whose wavefunction satisfies the space time 
Schr\"odinger equation. We take the hypothesis that $\hat{\bv}$ of $\mathcal{F}_Q$ is
irrotational, giving definition $(m\hat{\bv}\defi \nabla\hat{S})$.  The corresponding $n$
component orbital flow $\hat{\mathcal{F}}_{O}(\{\psi_i,\Pu_i,\bu_i,\bv_i\})$ is defined by a
set of linearly independent orbitals of state $\{\psi_1,\cdots,\psi_n\}$ that satisfy the
sequence of $n$ Schr\"odinger Eqs~(\ref{5855}), and the elements of the set $\{S_1,\cdots, S_n\}$
satisfy Eq.~(\ref{5857}), where
\begin{equation}\hs{-6ex}\label{2820}
  \psi_i \defi \pm\sqrt{\rho_i} e^{iS_i/\hbar}, \quad
  \psi_i:\mathbb{R}^3\rightarrow\mathbb{C}, \quad
  \rho_i \defi \psi_i\psi_i^*, \quad \int_1 \psi_i\psi_i^* = 1.
\end{equation}
Also, the pair density $\rho_2$ of $\mathcal{F}_Q$ determines both the density $\hat{\rho}$ and potential
$V_{\mbox{\small e}}$.  Furthermore, as in the above definitions, we have
\[\hs{-5ex}\bar{\rho}_i\bu_i \defi -\zeta_0\nabla\rho_i, \;\;m\bv_i \defi  \nabla S_i, \qquad
\Pu_i \defi \zeta_0\nabla\cdot(\rho_i\bu_i), \;\; \Pv_i \defi -\zeta_0\nabla\cdot(\rho_i\bv_i),\]
\[\hs{-5ex}\varepsilon_i = -\partial S_i, \quad\text{and}\quad \varphi_i = -\partial\theta_i.\]
In addition, according to Theorems~(\ref{theorem00}) and (\ref{5293}), the orbital flow
$\hat{\mathcal{F}}_O$ satisfy the following two $n$ equations sequences, where all terms run:
\begin{equation} \hs{-2ex} \zlabel{3822}
\hs{-7ex}\fc12\rho_m u^2 + \fc12\rho_m v^2 + \Pu + (V + V_{\mbox{\small e}})\rho = \varepsilon\rho, \quad
\partial\rho + \nabla\cdot(\rho\bv) = \zero,  \quad\text{and} 
\end{equation}
\begin{equation} \hs{-5ex} \zlabel{3823}
\rho_m\partial\bv + \fc12\rho_m\nabla\left(u^2 + v^2\right)
+ \nabla\Pu + (\nabla\cdot\rho\bu)\bu + \rho\nabla(V + V_{\mbox{\small e}}) = \mathbf{0}.
\end{equation}
\end{definition}
Next we show that the orbital and mixed quantum flows are related. 
\begin{theorem} Consider the following equations:

\noindent
{\bf 1)} The $n$-body Schr\"odinger equation
\begin{equation} \hs{-5ex}\zlabel{4833}
i\hbar\partial\hat{\mathcal{I}} = \hat{H}\hat{\mathcal{I}}; \qquad
\hat{H} = -\fc{\hbar^2}{2m} \sum_{i=1}^n \nabla_i^2 + \sum_{i=1}^n \left[V + V_{\mbox{\small e}}\right](\mrr_i),
\end{equation}
where $(\hat{\mathcal{I}}:\mathbb{R}^{3n}\rightarrow \mathbb{C})$ is the Slater determinant
$\hat{\mathcal{I}}$ constructed from $\{\psi_1,\cdots,\psi_n\}$. 

\noindent
{\bf 2)} The sequence of $n$ Schr\"odinger Eqs.~(\ref{5855}) of the orbitals of state $\{\psi_1,\cdots,\psi_n\}$, in the form
\begin{equation*} 
\hs{-3ex}i\hbar\partial\psi\rangle\langle\psi = 
-\fc{\hbar^2}{2m}\nabla^2\psi\rangle\langle\psi  + (V + V_{\mbox{\small e}})\hspace{0.25ex}\psi\rangle\langle\psi.
\end{equation*}
\noindent
{\bf 3)} The single Schr\"odinger Eq.~(\ref{4850}), in the form
\[\hs{-3ex}i\hbar\partial\mathcal{I}\rangle\langle\mathcal{I} = -\fc{\hbar^2}{2m}\nabla^2\hspace{0,15ex}\mathcal{I}\rangle\langle\mathcal{I}
+ (V + V_{\mbox{\small e}})\hspace{0.35ex}\mathcal{I}\rangle\langle\mathcal{I},\]
for the special case, such that
\begin{equation} \hs{-3ex}\zlabel{1720}
  \mathcal{I}\rangle\langle\mathcal{I} = \psi_1\rangle\langle\psi_1 + \cdots + \psi_n\rangle\langle\psi_n;
  \qquad \intf{\,}\psi_i\psi_j^* = \delta_{ij}.
\end{equation}
    {\bf 1} $\imply$ {\bf 2} $\imply$ {\bf3}. Also, the function
    $\mathcal{I}\rangle\langle\mathcal{I}:\mathbb{R}^3\times\mathbb{R}^3\longrightarrow\mathbb{C}$
is the 1-dentrix from the Slater determinant $\hat{\mathcal{I}}$ constructed from the orthonormal orbital set 
$\{\psi_1,\cdots,\psi_n\}$.
\end{theorem}
\begin{proof}
Since Eq.~(\ref{4833}) is separable, there must be a set of $n$ orthonormal orbitals giving
Slater determinant $\hat{\mathcal{I}}$ that satisfies Eq.~(\ref{4833}), such that the orbital
set satisfies the equation sequence~(\ref{5855}).  Summing Eq.~(\ref{5855})$\times\psi^*$ over
the $n$ equations of the sequence, we obtain Eqs.~(\ref{4850}), in the form above, where the
1-dentrix satisfies Eq.~(\ref{1720}).  Since there is a one-to-one correspondence between
determinantal states and 1-dentrixes given by the sum from Eq.~(\ref{1720}), the function
$\mathcal{I}\rangle\langle\mathcal{I}$ is the 1-dentrix from the Slater determinant
$\hat{\mathcal{I}}$ formed from $\{\psi_1,\cdots,\psi_n\}$.
\end{proof}
\begin{flushright}\fbox{\phantom{\rule{0.5ex}{0.5ex}}}\end{flushright}
\subsection{Obtaining solutions of \sch Eq.~(\ref{5855b}) and energy Eq.~(\ref{4750})}

In the discussions that follows, we restrict the discussion to stationary states. Furthermore,
we assume that Hypothesis~\ref{2234} holds, or provides a good approximation, for a range of
states, especially states that can be described by a single Slater determinant.

One approach to obtain an approximate solution sets $\{\psi_1,\cdots,\psi_n\}$ for the spatial
Schr\"odinger Eq.~(\ref{5855b}), where $V_{\mbox{\small e}}$ is determined by the 2-density, is
to use an approximate map $\hat{\rho}\mapsto\rho_2$ of the 1-density, where $\hat{\rho}$ is
equal to the orbital density sum $|\psi|^2$; $|\psi|^2\in\mathbf{\Sigma}$. Given such a map as a composite
function or functional, the Schr\"odinger Eqs.~(\ref{5855b}) can be solved using the self
consistent field method.  The energy can then be calculated using Eq.~(\ref{4750}), where, as
in the approximation of the kinetic energy in Kohm--Sham density functional theory, 
the kinetic energy is approximated by the one obtained from the 1-dentrix
$\mathcal{I}\rangle\langle\mathcal{I}$ of Eq.~(\ref{1720}) of the noninteracting state
$\mathcal{I}$.

Approximations of $\rho_2$ can come from functionals of the electron interaction energy
$E_{\text{e}}$---the last term in the rhs of energy Eq.~(\ref{4750}.1)---given as a functional of
the density, where the $\mrr_{12}$ function is replaced by variable $W_{12}$:
\[\left[E_{\text{\hspace{0.075ex}e}}(W_{12})\right] (\hat{\rho}) \defi \left[\intf{1}\intf{2}\, W_{12}\,\rho_2\right](\hat{\rho}).\]
The 2-density $\rho_2$ is then given by a functional derivative:
\[\hs{-10ex}\rho_2 = \fc{\delta E_{\text{e}}}{\delta \mathbf{r}_{12}^{-1}},
\qquad \fc{\delta E_{\text{e}}}{\delta \mathbf{r}_{12}^{-1}}
\defi \left[\fc{\delta E_{\text{e}}}{\delta W_{12}}\right]\hs{-0.6ex}\Bigl(\mathbf{r}_{12}^{-1}\Bigl);
\quad E_{\text{e}}[\rho] = \fc12\intf{1} \intf{2} \fc{1}{\mrr_{12}}\hat{\rho}(\mrr_1)\hat{\rho}(\mrr_2) + E_{\text{xc}}[\rho].\]
There is a disconnect between most of the functional used in Kohm--Sham density functional
theory, since the theory does not provide for the construction of a useful way to obtain
approximation. The same functionals are applicable other approaches. For example, the functionals
can be used for reference state one-particle density-matrix theory, involving Brueckner
orbitals \cite{Finley-brueckner,Finley-functionals}.  Similarly, the same exchange-correlation
energy functionals $E_{\text{xc}}$ used in Kohm--Sham density functional theory can be used to obtain an 
approximate for $(\hspace{0.05ex}\rho_2 = \rho_2[\hat{\room\rho}\room]\hspace{0.1ex})$.



The 2-body function $\mathcal{I}\rangle\langle\mathcal{I}$ from a determinantal state
$\hat{\mathcal{I}}$, given by Eq.~(\ref{1720}), cannot be the 1-dentrix of an $n$-body
interacting state.
We now consider if it is at all possible that the formalism can be modified so that the $n$
orbital expression~(\ref{1720}) can deliver an interacting state 1-dentrix.
First note that $n$ linearly independent orbitals can give the 1-dentrix of an interacting
state, if the orbitals are nonorthogonal. This can be seen by partitioning natural orbitals
into $n$ separate sums of orbitals, where some orbitals are divided into two by using ($\psi =
c\psi + (1-c)\psi$), such that each of the $n$ orbitals are normalized. Since mixed states of
quantum mechanics are, in general, nonorthogonal, the proper description for the embedded
states considered here might be nonorthogonal orbitals.
Unfortunately, nonorthogonal orbitals cannot provide a solution to equation
sequence~(\ref{5855b}), because the eigenstates of a Hermitian operator are orthogonal.
Therefore, the only option is to generalize the Hamiltonian to be nonhermitian, or use a
generalized Euler equation. The use of nonhermitian operators provide opportunities for
interpretations arising from the possibility of linearly dependent-eigenfunctions and complex
valued-eigenvalues, given that the orbital states can be considered part of a mixed state.

\subsection{Reduced Cross Flows \zlabel{7928}} 

In this subsection, we investigate an approach, based on smooth, cross flows, that permit a way
to get around the problem of a complete reduction of equation of motion~(\ref{1502r}).  Such
flows do not have the spoiler integral in their equation of motion. From Theorem~\ref{2629},
these flows have a constant density on each streamline, but the density depends on the
streamline.  The next theorem demonstrates that if there exists a constant density cross flow
$\mathcal{F}_\times(\Pu_1,\rho_m,\mmu_1,\bv_1)$, corresponding to smooth flow
$\mathcal{F}(\Pu_1,\rho_m,\bu_1,\bv_1)$, that satisfies the equation of motion
\begin{equation} \zlabel{1895h}
\rho_m\mathcal{D}\oomega + \nabla\Pu_1 + \rho\nabla V_1 + \rho\nabla W_1 = \rho\vec{F},
\end{equation}
for any non work force $\vec{F}$, then flows $\mathcal{F}_\times(\Pu,\rho_m,\mmu,\bv)$ and
$\mathcal{F}(\Pu,\rho_m,\bu,\bv)$ satisfy the same Euler equation, except for nowork
forces. Conditions for the energy to be uniform and the satisfaction of the spatial \sch
equation are also given. Also, if the cross flow is solenoidal, mass is conserved.

\begin{theorem}[Cross Flow Euler Equation Reduction] \zlabel{5945}

Let $\mathcal{C}_1$ be the set of all cross flows
$\mathcal{F}_\times(\Pu_1,\rho_m,\mmu_1,\bv_1)$ of parameterized, smooth flow
$\mathcal{F}(\Pu_1,\rho_m,\bu_1,\bv_1)$ that satisfy Euler cross flow Eq.~(\ref{1895h}),
corresponding to Euler Eq.~(\ref{1502}), with an additional nowork force~$\vec{F}$,
where $(\oomega = \mmu_1 + \bv_1)$, and $\vec{F}$ can depend on the cross flow.  In addition,
the $n$-body wavefunction $(\Psi = \Psi(S,\rho))$ of $\mathcal{F}$ is parameterized by the
point set $(\mq_1^\pr)\in\mathbb{R}^{n-1}$, and $\Psi$ satisfies the TISE
(\ref{1500}). \vspace{1ex}

\noindent
 {\bf 1)} If $\mathcal{C}_1$ is nonempty, each $\mathcal{F}_\times(\Pu_1,\rho_m,\mmu_1,\bv_1)
 \in\mathcal{C}_1$ satisfies energy equation~(\ref{eneq-cross}), and their corresponding
 reduced flows $\mathcal{F}_\times(\hat{\Pu},\hat{\rho}_m,\hat{\mmu},\hat{\bv})$ satisfy the
 following reduced cross flow Euler equation:
\begin{equation} \zlabel{1898b}
  \mathcal{D}(\hat{\rho}_m\hat{\oomega}) + \nabla\hat{\Pu} + \hat{\rho}\nabla V_1 =
  \hat{\rho}\vec{\mathcal{E}} + \hat{\rho}\vec{f}; \qquad \hat{\rho}\vec{f} = \intp \rho
  \vec{F},
\end{equation}
where $(\hat{\oomega} = \hat{\mmu} + \hat{\bv})$.
\end{theorem}

Parts {\bf 2} and {\bf 3} below hold under hypothesis~\ref{1088}: $(\vec{\mathcal{E}}\cdot
d\hat{\mathbf{s}} = -\nabla V_{\mbox{\small e}}\cdot d\hat{\mathbf{s}})$, giving from
Eq.~(\ref{1898b}),
\begin{equation} \zlabel{1898c} 
\mathcal{D}(\hat{\rho}_m\hat{\oomega}) + \nabla\hat{\Pu}_1 + \hat{\rho}\nabla V_1  + \hat{\rho}\nabla V_{\mbox{\small e}}
  = \hat{\rho}\vec{f}.
\end{equation}
\vspace{1ex}

\noindent
{\bf 2)} If $\mathcal{C}_1$ is nonempty, then each $\mathcal{F}_\times(\hat{\mmu},\hat{\bv})$
of $\mathcal{F}_\times(\mmu_1,\bv_1)\in \mathcal{C}_1$ satisfies the following energy equation:
\begin{equation} \label{eneq-crossc} 
  \hat{\ES}\hat{\rho} = \fc12\hat{\rho}_m\hat{\mu}^2 + \fc12\hat{\rho}_m\hat{v}^2 + \hat{\Pu}
  + (V_1 + V_{\mbox{\small e}})\hat{\rho}.
\end{equation}
The continuity equation $(\hs{0.1ex}\nabla\cdot \hat{\rho}_m\hat{\bv}) = \zero)$ is also satisfied.
\vspace{1ex}

\noindent 
    {\bf 3)} If $(\hat{u}^2 = \hat{\mu}^2)$ then reduced flow
    $\mathcal{F}(\hat{\bu},\hat{\bv})$ also satisfies energy Eq.~(\ref{eneq-crossc}).
    If, in addition, energy $\hat{E}$ of Eq.~(\ref{eneq-crossc}) is uniform, and $\hat{\bv}$ is
    irrotational, giving $(\nabla\mathcal{S} \defi m\hat{\bv})$, then the wavefunction
    $(\mathcal{I} = \sqrt{\hat{\rho}}e^{i\mathcal{S}/\hbar})$ of reduced from
    $\mathcal{F}(\hat{\bu},\hat{\bv})$ satisfies the 1-body \sch Eq.~(\ref{4740}).
%

\noindent
\textbf{\textit{Proof}} \mbox{}

\noindent 
 {\bf 1)} Each $\mathcal{F}_\times(\Pu_1,\rho_m,\mmu_1,\bv_1)\in\mathcal{C}_1$ satisfies energy
 equation~(\ref{eneq-cross}) follows from Theorem~\ref{5084}. Each
 $\mathcal{F}_\times(\hat{\Pu},\hat{\rho}_m,\hat{\mmu},\hat{\bv})$ satisfies the reduced cross
 flow Euler Eq.~(\ref{1898b}), follows from Theorem~(\ref{1965}), equality~(\ref{2042}), for
 the forces, and that $(\rho_m\mathcal{D}\oomega = \mathcal{D}(\rho_m\oomega))$, because,
 according to Theorem~\ref{2629}, smooth cross flows have constant density, and
 $\rho_m\oomega$ can be reduced. \vspace{1ex}

\noindent
{\bf 2} Each $\mathcal{F}_\times(\hat{\mmu},\hat{\bv})$ satisfies reduced energy
equation~(\ref{eneq-crossc}) follows from the direct integration of Euler Eq.~(\ref{1898c})
along the streamline, as is done in Theorem~(\ref{5084}). 
Since the wavefunction $\Psi$ satisfies the TISE, the continuity equation, given by
$(\nabla_1\cdot(\rho_m\bv_1) = \zero)$, is satisfied, implying that $(\nabla
\cdot(\hat{\rho}_m\hat{\bv}) = \zero)$.\vspace{1ex}

\noindent
{\bf 3} If $(\hat{u}^2 = \hat{\mu}^2)$ then reduced flow $\mathcal{F}(\hat{\bu},\hat{\bv})$
also satisfies energy Eq.~(\ref{eneq-crossc}) because the other flow parameters are the same for both reduced flows.
The satisfaction of the \sch equation follows from Theorem~\ref{theorem00}.
\begin{flushright}\fbox{\phantom{\rule{0.5ex}{0.5ex}}}\end{flushright}


\section{Summary and Conclusion}


The work presented here provides a foundation of a formalism for quantum states based on
the cause and effect of Newton's second law.  While the energy and equations of motion
derived here are applicable to models involving fluid flow and particle trajectories, the
best possible description of the nature of states of quantum systems, consistent with these
equations, might be something different. Since a trajectory model has limitations set by
the uncertainty principle involving position and momentum, a model based on a continuum is
preferred over the trajectory of point masses. According to Broer \cite{Broer}, at least
during the 1970s, there was considerable interests in fluid models for this reason.  At
this stage of development, it is useful to develop many interpretations of the equations,
for use in a process of elimination. 

Theorem~\ref{conti} connects different definitions associated with the imaginary part of the
\sch equation in density matrix form.  Theorem~\ref{theorem00} demonstrates that $n$-body
Bernoullian Energy Eq.~(\ref{eneq}.1) and continuity Eq.~(\ref{eneq}.2) are equivalent to the
\sch equation. Theorem~\ref{theorem22b} provides a generalized equation that removes a
nonclassical feature of the kinetic energy formula.  Section~\ref{2215} gives a complete
identification of the terms in the spectrum of the Euler equations of Theorem~\ref{5293},
\ref{5293b}, and \ref{5293c}, and identifies the flow type as compressible, inviscid,
irrotational, and variable mass. The cross flows of Sec.~\ref{4938} demonstrate that it is
possible to assign kinetic energy to a velocity vector that have properties that are closer to
ones seen in systems of classical mechanics, including orbit trajectories with continuous
velocities and conserved mass fluid flow. The standard Lagrangian formulation of classical
mechanics is described in Sec.~\ref{9918}.

Via theorems, and reasonable hypotheses, for orbital flows of $n$-body states, Sec.~\ref{4918}
obtains sets of $n$ Energy~(\ref{3822}), Euler~(\ref{3823}) and 1-body \sch Eqs.~(\ref{5855}).
The \sch equations are generalizations of the Hartree--Fock equations with the Coulomb and
exchange potentials replaced by potential $V_{\mbox{\small e}}$, giving a quantum Coulomb's law
of force. The same set of three equations are obtained for flows, called mixed flows, that are
defined by Quantum Flow Defs.~(\ref{orbflow}).  For mixed flows, instead of $n$ 1-body
equations, there is only one reduced Euler, energy and \sch equation, to represent an $n$-body
state. Section~\ref{7928} demonstrates that certain cross flows, that are investigated in
Sec.~\ref{4938}, can be reduced to the same Euler equation of mixed flows, but without the need
for hypothesis.


\appendix

\section{Bohmian mechanics Eqs.\ and \sch Eq.\ equivalence \zlabel{7825}} 

\begin{theorem}[Bohmian Mechanics Equations] \zlabel{bohmian}
  The Bohmian Mechanics equations
\begin{equation} \zlabel{5200} 
  -\partial S = \fc12mv^2 + Q + U,  \quad\partial\rho + \fc{1}{m}\nabla\cdot(\rho\nabla S) =\mathbf{0}, 
\end{equation}
and the definition of the quantum potential $Q$:
\begin{equation} \zlabel{p3322} 
Q \defi\! -\fc{\hbar^2}{2m}\rho^{-1}R\nabla^2R,\quad Q\in\mathbf{\Sigma},
\end{equation}
are equivalent to equation set~(\ref{eneq})$\times\rho^{-1}$, definitions~(\ref{en2}) for $\ES$
and $\Etheta$, definition~(\ref{4720}) for $\Pu$ and $\Pv$, and the second equation of
(\ref{def77}) for $u^2$. The Bohmian mechanics Eqs.~(\ref{5200}) and (\ref{p3322}) are
also equivalent to the \sch Eq.~(\ref{schrodinger}).

{\rm Note that both approaches use
\begin{equation}\zlabel{5220} \bigl((m\bv)\defi (\nabla S)\bigl) \quad\text{giving}\quad
  \partial\rho + \nabla\cdot(\rho\bv)=\mathbf{0},
\end{equation}
but this is not needed for the proof that follows.}
\end{theorem}
\textbf{\textit{Proof}}
First we prove that the first equations of (\ref{eneq})$\times\rho^{-1}$ and (\ref{5200}) are
equivalent. Comparing these two equations, and using the first definition from (\ref{en2})
($\ES = -\partial S$), the definition above for $(Q)$, and noting that the equations have the
same definition for $(v^2)$, it follows that the two equations are equivalent, if and only if
\begin{equation} \zlabel{5552}
-\fc{\hbar^2}{2m}\rho^{-1}R\nabla^2R = \fc12mu^2 + \Pu\rho^{-1}.
\end{equation}
The identity:
\begin{equation*} 
-\fc{\hbar^2}{2m}\rho^{-1}R\nabla^2R = \fc{\hbar^2}{8m}\rho^{-2}\nabla\rho\cdot\nabla\rho - \fc{\hbar^2}{4m}\rho^{-1}\nabla^2\rho,
\end{equation*}
is easily proved by expanding out $(\nabla^2\rho)$ or by setting $S$ to a constant in (\ref{pexpect}).
(A proof is also give elsewhere \cite{Finley-Bern}.)
Hence, we only need to prove that the rhs of this identity is equivalent to the rhs of (\ref{5552}).
Substituting definition (\ref{def77}) for $u^2$ and (\ref{4725}) for $\Pu$, the above identity becomes~(\ref{5552}), 
completing the first part of the proof.

To complete the proof for the first statement, we show that the second equations of
(\ref{eneq}) and (\ref{5200}) are equivalent:
\[ \hspace{-14ex} 
\text{Substituting  definitions (\ref{en2}):} \quad  \Etheta\rho \defi \fc{\hbar}{2}\partial\rho,\quad \text{and} \quad (\ref{4725}):
\quad\Pv \defim -\fc{\hbar}{2m}\nabla\cdot(\rho\nabla S),\quad\text{into} 
\]
definition (\ref{eneq}): ($\Etheta\rho = \Pv$), gives
the second equation of (\ref{5200}), if 
the equation is multiplied by $2/\hbar$ and solved for the zero function.

For the second statement, the equivalence of the Bohmian mechanics Eqs.~(\ref{5200}) and
(\ref{p3322}) and the \sch Eq.~(\ref{schrodinger}) follows immediately from Theorem~\ref{theorem00}.
\begin{flushright}\fbox{\phantom{\rule{0.5ex}{0.5ex}}}\end{flushright}

\bibliography{ref}

\begin{thebibliography}{10}

\bibitem{Finley-Bern}
J.~P. Finley.
\newblock {\em J. Phys. Commun.}, 6:04002, 2022.

\bibitem{Bohm:52a}
D.~Bohm.
\newblock {\em Phys. Rev.}, 85:166, 1952.

\bibitem{Bohm:52b}
D.~Bohm.
\newblock {\em Phys. Rev.}, 85:180, 1952.

\bibitem{Takabayasi}
T.~Takabayasi.
\newblock {\em Prog. Theor. Phys.}, 11:341, 1954.

\bibitem{B2}
A.~S. Sanz.
\newblock {\em Front. Phys}, 14:11301, 2019.

\bibitem{B4}
A.~S. Sanz and S.~Miret-Arte\'s.
\newblock {\em A Trajectory Description of Quantum Processes. II. Applications,
  vol. 831 of Lecture Notes in Physics}.
\newblock Springer, Berlin, 2014.

\bibitem{B5}
D.~D\"urr and S.~Teufel.
\newblock {\em Bohmian Mechanics}.
\newblock Springer, Heidelberg, 2009.

\bibitem{B6}
R.~E. Wyatt.
\newblock {\em Quantum Dynamics with Trajectories}.
\newblock Springer, New York, 2005.

\bibitem{B7}
{P. K. Chattaraj, ed.}
\newblock {\em Quantum Trajectories}.
\newblock CRC Taylor and Francis, New York, 2010.

\bibitem{B8}
K.~H. Hughes and {G. Parlant, Eds.}
\newblock {\em Quantum Trajectories}.
\newblock CCP6, Daresbury, 2011.

\bibitem{B9}
X.~Oriols and {J. Mompart, Eds.}
\newblock {\em Applied Bohmian Mechanics: From Nanoscale Systems to Cosmology}.
\newblock Pan Standford Publishing, Singapore, 2012.

\bibitem{Jung}
K.~Jung.
\newblock {\em J. Phys.: Conf. Ser.}, page 442, 2013.

\bibitem{Renziehausenb}
K.~Renziehausen and I.~Barth.
\newblock {\em Found. Phys.}, 50:772, 2020.

\bibitem{Madelung:26}
E.~Madelung.
\newblock {\em Naturwissenschaften}, 14:1004, 1926.

\bibitem{Madelung:27}
E.~Madelung.
\newblock {\em Z. Phys.}, 40:322, 1927.

\bibitem{Heifetz}
E.~Heifetz and E.~Cohen.
\newblock {\em Found. Phys.}, 45:1514, 2015.

\bibitem{Heifetz2}
E.~Heifetz et~al.
\newblock {\it On entropy production in the Madelung fluid and the role of
  Bohm’s potential in classical diffusion}, the manuscript is available from
  http://arXiv.org/abs/1509.01265v2, 2015.

\bibitem{Sorokin}
A.~L. Sorokin.
\newblock {\em Dokl. Phys}, 46:576, 2001.

\bibitem{Broadbridge}
P.~Broadbridge.
\newblock {\em Symmetry}, 7:1803, 2015.

\bibitem{Schonberg}
M.~Sch\"onberg.
\newblock {\em Il Nuovo Cimento}, 12:103, 1954.

\bibitem{Caliari}
M~Caliari et~al.
\newblock {\em Il Nuovo Cimento}, 6:69, 2004.

\bibitem{Jamali}
A.~Jamali.
\newblock {\it Nonlinear Generalisation of Quantum Mechanics}, the manuscript
  is available from http://preprints.org/manuscript/202108.0525/v1, 2021.

\bibitem{Waegell}
M.~Waegell.
\newblock {\em Quantum Stud.: Math. Found.}, pages 1--18, 2024.

\bibitem{Tsekov}
R.~Tsekov.
\newblock {\em Ann. Univ. Sofia, Fac. Phys. Special Edition, [arXiv
  0904.0723]}, page 112, 2012.

\bibitem{Vadasz}
P.~Vadasz.
\newblock {\em Fluids}, 1:18, 2016.

\bibitem{6}
T.~Takabayasi.
\newblock {\em Prog. Theor. Phys.}, 8:143, 1952.

\bibitem{7}
T.~Takabayasi.
\newblock {\em Prog. Theor. Phys.}, 14:283, 1955.

\bibitem{8}
D.~Bohm, R.~Schiller, and J.~Tiomno.
\newblock {\em Nuovo Cimento 1}, 1:48, 1955.

\bibitem{9}
{R. J. Harvey}.
\newblock {\em Phys. Rev.}, 152:1115, 1966.

\bibitem{10}
I.~Bialynicki-Birula and Z.~Bialynicka-Birula.
\newblock {\em Phys. Rev. D}, 3:2410, 1971.

\bibitem{11}
N.~Rosen.
\newblock {\em Nuovo Cimento B}, 19:90, 1974.

\bibitem{12}
B.~M. Deb and A.~S. Bamzai.
\newblock {\em Mol. Phys}, 35:1349, 1978.

\bibitem{14}
T.~C. Wallstrom.
\newblock {\em Phys. Lett. A}, 184:229, 1994.

\bibitem{15}
E.~Recami and G.~Salesi.
\newblock {\em Phys. Rev. A}, 57:98, 1998.

\bibitem{16}
R.~E. Wyatt.
\newblock {\em J. Chem. Phys.}, 117:9568, 2002.

\bibitem{17}
C.~L. Lopreore and R.~E. Wyatt.
\newblock {\em Phys. Rev. Lett.}, 82:5190, 1999.

\bibitem{18}
C.~L. Lopreore and R.~E. Wyatt.
\newblock {\em Chem. Phys. Lett.}, 325:73, 2000.

\bibitem{19}
C.~L. Lopreore and R.~E. Wyatt.
\newblock {\em J. Chem. Phys.}, 116:1228, 2002.

\bibitem{20}
T.~Koide.
\newblock {\em Phys. Rev. C}, 87:034902, 2013.

\bibitem{21}
R.~E. Wyatt, D.~J. Kouri, and D.~K. Hoffman.
\newblock {\em J. Chem. Phys.}, 112:10730, 2000.

\bibitem{22}
R.~E. Wyatt and E.~R. Bittner.
\newblock {\em J. Chem. Phys.}, 113:8898, 2000.

\bibitem{8b}
L.S. Kuzmenkov and S.G.Maksimov.
\newblock {\em Theor. Math. Phys.}, 118:227, 1999.

\bibitem{Renziehausen}
K.~Renziehausen and I.~Barth.
\newblock {\em Prog. Theor. Exp. Phys}, page 013A05, 2018.

\bibitem{14b}
L.S. Kuzmenkov, S.G.Maksimov, and V.~V. Fedoseev.
\newblock {\em Theor. Math. Phys.}, 126:110, 2001.

\bibitem{15b}
P.A. Andreev and L.~S. Kuzmenkov.
\newblock {\em Russ. Phys. J.}, 50:1251, 2007.

\bibitem{16b}
P.~A. Andreev and L.~S. Kuzmenkov.
\newblock {\em Phys. Rev. A}, 78:053624, 2008.

\bibitem{17b}
P.~A. Andreev.
\newblock {\em In: PIERS Proceedings, Moscow, p. 154}, page~19, 2012.

\bibitem{18b}
P.~A. Andreev.
\newblock {\em Int. J. Mod. Phys. B}, 26:1250186, 2012.

\bibitem{22b}
M.I. Trukhanova and P.~A. Andreev.
\newblock {\em Phys. Plasmas}, 22:022128, 2015.

\bibitem{23b}
M.~I. Trukhanova.
\newblock {\em rog. Theor. Exp. Phys}, 2013:111I01, 2013.

\bibitem{Holland:99}
P.~Holland.
\newblock {\em Phys. Rev. A}, 60:4326, 1999.

\bibitem{Colin1}
C.~Colijn and E.R. Vrscay.
\newblock {\em Phys. Lett. A}, 300:334, 2002.

\bibitem{Colin2}
C.~Colijn and E.R. Vrscay.
\newblock {\em J. Phys. A}, 36:4689, 2003.

\bibitem{Finley-Arxiv}
J.~P. Finley.
\newblock {\it The Fluid Dynamics of the One-Body Stationary States of Quantum
  Mechanics with Real Valued Wavefunctions}, the manuscript is available from
  http://arXiv.org/abs/2107.10315, 2021.

\bibitem{Lieb}
E.~H. Lieb.
\newblock {\em Int. J. Quantum Chem.}, 24:243, 1983.

\bibitem{Broer}
{L. J. F. Broer}.
\newblock {\em Physica}, 76:364, 1974.

\bibitem{Salesi}
G.~Salesi.
\newblock {\em Mod. Phys. Lett. A}, 11:1815, 1996.

\bibitem{Hestenes}
D.~Hestenes.
\newblock {\em Found. Phys.}, 23:365, 1993.

\bibitem{Parr:89}
W.~Yang R.~G.~Parr.
\newblock {\em Density Functional Theory of Atoms and Molecules}.
\newblock Oxford University Press, Oxford, 1989.

\bibitem{Cioslowski}
J.~Cioslowski.
\newblock {\em Many-Electron Densities and Reduced Density Matrices}.
\newblock Kluwer, New York, Boston, London, 2000.

\bibitem{Dreizler}
R.~M. Dreizler and {E.K.U. Gross}.
\newblock {\em Density Functional Theory}.
\newblock Springer-Verlag, Berlin, New York, 1990.

\bibitem{Munson}
B.~R. Munson, D.~F. Young, and T.~H. Okiishi.
\newblock {\em Fundamentals of Fluid Dynamics}.
\newblock Wiley, Hoboken, 5th edition, 2006.

\bibitem{Bransden}
B.~H. Bransden and C.~J. Joachain.
\newblock {\em Physics of Atoms and Molecules}.
\newblock Longman, London and New York, 1989.

\bibitem{Shapiro}
A.~H. Shapiro.
\newblock {\em The Dynamics and Thermodynamics of Compressible Fluid Flow},
  volume~1.
\newblock Ronald, New York, 1953.

\bibitem{Kelly}
P.~A. Kelly.
\newblock Foundations of continuum mechanics.
\newblock {\it Foundations of Continuum Mechanics}, Available from
  http://pkel015.connect.amazon.auckland.ac.nz, 2020.

\bibitem{Oldofredi}
A.~Oldofredi.
\newblock {\em Latosensu}, 5(1):77--85, 2018.

\bibitem{Towler}
M.~D.Towler.
\newblock {\em Proc. R. Soc. A}, 468:990--1013, 2011.

\bibitem{Finley-equations}
J.~P. Finley.
\newblock {\it Field Equations of Classical Mechanics for Quantun Mechanics},
  the manuscript is available from http://arXiv.org/abs/2207.04349, 2022.

\bibitem{Finley-brueckner}
J.~P. Finley.
\newblock {\em Phys. Rev. A}, 69:042514, 2004.

\bibitem{Finley-functionals}
J.~P. Finley.
\newblock {\em J. Mol. Phys.}, 102:627--639, 2004.

\end{thebibliography}
\bibliographystyle{unsrt}
\end{document}